\documentclass[a4paper,USenglish,cleveref, autoref, thm-restate]{lipics-v2021}
\hideLIPIcs
\pdfoutput=1

\usepackage{cite}

\usepackage{xcolor}
\usepackage[graphicx]{realboxes}

\nolinenumbers

\title{Computing NP-hard Repetitiveness Measures via MAX-SAT}

\author{Hideo Bannai}{Tokyo Medical and Dental University, Japan}{hdbn.dsc@tmd.ac.jp}{https://orcid.org/0000-0002-6856-5185
 }{Supported by JSPS KAKENHI Grant Number JP20H04141}
\author{Keisuke Goto}{Independent Researcher}{keisukegotou@gmail.com}{https://orcid.org/0000-0001-6964-6182}{}
\author{Masakazu Ishihata}
{NTT Communication Science Laboratories, Japan}
{masakazu.ishihata.ze@hco.ntt.co.jp}
{}
{}
\author{Shunsuke Kanda}{Independent Researcher}{shnsk.knd@gmail.com}{https://orcid.org/0000-0002-5462-122X}{}
\author{Dominik K\"oppl}{Tokyo Medical and Dental University, Japan}{koeppl.dsc@tmd.ac.jp}{https://orcidid.org/0000-0002-8721-4444}{Supported by JSPS KAKENHI Grant Numbers JP21H05847 and JP21K17701.}
\author{Takaaki Nishimoto}{RIKEN Center for Advanced Intelligence Project, Japan}{takaaki.nishimoto@riken.jp}{}{}

\authorrunning{H. Bannai, K. Goto, M. Ishihata, S. Kanda, D. K\"oppl, and T. Nishimoto} 

\Copyright{Hideo Bannai, Keisuke Goto, Masakazu Ishihata, Shunsuke Kanda, Dominik K\"oppl, and Takaaki Nishimoto} 



\category{} 

\relatedversion{} 

\supplement{Our implementation is available at \url{https://github.com/kg86/satcomp}.}

\EventEditors{Shiri Chechik, Gonzalo Navarro, Eva Rotenberg, and Grzegorz Herman}
\EventNoEds{4}
\EventLongTitle{30th Annual European Symposium on Algorithms (ESA 2022)}
\EventShortTitle{ESA 2022}
\EventAcronym{ESA}
\EventYear{2022}
\EventDate{September 5--9, 2022}
\EventLocation{Berlin/Potsdam, Germany}
\EventLogo{}
\SeriesVolume{244}
\ArticleNo{69}

\newcommand{\occ}{\mathit{occ}}
\newcommand{\cover}{\mathit{cover}}
\newcommand*{\teigi}[1]{\emph{#1}}

\newcommand{\fbeg}[1]{p_{#1}}

\newcommand{\match}[1]{\mathit{M}_{#1}}
\newcommand{\rot}[1]{\mathit{root}_{#1}}
\newcommand{\dref}[1]{\mathit{dref}_{#1}}

\newcommand{\refpos}[1]{\mathit{ref}_{#1}}
\usepackage{longtable}
\usepackage{siunitx}

\begin{document}

\maketitle

\begin{abstract}
    Repetitiveness measures reveal profound characteristics of datasets, and give rise to compressed data structures and algorithms working in compressed space.
    Alas, the computation of some of these measures is NP-hard,
    and straight-forward computation is infeasible for datasets of even small sizes.
    Three such measures are the smallest size of a string attractor,
    the smallest size of a bidirectional macro scheme,
    and the smallest size of a straight-line program.
    While a vast variety of implementations for heuristically computing approximations exist,
    exact computation of these measures has received little to no attention.
    In this paper, we present MAX-SAT formulations that provide the first non-trivial implementations
    for exact computation of smallest string attractors,
    smallest bidirectional macro schemes, and
    smallest straight-line programs.
    Computational experiments show that our implementations work for
    texts of length up to a few hundred for straight-line programs and bidirectional macro schemes,
    and texts even over a million for string attractors.
\end{abstract}

\section{Introduction}
Text compression is a fundamental topic in computer science with countless practical applications.
\teigi{Dictionary compression} is a type of text compression where the original input is
transformed into a sequence of elements taken from a dictionary,
where the dictionary is usually constructed in some way from the input.
Due to the advent of \teigi{highly repetitive} datasets
such as multiple genome sequences from the same species or versioned document collections
(e.g., Wikipedia, GitHub),
dictionary compression methods have recently (re)gained massive attention since they can better capture
more widespread repetitions in such data compared to statistical compression methods~\cite{navarro2021indexing},
and further allow space-efficient full-text indices to be built~\cite{DBLP:journals/csur/Navarro21}.
Some well known methods that fall in this category are
Lempel--Ziv 76/77 factorization based methods~\cite{lempel1976complexity,ziv77lz,DBLP:journals/tcs/KreftN13},
grammar-based compression such as
LZ78~\cite{DBLP:journals/tit/ZivL78},
Re-Pair~\cite{larsson99repair},
SEQUITUR~\cite{DBLP:journals/jair/Nevill-ManningW97},
LCA~\cite{DBLP:conf/spire/SakamotoKS04},
LZD~\cite{DBLP:conf/cpm/GotoBIT15},
and methods involving bidirectional referencing, such as the
run-length encoded Burrows--Wheeler transform (RLBWT)~\cite{DBLP:journals/njc/MakinenN05},
and more recently, lcpcomp~\cite{dinklage2017compression},
plcpcomp~\cite{dinklage19plcpcomp}, lexcomp~\cite{navarro2020approximation},
a method by Russo et al.~\cite{DBLP:conf/dcc/RussoCNF20},
and LZRR~\cite{nishimoto2021lzrr}.

A vital issue in evaluating and comparing these various methods is
to understand how well they can compress a given input compared to the ``optimum''.
While the theoretically smallest representation (aka Kolmogorov complexity)
is incomputable~\cite{DBLP:series/txcs/LiV19},
Kempa and Prezza~\cite{KempaP18}
regarded the output sizes of these methods as \teigi{repetitiveness measures}
and characterized them with respect to the new notion of \teigi{string attractors}.
Namely, they showed that for any input text,
the size of the smallest string attractor is a lower bound for
the output sizes of all known dictionary compressors.
Since then, relations between these various repetitiveness measures
have been heavily investigated~\cite{DBLP:journals/jda/BilleGGP18,kociumaka20towards,DBLP:conf/focs/KempaK20,navarro2021indexing,navarro2020approximation,DBLP:conf/spire/BannaiFIKMN21,doi:10.1137/1.9781611977073.111}.

In this paper, we consider three such repetitiveness measures:
the size $\gamma$ of the smallest string attractor,
the size $g$ of the smallest straight-line program (SLP)~\cite{DBLP:journals/njc/KarpinskiRS97},
and the size $b$ of the smallest bidirectional macro scheme (BMS)~\cite{storer1982data},
all of which are known to be NP-hard to compute~\cite{storer1982data,sakamoto02minimization,KempaP18}.
Thus, any efficient dictionary compression algorithm can (most likely)
merely compute approximations of $\gamma$, $b$, or $g$.
Although for any text, the relation $\delta \leq \gamma\leq b\leq z \leq g$ is known,
where $\delta$~\cite{kociumaka20towards} and $z$~\cite{lempel1976complexity} are repetitiveness measures known to be computable in linear time (cf.~\cite[Lemma~5.7]{christiansen21optimaltime} for $\delta$ and \cite{crochemore08lpf} for $z$),
the gap between the measures can be quite large;
string families giving a logarithmic factor gap are known for each pair of measures~\cite{navarro2021indexing,DBLP:conf/spire/BannaiFIKMN21}.
Since the sizes of some recent data structures such as \cite{navarro19indexing,christiansen21optimaltime},
depend on these repetitiveness measures, their exact sizes are crucial knowledge.

While there exist a vast variety of approximation algorithms for
computing smallest BMSs and grammars as mentioned above,
development of exact algorithms have received very little to almost no attention.
For string attractors, the results of Kempa et al.~\cite{kempa18string} imply a straightforward $O(n2^n)$ time algorithm.
For the smallest grammar, Casel et al.~\cite[Theorem~13]{DBLP:journals/mst/CaselFGGS21} show an
$O^*(3^n)$\footnote{The abstract of~\cite{DBLP:journals/mst/CaselFGGS21} mentions $O(3^n)$ while the statement of the theorem is $O^*(3^n)$.}
time algorithm.
However, we are unaware of any non-trivial implementations or empirical evaluations for computing these measures.
In fact, the only publicly available implementation we could find
was a straight-forward Python script to compute $\gamma$ by  Michael S.\ Branicky~\cite{oeis.org_A339391}.

The main contribution of this paper is to present MAX-SAT formulations~\cite{biere2009handbook} for computing
the smallest string attractor, BMS, and SLP,
thereby providing the first non-trivial implementation for exact computation of
the measures $\gamma$, $b$, and $g$.
The rationale for this approach is that although MAX-SAT is NP-hard,
there are highly optimized solvers whose performance has made incredible progress in recent years.
These solvers can cope with very large instances and can be leveraged,
provided that suitable encodings can be designed~\cite{10.5555/2898950}.
While straight-forward (non-MAX-SAT) implementations become infeasible even for very small text lengths
(e.g. $40$),
computational experiments show that our implementations work for
texts of length up to a few hundred for $b$, $g$, and even more than 1 million for $\gamma$.
Since our addressed problems are all NP-hard,
there is perhaps little hope for our implementations to obtain exact solutions for larger but practically interesting datasets.
Nevertheless, we believe they can make significant impact as a tool for analyzing these repetitiveness measures.
We stress that our solutions not only report the sizes~$\gamma$, $b$ and~$g$, but also give valid
instances
having exactly these sizes (e.g., an SLP that has size~$g$).
It may therefore be possible to improve compression heuristics by studying some of these optimal instances on smaller input strings.
As an example application, we analyzed the recently introduced notion of
\teigi{sensitivity}~\cite{DBLP:journals/corr/abs-2107-08615} of $\gamma$
by conducting an exhaustive computation of $\gamma$ for strings up to certain lengths.
From these computations, we were able to discover a family of strings that exhibit
a multiplicative sensitivity of 2.5, improving the previously known lower bound of 2.0~\cite{DBLP:journals/corr/abs-2107-08615}.

\subsection*{Related Work}

The exact values for $\gamma$, $b$, and $g$ have been characterized only
for a few families of strings.
For standard Sturmian words, $\gamma=2$~\cite{DBLP:journals/tcs/MantaciRRRS21}
and $b=O(1)$ since the RLBWT has constant size~\cite{DBLP:journals/ipl/MantaciRS03} and can be regarded as a BMS.
For the $n$th Thue--Morse word, $\gamma=4$ for $n\geq 4$~\cite{DBLP:conf/spire/KutsukakeMNIBT20},
and $b = n+2$ for $n\geq 2$~\cite{DBLP:conf/spire/BannaiFIKMN21}.
For the $n$th Fibonacci word, $g=n$~\cite{DBLP:journals/corr/abs-2202-08447}.
The smallest attractor sizes of automatic sequences have also been studied~\cite{DBLP:journals/corr/abs-2012-06840}.

\section{Preliminaries}
Let $\Sigma$ be a set of $\sigma$ symbols called the \teigi{alphabet},
and let $\Sigma^*$ denote the set of strings over $\Sigma$.
Given a string $T$, if $T = xyz$ for strings $x,y,z$, then
$x,y,z$ are respectively called a \teigi{prefix}, \teigi{substring}, and \teigi{suffix} of $T$.
They are called \teigi{proper} if they are not equal to $T$.
The length of $T$ is denoted by $|T|$.
For any $i \in [1,|T|]$, let $T[i]$ denote the $i$th symbol of $T$,
i.e., $T = T[1]\cdots T[|T|]$.
For any $1 \leq i \leq j \leq |T|$, let $T[i..j] = T[i] \cdots T[j]$ and $T[i..j) = T[i] \cdots T[j-1]$.

For the rest of this paper, we fix a string~$T$, and let $n := |T|$ denote its length.
Further, we assume that each symbol of $\Sigma$ appears in $T$.
Let $\occ(P) = \{ i \mid T[i..i+|P|-1] = P, 1 \leq i \leq n-|P|+1\}$ be the set of starting positions of all occurrences of a substring $P$ in~$T$,
and let
$\cover(P) = \{ i + k - 1\mid i \in \occ(P), 1 \leq k \leq |P|\}$
be the set of all text positions covered by all occurrences of~$P$ in~$T$.

A set of positions $\Gamma\subseteq [1,n]$ is a \teigi{string attractor}~\cite{KempaP18} of $T$
if every substring $P$ of $T$ has an occurrence in $T$ that contains an element of $\Gamma$, that is, $\Gamma\cap\cover(P)\neq \emptyset$.
We denote the size of the smallest string attractor of~$T$ by $\gamma$.
For example, $[1,n]$ is a trivial string attractor.
$\{1,2,3\}$ is a (smallest) string attractor of $T=\mathtt{banana}$.
(See also Figure~\ref{fig:sat_attr})

A \teigi{straight-line program} (SLP)~\cite{DBLP:journals/njc/KarpinskiRS97} is a grammar in Chomsky normal form whose language consists solely of~$T$.
In other words, (1) each production rule is of the form
$X\rightarrow X_\ell X_r$ or $X\rightarrow c$, where $X_\ell$, $X_r$ are non-terminals and $c\in \Sigma$,
(2) there is exactly one such production rule for any given non-terminal symbol $X$,
and
(3) there is a start symbol whose iterative expansion finally leads to $T$.
The \teigi{size} of an SLP is the number of its production rules, or equivalently
(assuming that each non-terminal is used at least once),
the number of distinct non-terminals.
We denote the size of the smallest SLP that produces~$T$ by $g$.
For example, the set of production rules
$\{ X_9 \rightarrow X_6 X_8$,
$X_8 \rightarrow X_7 X_7$,
$X_7 \rightarrow X_1 X_3$,
$X_6 \rightarrow X_4 X_5$,
$X_5 \rightarrow X_3 X_3$,
$X_4 \rightarrow X_3 X_1$,
$X_3 \rightarrow X_1 X_2$,
$X_2 \rightarrow \mathtt{b}$,
$X_1 \rightarrow \mathtt{a}\}$
is an SLP of size 9 for $T=\mathtt{abaababaabaab}$.
See also Figure~\ref{fig:slp_example}.

\begin{figure}[t]
    \begin{center}
        \includegraphics[width=0.8\linewidth]{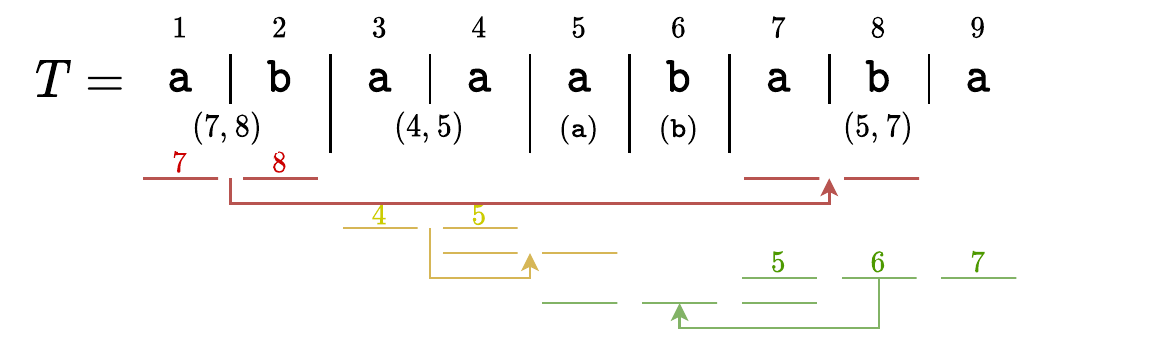}
    \end{center}
    \hfill
    \caption{A bidirectional macro scheme (BMS) of $T = \mathtt{abaaababa}$.
        The figure depicts the BMS $(7,8),(4,5),\mathtt{a},\mathtt{b},(5,7)$.
        The reference of each non-ground phrase is visualized by an arrow.
        The phrase references imply a reference for each symbol in non-ground phrases.
    }
    \label{fig:k_bms}
\end{figure}

A \teigi{bidirectional macro scheme} (BMS) ~\cite{storer1982data} of size $m$ representing~$T$,
is a factorization $T = F_1, \cdots, F_m$,
where each factor (or phrase) is a single symbol (which we call a ground phrase), or,
is encoded as a pair of integers $(i,j)$ indicating that it
references (i.e., is a copy of) substring $T[i..j]$.
A BMS is said to be {\em valid}, if $T$ can be reconstructed from
the representation of such a factorization, i.e.,
the implied references of each symbol in a non-ground phrase is acyclic,
and eventually leads to a ground phrase.
We denote the size of the smallest valid BMS that represents~$T$ by $b$.
Figure~\ref{fig:k_bms} shows a valid BMS
$(7,8),(4,5),\mathtt{a},\mathtt{b},(5,7)$ representing the string $\mathtt{abaaababa}$.
For example, the $\mathtt{a}$ at position $9$ references position $7$,
which in turn references position $5$, a ground phrase.

The satisfiability (SAT) problem asks
for an assignment of variables that satisfies a given Boolean formula~\cite{10.5555/2898950,biere2009handbook}.
The input formula is usually given in \teigi{conjunctive normal form} (CNF),
which consists of a conjunction of clauses, and each clause is a disjunction of literals.
A \teigi{literal} is a Boolean variable or its negation.
In this form, the given formula is satisfied if and only if all the clauses
(which we will sometimes call constraints) are satisfied.
The \teigi{size} of a CNF is the sum of the literals in all clauses.

A \teigi{maximum satisfiability} (MAX-SAT) problem is an extension of SAT,
where two types of clauses, \teigi{hard} and \teigi{soft}, are considered~\cite{biere2009handbook}.
A solution to a MAX-SAT instance is a truth assignment of the variables
such that the number of satisfied soft clauses is maximized
under the restriction that all hard clauses must be satisfied.

We will use $1$ to denote true, and $0$ to denote false.
Furthermore, for a set $\{ v_i \}_{i=1}^k$ of Boolean variables,
cardinality constraints of the form $\sum_{i=1}^k v_i \leq 1$ are known as
\teigi{atmost-one constraints}.
Although a straightforward encoding has size $\Theta(k^2)$,
$O(k)$ size encodings are known~\cite{DBLP:conf/cp/Sinz05}.
Constraints of the form
$\sum_{i=1}^k v_i = 1$ can be encoded using a combination of an atmost-one constraint
and a simple disjunction of all the variables (i.e., atleast-one)
and thus can also be encoded in $O(k)$ size.

\section{Reductions to MAX-SAT}
In what follows, we present our encodings for the aforementioned problems.
Common to all encodings is the idea that we have a Boolean variable $p_i$ for each text position~$i \in [1,n]$,
which counts, when set to true,
an element of a string attractor,
a non-terminal (actually, to be precise, a factor in a grammar parsing) of an SLP,
or a phrase of a BMS.
Since our goal is to have as few $p_i$'s set to true as possible,
our soft clauses have the form $D_i = \lnot p_i$ for $i\in[1,n]$.
Consequently, all our encodings have the same number of soft clauses,
and only differ in how the hard clauses are defined.

\subsection{Smallest String Attractor as MAX-SAT}
We start with a simple
encoding based on the definition of string attractors.
Subsequently, we utilize an observation similar to but slightly more generalized than that made in~\cite{kempa18string},
in order to reduce the size of hard clauses.

\subsubsection{Simple Encoding}
\label{sec:attr:naiive}
\newcommand*{\SubstringSet}{\ensuremath{\mathcal{S}_T}}
Our idea is to design a CNF so that a MAX-SAT solution will encode a string attractor $\Gamma$,
where $p_i = 1$ if and only if position $i$ is an element of $\Gamma$~(i.e., $\Gamma = \{ i \mid 1 \leq i \leq n, p_{i} = 1 \}$).
Let $\SubstringSet$ denote the set of all non-empty substrings of $T$, i.e., $\SubstringSet = \{T[i..j] \mid 1 \leq i \leq j \leq n\}$.
For each substring $S$ of $\SubstringSet$, we define a hard clause
$C_S = \bigvee_{i \in \cover(S)}p_i$.
(See Figure~\ref{fig:sat_attr} for an example.)
By the definition of $\cover(S)$,
the set $\Gamma$ corresponding to any truth assignment for $p_i$ will be a string attractor if and only if
all hard clauses $C_S$ are satisfied.
Since our soft clauses have the form $D_i = \lnot p_i$ for $i \in [1, n]$,
the soft clauses ensure that the MAX-SAT solution minimizes the number of $p_i$'s being true.
Thus, we can obtain the smallest string attractor by solving the MAX-SAT on $C_S$ and $D_i$.

Each hard clause $C_{S}$ has size $|\cover(S)|=O(n)$.
Since there are $O(n^2)$ substrings, the number of hard clauses is $O(n^2)$.
Hence, the total size of the CNF is $O(n^3)$.
In the next subsection, we reduce the size to $O(n^2)$.

\begin{figure}[h]
    \begin{center}
        \includegraphics[scale=0.5]{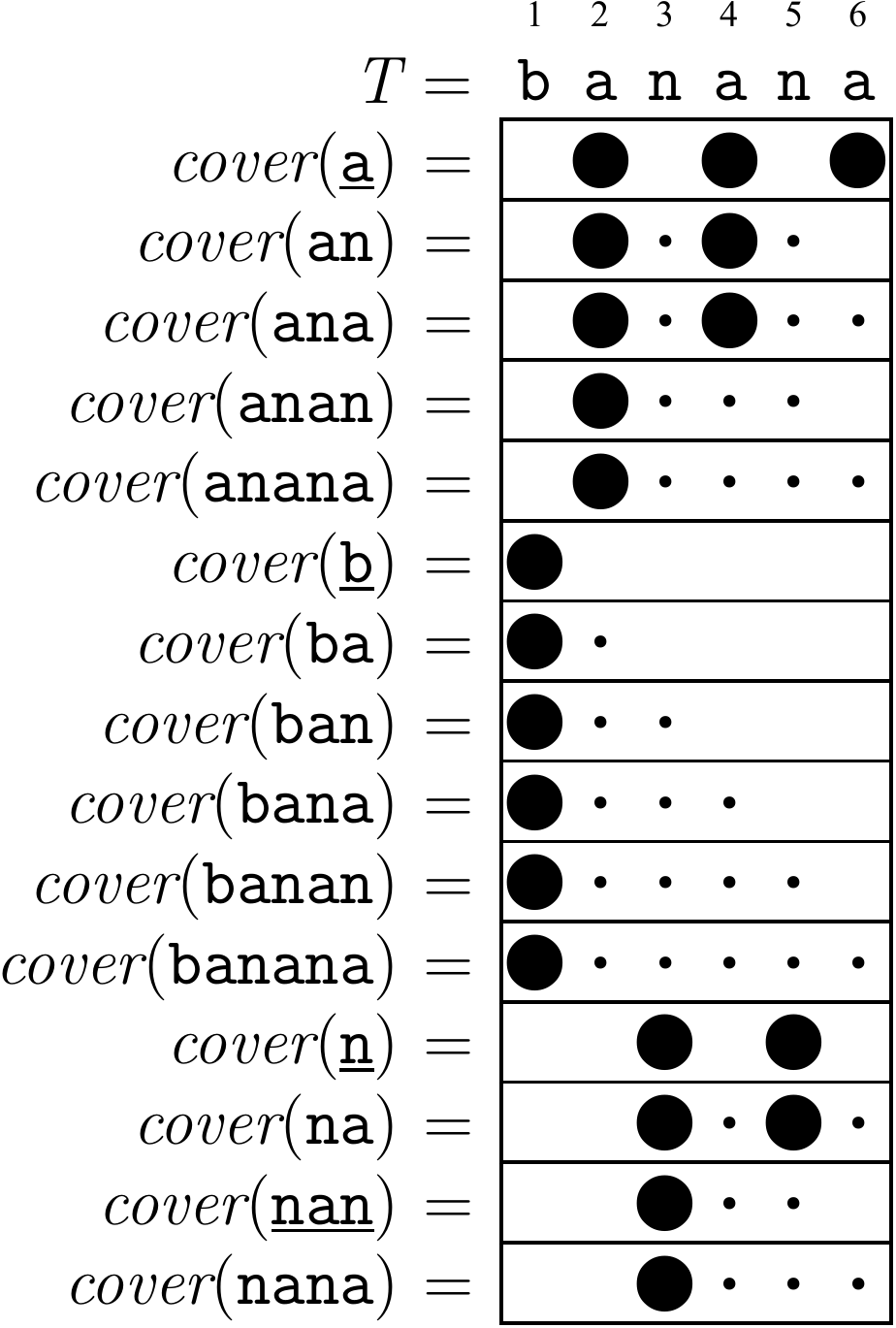}
    \end{center}
    \caption{String $T=\mathtt{banana}$ and the positions that each distinct substring of $T$ covers.
        We list all distinct substrings of $T$ on the left hand side, and show on the right hand side
        their covers.
        A dot at position~$k$ in the row for substring $S$ indicates that $k$ is included in $\cover(S)$ (i.e. is covered by $S$),
        and a large dot indicates $k\in\occ(S)$.
        Underlined substrings are minimal substrings of $T$.
        For example,
        $\cover(\mathtt{an}) =\{ 2, 3, 4, 5 \}$.
        The clause defined for $\mathtt{an}$ in our encoding
        is $C_{\mathtt{an}} = p_{2} \vee p_{3} \vee p_{4} \vee p_{5}$.
    }
    \label{fig:sat_attr}
\end{figure}

\subsubsection{Reducing CNF Clauses via Minimal Substrings}
\label{sec:minsub}

We can reduce the number of hard clauses in our CNF by considering only members of $\SubstringSet$ that are \teigi{minimal substrings}\footnote{Kempa et al.~\cite{kempa18string} use a similar idea when reducing the problem to set cover.
    Their formulation can be regarded as considering only right-minimal substrings (i.e., $|\occ(S[1..|S|-1])| > |\occ(S)|$),
    while we consider a potentially smaller subset requiring both right-minimality and left-minimality.
    For texts in the Calgary corpus, we observed that the difference between
    minimal and right-minimal substrings can result in a difference
    as large as $50$ times in their total lengths (progp and trans),
    i.e., the total size of hard clauses.
}.
A substring $S$ of string $T$ is called a minimal substring of $T$ if all proper substrings of $S$ occur more often than $S$ in $T$~(i.e.,
$|\occ(S[i..j])| > |\occ(S)|$ for every proper substring $S[i..j]$ of $S$).
By the definition of minimal substrings, the following lemma holds.

\begin{lemma}\label{lem:minimal_substring}
    For every non-minimal substring $S$ of $T$,
    there is a minimal substring $S_{\mathsf{min}}$ of $S$
    with $\cover(S_{\mathsf{min}}) \subseteq \cover(S)$.
\end{lemma}
\begin{proof}
    Because $S$ is not minimal, it has substrings that have the same number of occurrences as $S$.
    Let $S_{\mathsf{min}} = S[e..e + |S_{\mathsf{min}}| - 1]$ be one of these substrings that is minimal, for some~$e$.
    Then by definition,
    for each occurrence $i_{\mathsf{min}} \in \occ(S_{\mathsf{min}})$ of $S_{\mathsf{min}}$,
    there exists an occurrence $i \in \occ(S)$ of $S$ such that $i = i_{\mathsf{min}} - e + 1$.
    $\{ i_{\mathsf{min}}, i_{\mathsf{min}} + 1, \ldots, i_{\mathsf{min}} + |S_{\mathsf{min}}| - 1 \} \subseteq \{ i, i+1, \ldots, i + |S| - 1 \}$,
    and hence, $\cover(S_{\mathsf{min}}) \subseteq \cover(S)$.
\end{proof}

In the example $T=\mathtt{banana}$,
$|\occ(\mathtt{nan})| < |\occ(\mathtt{na})|, |\occ(\mathtt{an})|, |\occ(\mathtt{a})|, |\occ(\mathtt{n})|$,
and thus, substring $\mathtt{nan}$ is a minimal substring of $T$.
Furthermore, $\mathtt{nan}$ is a substring of $\mathtt{nana}$, and $|\occ(\mathtt{nana})| = |\occ(\mathtt{nan})|$.
Thus, $\cover(\mathtt{nan}) \subseteq \cover(\mathtt{nana})$ by Lemma~\ref{lem:minimal_substring}.
(See also Figure~\ref{fig:sat_attr})

Lemma~\ref{lem:minimal_substring} ensures that if an assignment of variables satisfies the hard clauses~$C_S$ for all minimal substrings~$S$ of $T$,
then the assignment satisfies the hard clauses~$C_S$ for all substrings~$S$ of $T$.
With this observation we can conclude that we can omit the hard clauses for all substrings~$S$ of $T$ that are not minimal.

The number $m$ of minimal substrings is $O(n)$ because
minimal substrings correspond to minimal strings,
defined by Blumer et al.~\cite{DBLP:journals/jacm/BlumerBHME87}
based on an equivalence relation over substrings of $T$,
and their number is known to be $O(n)$ (Lemma 3 in~\cite{DBLP:journals/algorithmica/NarisawaHIBT17}).
Hence, the total size of the CNF is reduced to $O(mn) \subseteq O(n^2)$.

In particular, the size of the CNF is $o(n^2)$ if $m = o(n)$. We can show that there exists a family of strings~$\{T_d\}_{d \in \mathcal{I}}$ for a non-finite set of natural numbers $\mathcal{I}$ with $|T_d| = d^2$ having $o(d^2)$ minimal substrings (hence, for $n = d^2$, $m = o(n)$).
To this end, let $T_d$ be the string $S_{1}S_{2} \cdots S_{d}$
of length $n = d^{2}$ over
the alphabet $\Sigma = \{ \mathtt{a}, \$_{1}, \$_{2}, \ldots, \$_{d} \}$,
where $S_{i} = \mathtt{a}^{d-1}\$_{i}$, and $\mathtt{a}^{d-1}$ is the repetition of character $\mathtt{a}$ with length $d-1$.
Then $m = 2\sqrt{n} - 1$
because the minimal substrings of $T_d$ are $\mathtt{a}^{1}$, $\mathtt{a}^{2}$, $\ldots$, $\mathtt{a}^{\sqrt{n}-1}$, $\$_{1}$, $\$_{2}$, $\ldots$, $\$_{\sqrt{n}}$.

\subsection{Smallest Straight-Line Program as MAX-SAT}
To encode a grammar in SAT, we utilize
a notion called \teigi{grammar parsing}
introduced by Rytter~\cite{DBLP:journals/tcs/Rytter03}.
Given an SLP $G$ that produces $T$, the \teigi{parse tree} of $T$ with respect to $G$ is
a derivation tree of $T$, where internal nodes are non-terminal symbols that derive two non-terminal symbols, and leaves are non-terminal symbols that derive a single terminal symbol.
The \teigi{partial parse tree} of $T$ with respect to $G$ is the tree obtained by
pruning the parse tree of $T$ with respect to $G$ so that any internal node is always a first occurrence
in a left to right pre-order traversal of the parse tree,
i.e., the non-terminal symbol of an internal node is not used in the partial parse tree for any
corresponding substring to its left.
In other words,
if a non-terminal symbol $X$ that derives two non-terminal symbols is a leaf of the partial parse tree,
the existence of a unique internal node having the same non-terminal symbol $X$ corresponding to a substring to its left is implied.
We will say that the leaf \teigi{references} the internal node.
The \teigi{grammar parsing} of $T$ with respect to $G$,
is the factorization of $T$ consisting of substrings corresponding to the leaves of the partial parse tree of $T$ with respect to $G$.
See Figure~\ref{fig:slp_example} for an example.
\begin{figure}
    \begin{center}
        \includegraphics[width=0.8\linewidth]{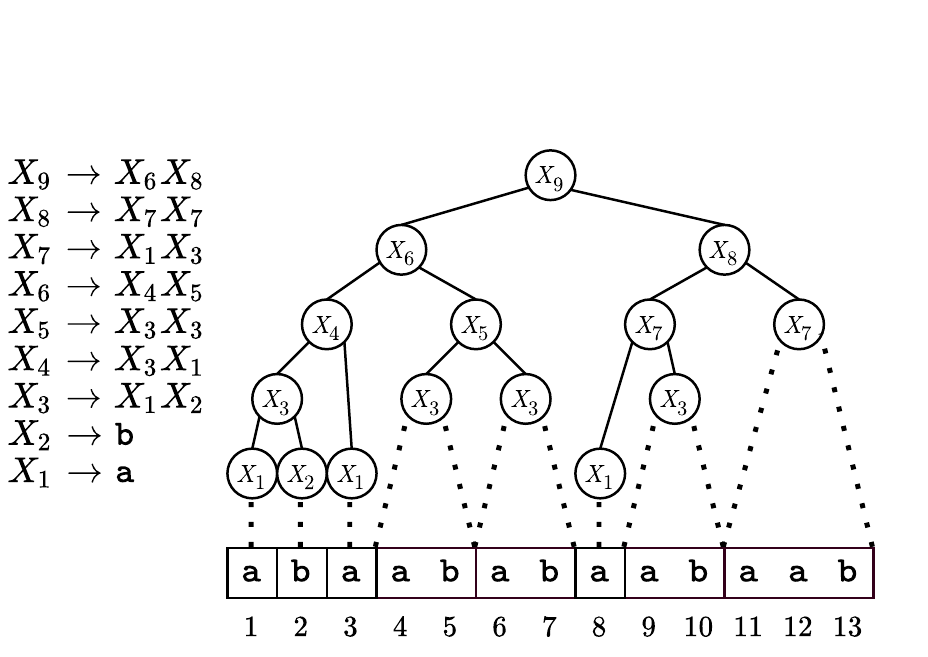}
    \end{center}
    \caption{The partial parse tree and the grammar parsing
        of an SLP for the string $T=\mathtt{abaababaabaab}$.
        Each internal node is a unique non-terminal symbol.
        The grammar parsing represented by the rectangles partitioning $T$ is
        $\mathtt{a},\mathtt{b},\mathtt{a},\mathtt{ab},\mathtt{ab},\mathtt{a},\mathtt{ab},\mathtt{aab}$ of size $8$.
        The size of the SLP is $8+|\{\mathtt{a},\mathtt{b}\}| - 1 =9$.
    }
    \label{fig:slp_example}
\end{figure}

The \teigi{size} of the grammar parsing is equal to the number of leaves in the partial parse tree.
It is easy to see that by definition,
the internal nodes in the partial parse tree are distinct,
consisting of (all) non-terminal symbols that derive two non-terminal symbols.
There are $\sigma$ more non-terminal symbols that derive a single terminal symbol.
Therefore,
$(\mbox{\# of internal nodes})+ \sigma$ is the size of the SLP\@.
Since the partial parse tree is a full binary tree,
$(\mbox{\# of internal nodes}) = (\mbox{\# of leaves})-1$,
and thus the size of the SLP is equal to
$(\mbox{size of the grammar parsing})+ \sigma - 1$.
As $\sigma$ is independent of the choice of the SLP for $T$,
minimizing the size of the grammar parsing is equivalent to minimizing the SLP.

Our formulation is based on the following lemma.

\begin{lemma}\label{lemma:grammar_parsing_and_reference_structure}
    A factorization $T=F_1\cdots F_m$ for $T$
    is the grammar parsing of an SLP for $T$ if and only if
    (i) for each factor $F_k$ longer than $1$, there exist $i_k < j_k < k$ such that:
    $F_k = F_{i_k}\cdots F_{j_k}$
    and (ii) for any pair of factors $F_{x}=F_{i_{x}}\cdots F_{j_{x}}$
    and $F_{y}=F_{i_{y}}\cdots F_{j_{y}}$
    longer than $1$,
    (i.e., $(x, y) \in \{ (x', y') \mid 1 \leq x', y' \leq m, |F_{x'}| > 1, |F_{y'}| > 1  \}$),
    the intervals $[i_{x}..j_{x}]$ and $[i_{y}..j_{y}]$
    are either disjoint or one is a sub-interval of the other.
\end{lemma}
\begin{proof}
    ($\Rightarrow$)
    Suppose $F_1\cdots F_m$ is the grammar parsing of some SLP for $T$.
    Then, any $F_k$ longer than $1$ has an implied corresponding internal node to the left in the partial parse tree.
    Since an internal node derives at least two leaves, it
    derives $F_{i_k} \cdots F_{j_k}$ corresponding to the interval $[i_k..j_k]$ of the factorization for some $i_k < j_k < k$.
    Furthermore, since all of these intervals are derived from internal nodes of a tree,
    they must respect the tree structure, i.e., any two of them must be disjoint or contained in one another.\\
    ($\Leftarrow$)
    Suppose we are given a factorization $T=F_1\cdots F_m$ of $T$,
    as well as for each $F_k$, a corresponding interval $[i_k..j_k]$
    of the factorization satisfying the conditions of the lemma.
    Since, for any pair of factors $F_{x}=F_{i_{x}}\cdots F_{j_{x}}$
    and $F_{y}=F_{i_{y}}\cdots F_{j_{y}}$,
    the intervals $[i_{x}..j_{x}]$ and $[i_{y}..j_{y}]$
    are disjoint or contained in one another, we can construct a
    tree with the internal nodes corresponding to the intervals and the leaves corresponding to the factors of the factorization,
    where a node is a descendant of another if and only if it is a sub-interval.
    Although such a tree can be multi-ary in general,
    we can add internal nodes and transform it into a full binary tree while preserving ancestor/descendant relations of nodes/leaves in the original tree
    (note that the resulting tree may not be determined uniquely, but its size will always be the same).
    We assign to each internal node a distinct non-terminal symbol.
    To each leaf corresponding to a factor $F_k$ longer than $1$,
    we assign the same non-terminal symbol that we assigned to the internal node corresponding to $F_{i_k}\cdots F_{j_k}$.
    Finally, we assign each leaf corresponding to a factor of length $1$ a non-terminal symbol that derives the corresponding terminal symbol.
    The resulting tree is a partial parse tree for an SLP of size $m+\sigma-1$ for $T$ with $F_1\cdots F_m$ as its grammar parsing.
\end{proof}

We define Boolean variables as follows to encode Lemma~\ref{lemma:grammar_parsing_and_reference_structure}.
\begin{itemize}
    \item $f_{i,\ell}$ for $i \in [1,n], \ell\in [1,n+1-i]$: $f_{i,\ell} = 1$ if and only if $T[i..i+\ell)$ is a factor of the grammar parsing.
    \item $p_{i}$ for $i\in[1,n+1]$: For $i \neq n+1$,
          $p_i = 1$ if and only if $i$ is a starting position of a factor of the grammar parsing.
          $p_{n+1}$ is for technical reasons. We set $p_1 = p_{n+1} = 1$.
    \item $\mathit{ref}_{i'\leftarrow i,\ell}$ for
          $i,\ell,i' \in [1,n]$, s.t. $\ell \geq 2$, $i'\leq i - \ell$ and
          $T[i'..i'+\ell) = T[i..i+\ell)$:
          $\mathit{ref}_{i'\leftarrow i,\ell} = 1$ if and only if $T[i..i+\ell)$ is a factor of the grammar parsing,
          and the implied internal node of the partial parse tree corresponds to $T[i'..i'+\ell)$.
    \item $q_{i',\ell}$ for $i' \in [1,n-1], \ell \in[2,n+1-i']$ s.t. $T[i'..i'+\ell)$ has an occurrence in $T[i'+\ell..n]$:
          $q_{i',\ell} = 1$ if and only if $T[i'..i'+\ell)$ corresponds to an internal node of the partial parse tree that is referenced by at least one factor of the grammar parsing.
\end{itemize}

We next define constraints that the above variables must satisfy.

First, since each factor of the grammar parsing is disjoint and the concatenation of all factors must be equal to $T$,
the truth values of $f_{i,\ell}$ must uniquely define the truth values for $p_i$ and vice versa. This can be encoded as
\begin{align}
    \forall i\in[1,n], \ell \in[1,n+1-i]:
    f_{i,\ell} \iff p_i \land (\lnot p_{i+1}) \cdots (\lnot p_{i+\ell-1}) \land p_{i+\ell}\label{constraint:slp:f_iff_p}
\end{align}
For all $i$ and $\ell \geq 2$ such that $T[i..i+\ell)$ is the first occurrence of a substring $S = T[i..i+\ell)$ of~$T$,
$T[i..i+\ell)$ cannot be a factor of a grammar parsing. Thus, we require:
\begin{align}
    \forall i\in[1,n-1], \ell\in[2,n-i+1] \mbox{ s.t. }
    T[i..i+\ell) \mbox{ does not occur in } T[1..i):
    \lnot f_{i,\ell}\label{constraint:firstoccurrence_not_factor}
\end{align}
If $T[i..i+\ell)$ is not the first occurrence of~$S$, $T[i..i+\ell)$ can be a factor.
If $T[i..i+\ell)$ is a factor of the grammar parsing of length at least $2$,
then, there must exist at least one
$i' \leq i-\ell$ such that
$T[i'..i'+\ell) = S$ and
$T[i'..i'+\ell)$ corresponds to an internal node of the partial parse tree.
This can be encoded as
\begin{align}
    \begin{split}
        \forall i\in[1,n],\ell\in[2,n+1-i]
        \mbox{ s.t. } T[i..i+\ell) \mbox{ occurs in } T[1..i):\\
        f_{i,\ell} \implies
        \bigvee_{i' \in \{ k \mid  T[k..k+\ell) = T[i..i+\ell), k\in[1,i-\ell) \}}
        \mathit{ref}_{i'\leftarrow i,\ell}.
        \label{constraint:slp:f_implies_ref}
    \end{split}
\end{align}
Furthermore, for any $i,\ell$,
a factor $T[i..i+\ell)$ references at most one position, i.e.,
\begin{align}
    \forall i\in[1,n],\ell\in[2,n+1-i]:
    \sum_{i' \in \{ k \mid  T[k..k+\ell) = T[i..i+\ell), k\in[1,i-\ell) \}} \mathit{ref}_{i'\leftarrow i, \ell}\leq 1.
    \label{constraint:slp:ref_atmost_one}
\end{align}
On the other hand, $\mathit{ref}_{i'\leftarrow i,\ell}=1$ implies that
$T[i..i+\ell)$ is a factor of the grammar parsing.
Therefore,
\begin{align}
    \begin{split}
        \forall i\in[1,n],\ell\in[2,n+1-i],
        i' \in \{ k \mid  T[k..k+\ell) = T[i..i+\ell), k\in[1,i-\ell) \}:\\
        \mathit{ref}_{i'\leftarrow i,\ell} \implies f_{i,\ell}
        \label{constraint:slp:ref_implies_f}
    \end{split}
\end{align}
By definition, it holds that
\begin{align}
    \begin{split}
        \forall i'\in[1,n-1], \ell\in[2,n+1-i']
        \mbox{ s.t. } T[i'..i'+\ell) \mbox{ has an occurrence in } T[i'+\ell..n]:\\
        q_{i',\ell} \iff \bigvee_{1\leq i'+\ell\leq i\leq n}\mathit{ref}_{i'\leftarrow i, \ell}
        \label{constraint:slp:q_iff_ref}
    \end{split}
\end{align}
Next, as shown in Lemma~\ref{lemma:grammar_parsing_and_reference_structure},
we require that the implied internal node that is referenced by some factor
must be an interval of size at least 2 of the factorization.
We encode this as:
\begin{align}
    \begin{split}
        \forall i'\in[1,n-1],\ell \in[2,n+1-i']
        \mbox{ s.t. } T[i'..i'+\ell) \mbox{ has an occurrence in } T[i'+\ell..n]:\\
        q_{i',\ell} \implies
        \lnot f_{i',\ell}\land p_{i'}\land p_{i'+\ell}
        \label{constraint:slp:q_implies_notf_and_p}
    \end{split}
\end{align}
Also, for any two such implied internal nodes
$T[i_1..i_1+\ell_1)$ and $T[i_2..i_2+\ell_2)$,
they must either be disjoint, or one is a sub-interval of the other.
In other words, it cannot be that
a proper prefix interval of one is a proper suffix interval of the other,
i.e.,
\begin{align}
    \begin{split}
        \forall i_1,i_2,\ell_1,\ell_2 \mbox{ s.t. } i_1 < i_2 < i_1+\ell_1 < i_2+\ell_2
        \mbox{ s.t. }\\
        T[i_k..i_k+\ell) \mbox{ has an occurrence in } T[i_k+\ell..n] \mbox{ for } k\in\{1,2\}:\\
        \lnot q_{i_1,\ell_1} \lor \lnot q_{i_2,\ell_2}\label{constraint:nooverlapping}
    \end{split}
\end{align}

In total, we have $O(n^3)$ Boolean variables dominated by $\mathit{ref}_{i'\leftarrow i,\ell}$.
The size of each clause is at most $O(n)$.
The total size of the resulting CNF is $O(n^4)$,
dominated by Constraint~(\ref{constraint:nooverlapping}) where there are $O(n^4)$ clauses of $O(1)$ size each.

\paragraph*{Correctness of the Encoding}
We now prove the correctness of our formulation.
From Lemma~\ref{lemma:grammar_parsing_and_reference_structure},
if we are given some SLP producing $T$, it is clear that the above Boolean variables corresponding to
its partial parse tree, referencing structure, and grammar parsing
will satisfy all of the constraints.

Next, suppose we are given $T$ and a truth assignment satisfying the above constraints.
Starting from the truth assignments of $p_i$ and Constraint~(\ref{constraint:slp:f_iff_p}),
we can obtain a factorization of $T$
where we regard $T[i..i+\ell)$ as a factor if and only if $f_{i,\ell} = 1$.
For any $i\in[1,n]$ and $\ell\in[2,n-i+1]$,
Constraint~(\ref{constraint:firstoccurrence_not_factor}) ensures that
$T[i..i+\ell)$ having an occurrence in $T[1..i)$ is a necessary condition for $f_{i,\ell}=1$.
If $f_{i,\ell} = 1$, Constraint~(\ref{constraint:slp:f_implies_ref})
implies that there is some $i' \in [1..i-\ell)$ such that
$T[i'..i'+\ell) = T[i..i+\ell)$ and $\mathit{ref}_{i'\leftarrow i,\ell} = 1$.
From Constraint~(\ref{constraint:slp:ref_atmost_one}),
we know that there is exactly one such $i'$.
On the other hand, Constraint~(\ref{constraint:slp:ref_implies_f}) ensures that
$\mathit{ref}_{i'\leftarrow i, \ell} = 0$ for all $i'$
when $f_{i,\ell} = 0$.
Thus, for each $f_{i,\ell}=1$ with $\ell > 1$ there exists
exactly one $i'$ such that $\mathit{ref}_{i'\leftarrow i,\ell}=1$,
and all other $\mathit{ref}_{\cdot\leftarrow\cdot,\cdot}$ are $0$.
From Constraint~(\ref{constraint:slp:q_iff_ref}),
it holds that $q_{i',\ell} = 1$ if and only if there is at least one
$i,\ell$ with $\mathit{ref}_{i'\leftarrow i,\ell}=1$ and thus $f_{i,\ell}=1$.
If $q_{i',\ell} = 1$, from Constraint~(\ref{constraint:slp:q_implies_notf_and_p}),
we have
$f_{i',\ell} = 0$, $p_{i'} = p_{i'+\ell} = 1$, implying that
$T[i'..i'+\ell)$ is not a factor, but is a concatenation of two or more factors.
Constraint~(\ref{constraint:nooverlapping}) requires that all such
$T[i'..i'+\ell)$ are either disjoint or that one is a sub-interval of the other.

Thus, from the above arguments, we can see that for the factorization
defined by the $p_i$'s,
we can associate for each factor, a subinterval of the factorization that satisfies the conditions of Lemma~\ref{lemma:grammar_parsing_and_reference_structure},
thus implying that the factorization is a grammar parsing of some SLP.

\subsection{Smallest Bidirectional Macro Scheme to MAX-SAT}

For our SLP encoding, we used the fact that only the leftmost occurrences of the non-terminals are internal nodes -- we modeled every later occurrence as a leaf referring to this leftmost occurrence.
We could therefore evade the problem of constructing reference cycles since all references point in the same direction.
However, in a bidirectional scheme, the references can point in either direction,
and the difficulty in defining the encoding is how to ensure that no cycles are introduced in the referencing.

\begin{figure}[h]
    \begin{center}
        \includegraphics[width=0.8\linewidth]{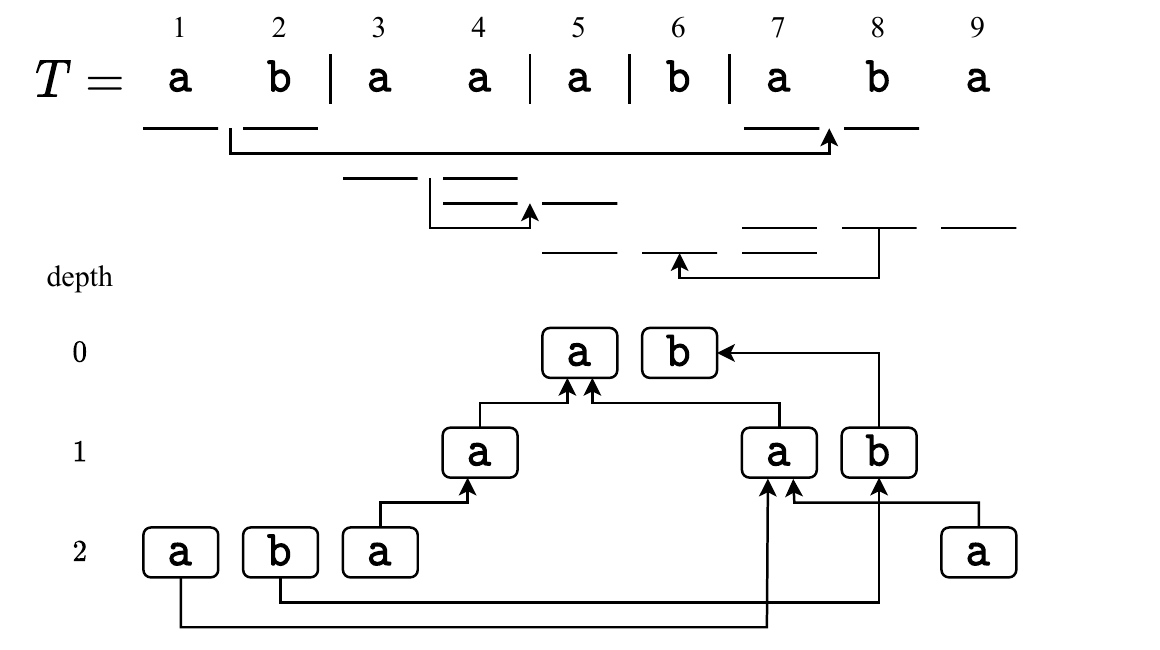}
    \end{center}
    \caption{Reference forest of the BMS of \cref{fig:k_bms}.
        The forest consists of two trees. The root of each tree corresponds to one of the two ground phrases of the BMS\@.
        For instance, to decode $T[1]$, we first need to decode $T[7]$ (the parent of $T[7]$), which has the tree root $T[5]$ as its parent.
        Hence, the number of ancestors of a node~$T[i]$ is the number of references we need to traverse to obtain a ground phrase storing the character of $T[i]$.
    }
    \label{fig:one_bms}
\end{figure}

Here, we present a solution that again works with a tree structure, but this time we have multiple trees --- a forest that represents the references.
In detail, we follow Dinklage et al.~\cite[Definition~6]{dinklage19plcpcomp}, who represented a bidirectional macro scheme by a reference forest,
where a text position~$i$ has text position~$j$ as its parent if the phrase covering $T[i]$ has a reference stating that $T[i]$ is copied from $T[j]$.
Figure~\ref{fig:one_bms} visualizes such a forest.
The roots of this reference forest are the positions of the ground phrases.

In order to find a BMS, we go in the inverse direction, and first encode a reference forest from which we
subsequently derive a BMS\@.
Since a forest has no cycles,
we can use the edges of the forest to define a valid BMS, where each factor has length one (each factor is represented by a node in the reference forest).
The final step is to glue together adjacent positions that have adjacent references into larger factors to obtain BMSs with fewer factors.

We start with the encoding for our reference forest.
The nodes of the forest coincide with the text positions, and are therefore enumerated from $1$ to $n$.
Since a text position $i$ can reference text position~$j$ only when $T[i] = T[j]$, it makes sense to restrict $j$ to belong to the set $\match{i} := \{j \in [1,n] \mid T[i] = T[j], i \neq j\}$.
In that case, we say that $j$ is the parent of $i$.
We make use of the following variables.
\begin{itemize}
    \item $\rot{i}$ for $i \in [1,n]$ : $\rot{i}=1$ if and only if node $i$ is the root of a tree.
          All roots are at depth $0$.
    \item $\dref{d, i\rightarrow j}$ for $d \in [1,n-1], i \in [1,n], j \in \match{i}$ : $\dref{d, i\rightarrow j}=1$ if and only if node $i$ at depth $d$ has a parent node $j$ at depth $d-1$.
\end{itemize}

To obtain a valid reference forest, we define the following constraints.
First, each node is a root node or has a parent.
\begin{align}\label{eqBMSDref}
    \forall i \in [1,n]: \rot{i} + \sum_{d\in[1,n],j\in \match{i}} \dref{d,i\rightarrow j}=1
\end{align}

According to Constraint~(\ref{eqBMSDref}),
a node~$i$ at depth~$d \ge 2$ must have exactly one parent~$j$, and $j$ must also have a parent node~$k$ (since $d \ge 2$).
To enforce acyclicity,
we additionally want that $k$ is exactly two levels above of~$i$.

\begin{align}\label{eqBMSTransitiveRef}
    \forall d \in [2,n], \forall i\in[1,n], \forall j \in \match{i}:
    \dref{d,i\rightarrow j} \implies \sum_{k\in \match{j}} \dref{d-1,j\rightarrow k}=1
\end{align}

Next, to translate our reference forest to a BMS, we additionally introduce the following Boolean variables.
\begin{itemize}
    \item $\refpos{i\rightarrow j}$ for $i\in [1,n], j\in \match{i}$: $\refpos{i\rightarrow j}=1$ if and only if position $i$ references position $j$.
    \item $\fbeg{i}$ for $i\in[1,n]$ : $\fbeg{i}=1$ if and only if position $i$ is a beginning of a phrase.
          Note that $p_1=1$.
\end{itemize}

The connection between the variables of the reference forest and our BMS is as follows.
For each position $i \in [1, n]$, $i$ can reference at most one position $j \in \match{i}$, i.e.,
\begin{align}\label{eqBMSAtMostOneRef}
    \forall i\in[1,n]: \sum_{j\in \match{i}} \refpos{i\rightarrow j} \leq 1
\end{align}
A position $i$ references $j$ if, on any depth~$d$ of the reference forest, there is an edge from $i$ to its parent~$j$ modeled by $\dref{d,i\rightarrow j}$.
\begin{align}\label{eqBMSOnlyOneParent}
    \forall d \in [1,n], \forall i\in[1,n], \forall j \in \match{i}:
    \dref{d,i\rightarrow j} \implies \refpos{i\rightarrow j}
\end{align}

Finally, the roots in our reference forest model the ground phrases of the BMS\@. The roots therefore cannot have a reference, but instead introduce a factor (of length one).
\begin{align}\label{eqBMSGroundPhrase}
    \forall i\in[1,n]:
    \rot{i} \implies \fbeg{i}.
    \text{~Additionally,~}
    \forall j \in \match{i}:
    \rot{i} \implies \lnot \refpos{i\rightarrow j}
\end{align}

Remembering that the phrases are determined by the variables $\fbeg{i}$'s witnessing their starting positions,
it is left to model the constraints for the truth assignment of the $p_i$'s.
For that, let us conceptually fix a text position~$i$ for which we assume that it references text position~$j$.
We consider two cases where $T[i-1]$ and $T[i]$ cannot be in the same phrase.
The first case is when $i$ or $j$ are at the start of the text
or $j-1\not\in M_{i-1}$:
\begin{align}\label{eqBMSStartingRef}
    \forall i\in[1,n], j \in M_i \mbox{ s.t. } i = 1 \mbox{ or } j = 1 \mbox{ or } T[i-1]\neq T[j-1]:
    \refpos{i\rightarrow j}\implies \fbeg{i}
\end{align}
The second case is when $j-1\in M_{i-1}$
but the position~$i-1$ does not reference position $j-1$ (it may reference a different position, or it could be a ground phrase):
\begin{align}\label{eqBMSDifferentRef}
    \begin{split}
        \forall i\in[2,n], \forall j \in \match{i} \text{~s.t.~} j>1 \mbox{ and }
        T[i-1]=T[j-1],\\
        \lnot \refpos{i-1\rightarrow j-1} \land \refpos{i\rightarrow j} \implies \fbeg{i}
    \end{split}
\end{align}

In total, we have $O(n^3)$ Boolean variables, dominated by $\dref{d,i\rightarrow j}$.
The size of the largest clause is $O(n^2)$ due to Constraint~(\ref{eqBMSDref}).
The total size of the resulting CNF is $O(n^4)$,
dominated by Constraint~(\ref{eqBMSTransitiveRef}) where there are $O(n^3)$ clauses of $O(n)$ size each.

\paragraph*{Correctness of the Encoding}
It is easy to see that any valid BMS satisfies the above constraints.
We now show that any solution that satisfies the hard clauses yields a valid BMS\@.
The truth assignments for all $p_i$ define a factorization of $T$.
We claim that each position is either a ground phrase,
or is assigned exactly one reference consistent with the factorization forming a valid BMS\@,
i.e., the references are acyclic, and, adjacent positions in the same non-ground phrase will refer to adjacent positions thus allowing the phrase to be encoded with the pair of references at both ends of the phrase.

Suppose $p_i = 1$.
If $\rot{i}=1$, then Constraint~(\ref{eqBMSDref}) ensures that
all $\dref{\cdot,i\rightarrow\cdot}=0$ and
Constraint~(\ref{eqBMSGroundPhrase}) ensures that all $\refpos{i\rightarrow\cdot}=0$,
i.e., $i$ does not have a reference.
Note that, $p_{i+1}=0$ implies $\refpos{i\rightarrow j}$ for some $j$ (shown in the next paragraph),
so $p_{i+1}=1$ must hold. Thus, position $i$ is properly factorized as a ground phrase.
If $\rot{i}=0$, then
Constraint~(\ref{eqBMSDref}) ensures that there exist unique $d,j$ such that $\dref{d,i\rightarrow j}=1$.
Furthermore, Constraint~(\ref{eqBMSOnlyOneParent}) ensures that $\refpos{i\rightarrow j}=1$.

Next, consider the case for $p_i = 0$ (which implies $i>2$).
From Constraint~(\ref{eqBMSGroundPhrase}) we have
$\rot{i} = 0$,
and from Constraint~(\ref{eqBMSDifferentRef})
we have $\forall j>1\in M_i$ s.t. $T[i-1] = T[j-1]$, $\refpos{i\rightarrow j}\implies \refpos{{i-1}\rightarrow{j-1}}$.
Since $\rot{i} = 0$,
Constraint~(\ref{eqBMSDref}) ensures that there exists unique $d,j$ such that $\dref{d,i\rightarrow j}=1$.
Furthermore, Constraint~(\ref{eqBMSOnlyOneParent})
ensures that $\refpos{i\rightarrow j}=1$.
Note that due to Constraint~(\ref{eqBMSStartingRef}),
neither $j = 1$ nor $T[i-1]\neq T[j-1]$ is possible, since this would imply $p_i=1$, contradicting the assumption that $p_i=0$.
Thus $j > 1$ and $T[i-1]=T[j-1]$, and thus we have $\refpos{{i-1}\rightarrow{j-1}}=1$.

The uniqueness of the reference $j$ for each position $i$ of a non-ground phrase is ensured by Constraint~(\ref{eqBMSAtMostOneRef}).
Thus, we have that references in adjacent positions in the same non-ground phrase point to adjacent positions.
Since the acyclicity of the references are ensured by Constraint~(\ref{eqBMSTransitiveRef}),
we have a valid BMS\@.

\section{Computational Experiments}\label{secExperiments}

We have implemented our encodings in PySAT~(\url{https://pysathq.github.io/}) written in the Python language\footnote{As far as we are aware of, this implementation is single threaded.}.
As datasets we used the files \textsc{trans}, \textsc{news}, \textsc{E.coli}, and \textsc{progc} from the Canterbury and Calgary corpus~(\url{https://corpus.canterbury.ac.nz/}).

Here, we evaluated the sum of the literals in all hard clauses, i.e., the size of the encoded CNF, and the execution time of the SAT solver for computing a solution.
In \cref{figPlots}, we evaluated our approach on different prefix lengths of the chosen datasets, starting from a prefix of 10 characters up to a prefix with 3000 characters.
We aborted an execution after reaching one hour of computation or after exceeding 16 GB of RAM, and hence the lines for computing $b$ and $g$ prematurely end due to these limits on all datasets.
Our experiments ran on an Ubuntu 20.04
machine with an AMD Ryzen Threadripper 3990X CPU.

As expected, the size of the encoded CNF correlates with the execution time in all instances.
We can see that the encoding for $\gamma$ needs the least number of literals, and is consequently not only the fastest, but also uses the least amount of memory,
allowing us to compute $\gamma$ for moderately large texts.
This is followed by $g$, and lastly by $b$.
Although the size of the CNF for $b$ is smaller than for $g$ in most cases,
clauses formed by Constraint~(\ref{eqBMSDref}) for computing $b$ can become quite large,
making the computation cumbersome.

\begin{figure}[h]
    \centerline{
        \includegraphics[page=3,width=0.9\textwidth]{plot/miniplot.pdf}
    }
    \begin{minipage}{0.75\textwidth}
        \caption{Evaluation of our encoded CNFs.
            The first row shows the running time of PySAT on our CNF instance in seconds.
            We omit the time needed to specify the CNFs, which is negligible for larger instances.
            The second row plots the size of the respective CNF.
            All axes are in logscale.
        }
        \label{figPlots}
    \end{minipage}\hfill
    \begin{minipage}{0.24\textwidth}
        \includegraphics[page=1]{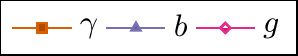}
    \end{minipage}
\end{figure}

\section{Application: Sensitivity of \texorpdfstring{$\gamma$}{Gamma}}
Akagi et al.~\cite{DBLP:journals/corr/abs-2107-08615} introduced and studied the
notion of \teigi{sensitivity} of a repetitiveness measure.
Given a repetitiveness measure $C$ (such as $\gamma$) for a string~$T$, the sensitivity of $C$ measures how much $C$ can increase when a single character edit operation is performed on $T$.
The authors studied an additive and a multiplicative sensitivity measure.
The latter, denoted $\mathit{MS}_\mathit{op}$, is defined as:
\[
    \mathit{MS}_\mathit{op}(C,n) := \max_{T\in\Sigma^n,T'\in\Sigma^*} \left\{ \frac{C(T')}{C(T)} \,\middle\vert\, \mathit{ed}_\mathit{op}(T,T')=1 \right\},
\]
i.e., the maximum multiplicative increase over all strings with the same length~$n$,
where
$\mathit{ed}_\mathit{op}(T,T')=1$ means that $T'$ can be built from $T$ by inserting a character into $T$, or deleting/replacing a character of~$T$.
Parameterizing $\gamma$ with the input string~$T$, for $C(T) = \gamma(T)$,
Akagi et al. showed $2 \le \mathit{MS}_\mathit{op}(\gamma,n) \in O(\log n)$.

To improve the lower bound, we conducted exhaustive search for short binary strings when inserting a unique character.
This search led us to
the string family
$\{T_k\}_{k \geq 2}$ with $T_k :=
    \mathtt{abbb{a}aa{b}}^{k}$,
with which we can improve the lower bound of $2$ to $5/2$.
For that, let us consider $\gamma(T_k)$ and its size after an insertion of a new character $\mathtt{c}$.
First, we observe that
$\gamma(T_k) =
    \gamma(\mathtt{abbb\underline{a}aa\underline{b}}^{k}) = 2$.
This is because a smallest string attractor is given by
$\Gamma(T_2) = \{ 4,7 \}$ and $\Gamma(T_k) = \{5,8\}$ for $k \ge 3$ (the characters at the positions in $\Gamma(T_k)$ are underlined).
Now let $T'_k$ denote $T_k$ after inserting the character $\mathtt{c}$ at text position~$9$.
For $k \ge 5$, it holds that $T'_k$ has a string attractor of size~$5$, i.e.,
$\gamma(T'_{k'}) =
    \gamma(\mathtt{\underline{a}bb\underline{b}a\underline{a}ab\underline{c}\underline{b}}^{k'}) = 5$ for $k' \ge 4$.
A minimal string attractor is given by $\Gamma(T'_{k'}) = \{1,4,6,9,10\}$.
We cannot remove a position from $\Gamma(T'_{k'})$ since
$\mathtt{abb},\mathtt{ba},\mathtt{aab},\mathtt{c}$, and $\mathtt{b}^{k'}$ are five substrings of $T'_{k'}$ having exactly one occurrence in $T'_{k'}$, and all of them are non-overlapping.
Since a string attractor has to be in the cover set of all substrings, we need a string attractor with at least five text positions.
Consequently, $\mathit{MS}_\mathit{op}(\gamma,n) \geq 2.5$ for any $n\geq 13$
with the insertion or replacement operation.

The availability of computer-aided search facilitated the discovery of strings having certain string attractors.

\bibliography{ref}

\begin{thebibliography}{10}

\bibitem{DBLP:journals/corr/abs-2107-08615}
Tooru Akagi, Mitsuru Funakoshi, and Shunsuke Inenaga.
\newblock Sensitivity of string compressors and repetitiveness measures.
\newblock {\em CoRR}, abs/2107.08615, 2021.
\newblock URL: \url{https://arxiv.org/abs/2107.08615}, \href
  {http://arxiv.org/abs/2107.08615} {\path{arXiv:2107.08615}}.

\bibitem{DBLP:conf/spire/BannaiFIKMN21}
Hideo Bannai, Mitsuru Funakoshi, Tomohiro I, Dominik K{\"{o}}ppl, Takuya Mieno,
  and Takaaki Nishimoto.
\newblock A separation of {\(\gamma\)} and b via {Thue-Morse} words.
\newblock In {\em Proc.\ {SPIRE}}, volume 12944, pages 167--178, 2021.
\newblock \href {https://doi.org/10.1007/978-3-030-86692-1\_14}
  {\path{doi:10.1007/978-3-030-86692-1\_14}}.

\bibitem{biere2009handbook}
Armin Biere, Marijn Heule, and Hans van Maaren.
\newblock {\em Handbook of satisfiability}, volume 185.
\newblock IOS press, 2009.

\bibitem{DBLP:journals/jda/BilleGGP18}
Philip Bille, Travis Gagie, Inge~Li G{\o}rtz, and Nicola Prezza.
\newblock A separation between {RLSLP}s and {LZ77}.
\newblock {\em J. Discrete Algorithms}, 50:36--39, 2018.
\newblock \href {https://doi.org/10.1016/j.jda.2018.09.002}
  {\path{doi:10.1016/j.jda.2018.09.002}}.

\bibitem{DBLP:journals/jacm/BlumerBHME87}
Anselm Blumer, J.~Blumer, David Haussler, Ross~M. McConnell, and Andrzej
  Ehrenfeucht.
\newblock Complete inverted files for efficient text retrieval and analysis.
\newblock {\em J. {ACM}}, 34(3):578--595, 1987.
\newblock \href {https://doi.org/10.1145/28869.28873}
  {\path{doi:10.1145/28869.28873}}.

\bibitem{DBLP:journals/mst/CaselFGGS21}
Katrin Casel, Henning Fernau, Serge Gaspers, Benjamin Gras, and Markus~L.
  Schmid.
\newblock On the complexity of the smallest grammar problem over fixed
  alphabets.
\newblock {\em Theory Comput. Syst.}, 65(2):344--409, 2021.
\newblock \href {https://doi.org/10.1007/s00224-020-10013-w}
  {\path{doi:10.1007/s00224-020-10013-w}}.

\bibitem{christiansen21optimaltime}
Anders~Roy Christiansen, Mikko~Berggren Ettienne, Tomasz Kociumaka, Gonzalo
  Navarro, and Nicola Prezza.
\newblock Optimal-time dictionary-compressed indexes.
\newblock {\em {ACM} Trans. Algorithms}, 17(1):8:1--8:39, 2021.
\newblock \href {https://doi.org/10.1145/3426473} {\path{doi:10.1145/3426473}}.

\bibitem{crochemore08lpf}
Maxime Crochemore and Lucian Ilie.
\newblock Computing longest previous factor in linear time and applications.
\newblock {\em Inf. Process. Lett.}, 106(2):75--80, 2008.
\newblock \href {https://doi.org/10.1016/j.ipl.2007.10.006}
  {\path{doi:10.1016/j.ipl.2007.10.006}}.

\bibitem{crochemore11computing}
Maxime Crochemore and German Tischler.
\newblock Computing longest previous non-overlapping factors.
\newblock {\em Inf. Process. Lett.}, 111(6):291--295, 2011.
\newblock \href {https://doi.org/10.1016/j.ipl.2010.12.005}
  {\path{doi:10.1016/j.ipl.2010.12.005}}.

\bibitem{dinklage19plcpcomp}
Patrick Dinklage, Jonas Ellert, Johannes Fischer, Dominik K{\"{o}}ppl, and
  Manuel Penschuck.
\newblock Bidirectional text compression in external memory.
\newblock In {\em Proc.\ ESA}, pages 41:1--41:16, 2019.
\newblock \href {https://doi.org/10.4230/LIPIcs.ESA.2019.41}
  {\path{doi:10.4230/LIPIcs.ESA.2019.41}}.

\bibitem{dinklage2017compression}
Patrick Dinklage, Johannes Fischer, Dominik K{\"{o}}ppl, Marvin L{\"{o}}bel,
  and Kunihiko Sadakane.
\newblock Compression with the tudocomp framework.
\newblock In {\em Proc.\ SEA}, volume~75 of {\em LIPIcs}, pages 13:1--13:22,
  2017.

\bibitem{DBLP:conf/cpm/GotoBIT15}
Keisuke Goto, Hideo Bannai, Shunsuke Inenaga, and Masayuki Takeda.
\newblock {LZD} factorization: Simple and practical online grammar compression
  with variable-to-fixed encoding.
\newblock In {\em Proc.\ {CPM}}, volume 9133, pages 219--230, 2015.
\newblock \href {https://doi.org/10.1007/978-3-319-19929-0\_19}
  {\path{doi:10.1007/978-3-319-19929-0\_19}}.

\bibitem{oeis.org_A339391}
OEIS~Foundation Inc.
\newblock Maximum, over all binary strings $w$ of length $n$, of the size of
  the smallest string attractor for $w$, entry {A339391} in the on-line
  encyclopedia of integer sequences.
\newblock Accessed: 2022-04-13.
\newblock URL: \url{https://oeis.org/A339391}.

\bibitem{DBLP:journals/njc/KarpinskiRS97}
Marek Karpinski, Wojciech Rytter, and Ayumi Shinohara.
\newblock An efficient pattern-matching algorithm for strings with short
  descriptions.
\newblock {\em Nord. J. Comput.}, 4(2):172--186, 1997.

\bibitem{DBLP:conf/focs/KempaK20}
Dominik Kempa and Tomasz Kociumaka.
\newblock Resolution of the {B}urrows-{W}heeler transform conjecture.
\newblock In Sandy Irani, editor, {\em Proc.\ {FOCS}}, pages 1002--1013.
  {IEEE}, 2020.
\newblock \href {https://doi.org/10.1109/FOCS46700.2020.00097}
  {\path{doi:10.1109/FOCS46700.2020.00097}}.

\bibitem{kempa18string}
Dominik Kempa, Alberto Policriti, Nicola Prezza, and Eva Rotenberg.
\newblock String attractors: Verification and optimization.
\newblock In {\em Proc.\ ESA}, pages 52:1--52:13, 2018.
\newblock \href {https://doi.org/10.4230/LIPIcs.ESA.2018.52}
  {\path{doi:10.4230/LIPIcs.ESA.2018.52}}.

\bibitem{KempaP18}
Dominik Kempa and Nicola Prezza.
\newblock At the roots of dictionary compression: string attractors.
\newblock In {\em Proc.\ {STOC}}, pages 827--840. {ACM}, 2018.
\newblock \href {https://doi.org/10.1145/3188745.3188814}
  {\path{doi:10.1145/3188745.3188814}}.

\bibitem{doi:10.1137/1.9781611977073.111}
Dominik Kempa and Barna Saha.
\newblock An upper bound and linear-space queries on the lz-end parsing.
\newblock In {\em Proceedings of the 2022 {ACM-SIAM} Symposium on Discrete
  Algorithms, {SODA} 2022, Virtual Conference / Alexandria, VA, USA, January 9
  - 12, 2022}, pages 2847--2866. {SIAM}, 2022.
\newblock \href {https://doi.org/10.1137/1.9781611977073.111}
  {\path{doi:10.1137/1.9781611977073.111}}.

\bibitem{10.5555/2898950}
Donald~E. Knuth.
\newblock {\em The Art of Computer Programming, Volume 4, Fascicle 6:
  Satisfiability}.
\newblock Addison-Wesley Professional, 1st edition, 2015.

\bibitem{kociumaka20towards}
Tomasz Kociumaka, Gonzalo Navarro, and Nicola Prezza.
\newblock Towards a definitive measure of repetitiveness.
\newblock In {\em Proc.\ {LATIN}}, pages 207--219, 2020.

\bibitem{DBLP:journals/tcs/KreftN13}
Sebastian Kreft and Gonzalo Navarro.
\newblock On compressing and indexing repetitive sequences.
\newblock {\em Theor. Comput. Sci.}, 483:115--133, 2013.
\newblock \href {https://doi.org/10.1016/j.tcs.2012.02.006}
  {\path{doi:10.1016/j.tcs.2012.02.006}}.

\bibitem{DBLP:conf/spire/KutsukakeMNIBT20}
Kanaru Kutsukake, Takuya Matsumoto, Yuto Nakashima, Shunsuke Inenaga, Hideo
  Bannai, and Masayuki Takeda.
\newblock On repetitiveness measures of {T}hue-{M}orse words.
\newblock In {\em Proc.\ {SPIRE}}, pages 213--220, 2020.
\newblock \href {https://doi.org/10.1007/978-3-030-59212-7\_15}
  {\path{doi:10.1007/978-3-030-59212-7\_15}}.

\bibitem{larsson99repair}
N.~Jesper Larsson and Alistair Moffat.
\newblock Offline dictionary-based compression.
\newblock In {\em Proc.\ DCC}, pages 296--305, 1999.
\newblock \href {https://doi.org/10.1109/DCC.1999.755679}
  {\path{doi:10.1109/DCC.1999.755679}}.

\bibitem{lempel1976complexity}
Abraham Lempel and Jacob Ziv.
\newblock On the complexity of finite sequences.
\newblock {\em IEEE Transactions on information theory}, 22(1):75--81, 1976.

\bibitem{DBLP:series/txcs/LiV19}
Ming Li and Paul M.~B. Vit{\'{a}}nyi.
\newblock {\em An Introduction to Kolmogorov Complexity and Its Applications,
  4th Edition}.
\newblock Texts in Computer Science. Springer, 2019.
\newblock \href {https://doi.org/10.1007/978-3-030-11298-1}
  {\path{doi:10.1007/978-3-030-11298-1}}.

\bibitem{DBLP:journals/njc/MakinenN05}
Veli M{\"{a}}kinen and Gonzalo Navarro.
\newblock Succinct suffix arrays based on run-length encoding.
\newblock {\em Nord. J. Comput.}, 12(1):40--66, 2005.

\bibitem{DBLP:journals/tcs/MantaciRRRS21}
Sabrina Mantaci, Antonio Restivo, Giuseppe Romana, Giovanna Rosone, and
  Marinella Sciortino.
\newblock A combinatorial view on string attractors.
\newblock {\em Theor. Comput. Sci.}, 850:236--248, 2021.
\newblock \href {https://doi.org/10.1016/j.tcs.2020.11.006}
  {\path{doi:10.1016/j.tcs.2020.11.006}}.

\bibitem{DBLP:journals/ipl/MantaciRS03}
Sabrina Mantaci, Antonio Restivo, and Marinella Sciortino.
\newblock {B}urrows-{W}heeler transform and sturmian words.
\newblock {\em Inf. Process. Lett.}, 86(5):241--246, 2003.
\newblock \href {https://doi.org/10.1016/S0020-0190(02)00512-4}
  {\path{doi:10.1016/S0020-0190(02)00512-4}}.

\bibitem{DBLP:journals/corr/abs-2202-08447}
Takuya Mieno, Shunsuke Inenaga, and Takashi Horiyama.
\newblock Repair grammars are the smallest grammars for fibonacci words.
\newblock {\em CoRR}, abs/2202.08447, 2022.
\newblock URL: \url{https://arxiv.org/abs/2202.08447}, \href
  {http://arxiv.org/abs/2202.08447} {\path{arXiv:2202.08447}}.

\bibitem{DBLP:journals/algorithmica/NarisawaHIBT17}
Kazuyuki Narisawa, Hideharu Hiratsuka, Shunsuke Inenaga, Hideo Bannai, and
  Masayuki Takeda.
\newblock Efficient computation of substring equivalence classes with suffix
  arrays.
\newblock {\em Algorithmica}, 79(2):291--318, 2017.
\newblock \href {https://doi.org/10.1007/s00453-016-0178-z}
  {\path{doi:10.1007/s00453-016-0178-z}}.

\bibitem{navarro2021indexing}
Gonzalo Navarro.
\newblock Indexing highly repetitive string collections, part {I:}
  repetitiveness measures.
\newblock {\em {ACM} Comput. Surv.}, 54(2):29:1--29:31, 2021.
\newblock \href {https://doi.org/10.1145/3434399} {\path{doi:10.1145/3434399}}.

\bibitem{DBLP:journals/csur/Navarro21}
Gonzalo Navarro.
\newblock Indexing highly repetitive string collections, part {II:} compressed
  indexes.
\newblock {\em {ACM} Comput. Surv.}, 54(2):26:1--26:32, 2021.
\newblock \href {https://doi.org/10.1145/3432999} {\path{doi:10.1145/3432999}}.

\bibitem{navarro2020approximation}
Gonzalo Navarro, Carlos Ochoa, and Nicola Prezza.
\newblock On the approximation ratio of ordered parsings.
\newblock {\em IEEE Transactions on Information Theory}, 67(2):1008--1026,
  2020.

\bibitem{navarro19indexing}
Gonzalo Navarro and Nicola Prezza.
\newblock Universal compressed text indexing.
\newblock {\em Theor. Comput. Sci.}, 762:41--50, 2019.
\newblock \href {https://doi.org/10.1016/j.tcs.2018.09.007}
  {\path{doi:10.1016/j.tcs.2018.09.007}}.

\bibitem{DBLP:journals/jair/Nevill-ManningW97}
Craig~G. Nevill{-}Manning and Ian~H. Witten.
\newblock Identifying hierarchical structure in sequences: {A} linear-time
  algorithm.
\newblock {\em J. Artif. Intell. Res.}, 7:67--82, 1997.
\newblock \href {https://doi.org/10.1613/jair.374}
  {\path{doi:10.1613/jair.374}}.

\bibitem{nishimoto2021lzrr}
Takaaki Nishimoto and Yasuo Tabei.
\newblock {LZRR: LZ77} parsing with right reference.
\newblock {\em Information and Computation}, page 104859, 2021.

\bibitem{DBLP:conf/dcc/RussoCNF20}
Lu{\'{\i}}s M.~S. Russo, Ana Sofia~D. Correia, Gonzalo Navarro, and
  Alexandre~P. Francisco.
\newblock Approximating optimal bidirectional macro schemes.
\newblock In Ali Bilgin, Michael~W. Marcellin, Joan Serra{-}Sagrist{\`{a}}, and
  James~A. Storer, editors, {\em Data Compression Conference, {DCC} 2020,
  Snowbird, UT, USA, March 24-27, 2020}, pages 153--162. {IEEE}, 2020.
\newblock \href {https://doi.org/10.1109/DCC47342.2020.00023}
  {\path{doi:10.1109/DCC47342.2020.00023}}.

\bibitem{DBLP:journals/tcs/Rytter03}
Wojciech Rytter.
\newblock Application of {Lempel-Ziv} factorization to the approximation of
  grammar-based compression.
\newblock {\em Theor. Comput. Sci.}, 302(1-3):211--222, 2003.
\newblock \href {https://doi.org/10.1016/S0304-3975(02)00777-6}
  {\path{doi:10.1016/S0304-3975(02)00777-6}}.

\bibitem{DBLP:conf/spire/SakamotoKS04}
Hiroshi Sakamoto, Takuya Kida, and Shinichi Shimozono.
\newblock A space-saving linear-time algorithm for grammar-based compression.
\newblock In {\em Proc.\ {SPIRE}}, volume 3246, pages 218--229, 2004.
\newblock \href {https://doi.org/10.1007/978-3-540-30213-1\_33}
  {\path{doi:10.1007/978-3-540-30213-1\_33}}.

\bibitem{sakamoto02minimization}
Hiroshi Sakamoto, Shinichi Shimozono, Ayumi Shinohara, and Masayuki Takeda.
\newblock On the minimization problem of text compression scheme by a reduced
  grammar transform.
\newblock Technical Report 195, Department of Informatics, 2001.
\newblock URL:
  \url{https://catalog.lib.kyushu-u.ac.jp/opac_download_md/3045/trcs195.pdf}.

\bibitem{DBLP:journals/corr/abs-2012-06840}
Luke Schaeffer and Jeffrey Shallit.
\newblock String attractors for automatic sequences.
\newblock {\em CoRR}, abs/2012.06840, 2020.
\newblock URL: \url{https://arxiv.org/abs/2012.06840}, \href
  {http://arxiv.org/abs/2012.06840} {\path{arXiv:2012.06840}}.

\bibitem{DBLP:conf/cp/Sinz05}
Carsten Sinz.
\newblock Towards an optimal {CNF} encoding of boolean cardinality constraints.
\newblock In Peter van Beek, editor, {\em Principles and Practice of Constraint
  Programming - {CP} 2005, 11th International Conference, {CP} 2005, Sitges,
  Spain, October 1-5, 2005, Proceedings}, volume 3709 of {\em Lecture Notes in
  Computer Science}, pages 827--831. Springer, 2005.
\newblock \href {https://doi.org/10.1007/11564751\_73}
  {\path{doi:10.1007/11564751\_73}}.

\bibitem{storer1982data}
James~A Storer and Thomas~G Szymanski.
\newblock Data compression via textual substitution.
\newblock {\em Journal of the ACM (JACM)}, 29(4):928--951, 1982.

\bibitem{ziv77lz}
Jacob Ziv and Abraham Lempel.
\newblock A universal algorithm for sequential data compression.
\newblock {\em {IEEE} Trans. Information Theory}, 23(3):337--343, 1977.

\bibitem{DBLP:journals/tit/ZivL78}
Jacob Ziv and Abraham Lempel.
\newblock Compression of individual sequences via variable-rate coding.
\newblock {\em {IEEE} Trans. Inf. Theory}, 24(5):530--536, 1978.
\newblock \href {https://doi.org/10.1109/TIT.1978.1055934}
  {\path{doi:10.1109/TIT.1978.1055934}}.

\end{thebibliography}
\clearpage
\appendix{}

\section{Minimal Substrings and Right-Minimal Substrings}
\label{appendix:minsubstr}
Table~\ref{tab:min-substr} shows statistics on the number and total lengths of minimal substrings, and right-minimal substrings discussed in the footnote of \cref{sec:minsub},
which directly correspond to the total size of the hard clauses for computing
the smallest string attractor.
\begin{table}[h]
    \caption{Number of minimal substrings of strings in the Calgary Corpus. \textrm{\#rmin} and \textrm{total\_rmin} are the number of right-minimal substrings and their total length, respectively. \textrm{\#lrmin} and \textrm{total\_lrmin} are the number of minimal substrings and their total length, respectively.}
    \label{tab:min-substr}
    \centering
    \begin{tabular}{lrrrrrrr}                       \\
        file             & $n$ &
        \textrm{\#lrmin} &
        \textrm{\#rmin}  &
        \parbox{1.3cm}{\textrm{\#lrmin}/ \\\textrm{\#rmin}} & \textrm{total\_lrmin}  &  \textrm{total\_rmin} & \parbox{1.7cm}{\textrm{total\_lrmin}/\\\textrm{total\_rmin}} \\
        bib & \num{111261} & \num{46197} & \num{171083} & \num{0.2700} & \num{285617} & \num{2361086} & \num{0.1210} \\
book1 & \num{768771} & \num{500936} & \num{1154048} & \num{0.4341} & \num{3532043} & \num{9958234} & \num{0.3547} \\
book2 & \num{610856} & \num{313379} & \num{935376} & \num{0.3350} & \num{2272154} & \num{10461032} & \num{0.2172} \\
geo & \num{102400} & \num{68169} & \num{130104} & \num{0.5240} & \num{241099} & \num{609743} & \num{0.3954} \\
news & \num{377109} & \num{189274} & \num{573175} & \num{0.3302} & \num{1123327} & \num{12639205} & \num{0.0889} \\
obj1 & \num{21504} & \num{11469} & \num{27540} & \num{0.4164} & \num{543281} & \num{1258644} & \num{0.4316} \\
obj2 & \num{246814} & \num{90807} & \num{380169} & \num{0.2389} & \num{460038} & \num{8416783} & \num{0.0547} \\
paper1 & \num{53161} & \num{27795} & \num{82189} & \num{0.3382} & \num{146463} & \num{816030} & \num{0.1795} \\
paper2 & \num{82199} & \num{46594} & \num{125407} & \num{0.3715} & \num{265472} & \num{1071864} & \num{0.2477} \\
paper3 & \num{46526} & \num{27979} & \num{70443} & \num{0.3972} & \num{141373} & \num{516646} & \num{0.2736} \\
paper4 & \num{13286} & \num{8355} & \num{20155} & \num{0.4145} & \num{35540} & \num{135281} & \num{0.2627} \\
paper5 & \num{11954} & \num{7080} & \num{18173} & \num{0.3896} & \num{28449} & \num{122620} & \num{0.2320} \\
paper6 & \num{38105} & \num{19923} & \num{59184} & \num{0.3366} & \num{99148} & \num{660580} & \num{0.1501} \\
pic & \num{513215} & \num{126267} & \num{763557} & \num{0.1654} & \num{661361467} & \num{1303736731} & \num{0.5073} \\
progc & \num{39611} & \num{19008} & \num{60777} & \num{0.3127} & \num{94940} & \num{616961} & \num{0.1539} \\
progl & \num{71646} & \num{25530} & \num{118137} & \num{0.2161} & \num{171982} & \num{3359343} & \num{0.0512} \\
progp & \num{49379} & \num{16323} & \num{82443} & \num{0.1980} & \num{99494} & \num{5769774} & \num{0.0172} \\
trans & \num{93695} & \num{25532} & \num{160087} & \num{0.1595} & \num{211012} & \num{10608211} & \num{0.0199} \\     \end{tabular}
\end{table}
\section{Elaborated Evaluation}\label{secAppExperiments}
Here, we present an extended benchmark to \cref{secExperiments}
using the datasets from the Canterbury corpus\footnote{\url{https://corpus.canterbury.ac.nz/}},
the Calgary corpus, and the four morphic word sequences Fibonacci
($F_0 = \mathtt{a}, F_1 = \mathtt{ab}, F_k = F_{k-1}F_{k-2}~\forall k \ge 2$),
period-doubling
($P_0 = \mathtt{a}, P_k = P_{k-1}[1..|P_{k-1}-1|]\overline{P_{k-1}[|P_{k-1}|]}$, with $\overline{\mathtt{a}} = \mathtt{b}$ and $\overline{\mathtt{b}} = \mathtt{a}$),
Thue--Morse $(T_0 = \mathtt{a}, T_k = T_{k-1}\overline{T_{k-1}}~\forall k \ge 1$), and
paper-folding ($A_0 = 11, A_k = \varphi(A_{k-1})$ with
$ \varphi(\mathtt{11}) = \mathtt{1101}, \varphi(\mathtt{01}) = \mathtt{1001}, \varphi(\mathtt{10}) = \mathtt{1100}, \varphi(\mathtt{00}) = \mathtt{1000}, $).
We evaluated the time (\cref{tab:time}),
the number of variables (\cref{tab:nvars}),
the number of hard clauses (\cref{tab:nhard}),
the number of literals in the largest clause (\cref{tab:sol_nmaxclause}), and
the sum of the literals in all clauses, i.e., the size of the encoded CNF (\cref{tab:sol_ntotalvars}).
The number of soft clauses is omitted --- this number is always equal to the text length.
We aborted an execution after reaching one hour of computation (no-time) or after exceeding 16 GB of RAM (no-mem).\footnote{The error ``unknown'' is caused for $g$ on datasets of length 1. This minor bug has subsequently been fixed.}
Unfortunately, we can only compute only all compression characteristics on some of the morphic words -- we are only able to compute $\gamma$ on the other datasets.
Finally, we complement our plots shown in \cref{figPlots} with a detailed evaluation for every dataset of the Canterbury and Calgary corpus.
In Figures~\ref{plot:aaa.txt} to~\ref{plot:xargs.1} we measured the total execution time of our program,
the time needed for the SAT solver (like in \cref{figPlots}),
the output size (i.e., $\gamma$, $b$, or~$g$),
the size of the encoded CNF,
the maximum clause size, i.e., the maximum number of literals a clause in our CNF attains,
and finally the number of hard clauses.
Like already hinted in \cref{secExperiments}, the slow execution times for computing~$b$
can likely be linked to the fact that our encoding for $b$ always needs the largest number of hard clauses,
and some of them are much larger than the clauses for $g$ or $\gamma$.
The precomputation time for $g$ is non-negligible since we check
Constraint~\ref{constraint:firstoccurrence_not_factor} with a longest non-overlapping factor table~\cite{crochemore11computing},
which we compute naively.
\def\justbeingincluded{justbeingincluded}

\begin{longtable}{l*{4}{r}{c}}
	\caption{Time in Seconds}
	\label{tab:time}
	\\
	file & $n$ & attractor~$\gamma$ & BMS~$b$ & SLCP~$g$ \\
alice29.txt & \num{152089} & \num{537.99} & no mem & no mem \\
asyoulik.txt & \num{125179} & \num{452.41} & no mem & no mem \\
bib & \num{111261} & \num{224.11} & no mem & no mem \\
cp.html & \num{24603} & \num{12.84} & no mem & no mem \\
fibonacci.00 & \num{1} & \num{0.00} & \num{0.00} & unknown \\
fibonacci.01 & \num{2} & \num{0.00} & \num{0.00} & \num{0.00} \\
fibonacci.02 & \num{3} & \num{0.00} & \num{0.00} & \num{0.00} \\
fibonacci.03 & \num{5} & \num{0.00} & \num{0.00} & \num{0.00} \\
fibonacci.04 & \num{8} & \num{0.00} & \num{0.00} & \num{0.00} \\
fibonacci.05 & \num{13} & \num{0.00} & \num{0.01} & \num{0.01} \\
fibonacci.06 & \num{21} & \num{0.00} & \num{0.06} & \num{0.02} \\
fibonacci.07 & \num{34} & \num{0.00} & \num{0.27} & \num{0.07} \\
fibonacci.08 & \num{55} & \num{0.00} & \num{1.21} & \num{0.37} \\
fibonacci.09 & \num{89} & \num{0.00} & \num{5.71} & \num{1.75} \\
fibonacci.10 & \num{144} & \num{0.00} & \num{25.94} & \num{10.71} \\
fibonacci.11 & \num{233} & \num{0.01} & \num{116.40} & \num{59.09} \\
fibonacci.12 & \num{377} & \num{0.01} & no mem & no mem \\
fibonacci.13 & \num{610} & \num{0.01} & no mem & no mem \\
fibonacci.14 & \num{987} & \num{0.02} & no mem & no mem \\
fibonacci.15 & \num{1597} & \num{0.04} & no mem & no mem \\
fibonacci.16 & \num{2584} & \num{0.07} & no mem & no mem \\
fibonacci.17 & \num{4181} & \num{0.15} & no mem & no mem \\
fibonacci.18 & \num{6765} & \num{0.24} & no mem & no mem \\
fibonacci.19 & \num{10946} & \num{0.49} & no mem & no mem \\
fibonacci.20 & \num{17711} & \num{2.87} & no mem & no mem \\
fields.c & \num{11150} & \num{2.55} & no mem & no mem \\
geo & \num{102400} & \num{403.34} & no mem & no mem \\
grammar.lsp & \num{3721} & \num{0.40} & no mem & no mem \\
obj1 & \num{21504} & \num{35.30} & no mem & no mem \\
paper1 & \num{53161} & \num{136.13} & no mem & no mem \\
paper2 & \num{82199} & \num{158.26} & no mem & no mem \\
paper3 & \num{46526} & \num{51.84} & no mem & no mem \\
paper4 & \num{13286} & \num{5.01} & no mem & no mem \\
paper5 & \num{11954} & \num{3.88} & no mem & no mem \\
paper6 & \num{38105} & \num{31.04} & no mem & no mem \\
paperfold.00 & \num{2} & \num{0.00} & \num{0.00} & \num{0.00} \\
paperfold.01 & \num{4} & \num{0.00} & \num{0.00} & \num{0.00} \\
paperfold.02 & \num{8} & \num{0.00} & \num{0.00} & \num{0.00} \\
paperfold.03 & \num{16} & \num{0.00} & \num{0.02} & \num{0.01} \\
paperfold.04 & \num{32} & \num{0.00} & \num{0.20} & \num{0.03} \\
paperfold.05 & \num{64} & \num{0.00} & \num{3.78} & \num{0.26} \\
paperfold.06 & \num{128} & \num{0.00} & no time & \num{2.12} \\
paperfold.07 & \num{256} & \num{0.01} & no time & \num{20.19} \\
paperfold.08 & \num{512} & \num{0.03} & no mem & \num{254.18} \\
paperfold.09 & \num{1024} & \num{0.54} & no mem & no mem \\
paperfold.10 & \num{2048} & \num{1.08} & no mem & no mem \\
paperfold.11 & \num{4096} & \num{9.34} & no mem & no mem \\
paperfold.12 & \num{8192} & \num{15.73} & no mem & no mem \\
paperfold.13 & \num{16384} & \num{378.53} & no mem & no mem \\
paperfold.14 & \num{32768} & \num{2543.88} & no mem & no mem \\
perioddoubling.00 & \num{1} & \num{0.00} & \num{0.00} & unknown \\
perioddoubling.01 & \num{2} & \num{0.00} & \num{0.00} & \num{0.00} \\
perioddoubling.02 & \num{4} & \num{0.00} & \num{0.00} & \num{0.00} \\
perioddoubling.03 & \num{8} & \num{0.00} & \num{0.00} & \num{0.00} \\
perioddoubling.04 & \num{16} & \num{0.00} & \num{0.04} & \num{0.01} \\
perioddoubling.05 & \num{32} & \num{0.00} & \num{1.89} & \num{0.04} \\
perioddoubling.06 & \num{64} & \num{0.00} & \num{38.01} & \num{0.42} \\
perioddoubling.07 & \num{128} & \num{0.00} & no time & \num{4.32} \\
perioddoubling.08 & \num{256} & \num{0.01} & no time & \num{49.51} \\
perioddoubling.09 & \num{512} & \num{0.01} & no mem & no mem \\
perioddoubling.10 & \num{1024} & \num{0.03} & no mem & no mem \\
perioddoubling.11 & \num{2048} & \num{0.06} & no mem & no mem \\
perioddoubling.12 & \num{4096} & \num{0.13} & no mem & no mem \\
perioddoubling.13 & \num{8192} & \num{0.31} & no mem & no mem \\
perioddoubling.14 & \num{16384} & \num{0.78} & no mem & no mem \\
perioddoubling.15 & \num{32768} & \num{2.16} & no mem & no mem \\
perioddoubling.16 & \num{65536} & \num{6.42} & no mem & no mem \\
perioddoubling.17 & \num{131072} & \num{23.80} & no mem & no time \\
perioddoubling.18 & \num{262144} & \num{83.41} & no mem & no time \\
perioddoubling.19 & \num{524288} & \num{348.78} & no mem & no time \\
perioddoubling.20 & \num{1048576} & \num{1351.12} & no mem & no time \\
progc & \num{39611} & \num{33.81} & no mem & no mem \\
random.txt & \num{100000} & \num{524.78} & no mem & no mem \\
sum & \num{38240} & \num{39.40} & no mem & no mem \\
thuemorse.00 & \num{1} & \num{0.00} & \num{0.00} & unknown \\
thuemorse.01 & \num{2} & \num{0.00} & \num{0.00} & \num{0.00} \\
thuemorse.02 & \num{4} & \num{0.00} & \num{0.00} & \num{0.00} \\
thuemorse.03 & \num{8} & \num{0.00} & \num{0.00} & \num{0.00} \\
thuemorse.04 & \num{16} & \num{0.00} & \num{0.02} & \num{0.01} \\
thuemorse.05 & \num{32} & \num{0.00} & \num{0.21} & \num{0.04} \\
thuemorse.06 & \num{64} & \num{0.00} & \num{18.68} & \num{0.25} \\
thuemorse.07 & \num{128} & \num{0.00} & \num{3397.38} & \num{2.22} \\
thuemorse.08 & \num{256} & \num{0.01} & no time & \num{20.06} \\
thuemorse.09 & \num{512} & \num{0.01} & no mem & no mem \\
thuemorse.10 & \num{1024} & \num{0.03} & no mem & no mem \\
thuemorse.11 & \num{2048} & \num{0.06} & no mem & no mem \\
thuemorse.12 & \num{4096} & \num{0.13} & no mem & no mem \\
thuemorse.13 & \num{8192} & \num{0.28} & no mem & no mem \\
thuemorse.14 & \num{16384} & \num{0.63} & no mem & no mem \\
thuemorse.15 & \num{32768} & \num{1.43} & no mem & no mem \\
thuemorse.16 & \num{65536} & \num{3.41} & no mem & no mem \\
thuemorse.17 & \num{131072} & \num{8.85} & no mem & no time \\
thuemorse.18 & \num{262144} & \num{26.15} & no mem & no time \\
thuemorse.19 & \num{524288} & \num{84.90} & no mem & no time \\
thuemorse.20 & \num{1048576} & \num{296.43} & no mem & no time \\
xargs.1 & \num{4227} & \num{0.61} & no mem & no mem \\
 \end{longtable}

\begin{longtable}{l*{4}{r}{c}}
	\caption{Number of Defined Literals}
	\label{tab:nvars}
	\\
	file & $n$ & attractor~$\gamma$ & BMS~$b$ & SLCP~$g$ \\
alice29.txt & \num{152089} & \num{152089} & no mem & no mem \\
asyoulik.txt & \num{125179} & \num{125179} & no mem & no mem \\
bib & \num{111261} & \num{111261} & no mem & no mem \\
cp.html & \num{24603} & \num{24603} & no mem & no mem \\
fibonacci.00 & \num{1} & \num{1} & \num{4} & unknown \\
fibonacci.01 & \num{2} & \num{2} & \num{6} & \num{7} \\
fibonacci.02 & \num{3} & \num{3} & \num{17} & \num{12} \\
fibonacci.03 & \num{5} & \num{5} & \num{56} & \num{34} \\
fibonacci.04 & \num{8} & \num{8} & \num{287} & \num{109} \\
fibonacci.05 & \num{13} & \num{13} & \num{1322} & \num{375} \\
fibonacci.06 & \num{21} & \num{21} & \num{5787} & \num{1313} \\
fibonacci.07 & \num{34} & \num{34} & \num{24908} & \num{4442} \\
fibonacci.08 & \num{55} & \num{55} & \num{106728} & \num{14684} \\
fibonacci.09 & \num{89} & \num{89} & \num{455017} & \num{47049} \\
fibonacci.10 & \num{144} & \num{144} & \num{1935574} & \num{147278} \\
fibonacci.11 & \num{233} & \num{233} & \num{8219666} & \num{450933} \\
fibonacci.12 & \num{377} & \num{377} & no mem & no mem \\
fibonacci.13 & \num{610} & \num{610} & no mem & no mem \\
fibonacci.14 & \num{987} & \num{987} & no mem & no mem \\
fibonacci.15 & \num{1597} & \num{1597} & no mem & no mem \\
fibonacci.16 & \num{2584} & \num{2584} & no mem & no mem \\
fibonacci.17 & \num{4181} & \num{4181} & no mem & no mem \\
fibonacci.18 & \num{6765} & \num{6765} & no mem & no mem \\
fibonacci.19 & \num{10946} & \num{10946} & no mem & no mem \\
fibonacci.20 & \num{17711} & \num{17711} & no mem & no mem \\
fields.c & \num{11150} & \num{11150} & no mem & no mem \\
geo & \num{102400} & \num{102400} & no mem & no mem \\
grammar.lsp & \num{3721} & \num{3721} & no mem & no mem \\
obj1 & \num{21504} & \num{21504} & no mem & no mem \\
paper1 & \num{53161} & \num{53161} & no mem & no mem \\
paper2 & \num{82199} & \num{82199} & no mem & no mem \\
paper3 & \num{46526} & \num{46526} & no mem & no mem \\
paper4 & \num{13286} & \num{13286} & no mem & no mem \\
paper5 & \num{11954} & \num{11954} & no mem & no mem \\
paper6 & \num{38105} & \num{38105} & no mem & no mem \\
paperfold.00 & \num{2} & \num{2} & \num{14} & \num{7} \\
paperfold.01 & \num{4} & \num{4} & \num{42} & \num{19} \\
paperfold.02 & \num{8} & \num{8} & \num{287} & \num{91} \\
paperfold.03 & \num{16} & \num{16} & \num{2198} & \num{434} \\
paperfold.04 & \num{32} & \num{32} & \num{16934} & \num{2133} \\
paperfold.05 & \num{64} & \num{64} & \num{133190} & \num{10155} \\
paperfold.06 & \num{128} & \num{128} & no time & \num{47511} \\
paperfold.07 & \num{256} & \num{256} & no time & \num{218983} \\
paperfold.08 & \num{512} & \num{512} & no mem & \num{994951} \\
paperfold.09 & \num{1024} & \num{1024} & no mem & no mem \\
paperfold.10 & \num{2048} & \num{2048} & no mem & no mem \\
paperfold.11 & \num{4096} & \num{4096} & no mem & no mem \\
paperfold.12 & \num{8192} & \num{8192} & no mem & no mem \\
paperfold.13 & \num{16384} & \num{16384} & no mem & no mem \\
paperfold.14 & \num{32768} & \num{32768} & no mem & no mem \\
perioddoubling.00 & \num{1} & \num{1} & \num{4} & unknown \\
perioddoubling.01 & \num{2} & \num{2} & \num{6} & \num{7} \\
perioddoubling.02 & \num{4} & \num{4} & \num{42} & \num{19} \\
perioddoubling.03 & \num{8} & \num{8} & \num{287} & \num{93} \\
perioddoubling.04 & \num{16} & \num{16} & \num{3038} & \num{546} \\
perioddoubling.05 & \num{32} & \num{32} & \num{22934} & \num{3095} \\
perioddoubling.06 & \num{64} & \num{64} & \num{194190} & \num{16544} \\
perioddoubling.07 & \num{128} & \num{128} & no time & \num{83951} \\
perioddoubling.08 & \num{256} & \num{256} & no time & \num{409696} \\
perioddoubling.09 & \num{512} & \num{512} & no mem & no mem \\
perioddoubling.10 & \num{1024} & \num{1024} & no mem & no mem \\
perioddoubling.11 & \num{2048} & \num{2048} & no mem & no mem \\
perioddoubling.12 & \num{4096} & \num{4096} & no mem & no mem \\
perioddoubling.13 & \num{8192} & \num{8192} & no mem & no mem \\
perioddoubling.14 & \num{16384} & \num{16384} & no mem & no mem \\
perioddoubling.15 & \num{32768} & \num{32768} & no mem & no mem \\
perioddoubling.16 & \num{65536} & \num{65536} & no mem & no mem \\
perioddoubling.17 & \num{131072} & \num{131072} & no mem & no time \\
perioddoubling.18 & \num{262144} & \num{262144} & no mem & no time \\
perioddoubling.19 & \num{524288} & \num{524288} & no mem & no time \\
perioddoubling.20 & \num{1048576} & \num{1048576} & no mem & no time \\
progc & \num{39611} & \num{39611} & no mem & no mem \\
random.txt & \num{100000} & \num{100000} & no mem & no mem \\
sum & \num{38240} & \num{38240} & no mem & no mem \\
thuemorse.00 & \num{1} & \num{1} & \num{4} & unknown \\
thuemorse.01 & \num{2} & \num{2} & \num{6} & \num{7} \\
thuemorse.02 & \num{4} & \num{4} & \num{26} & \num{19} \\
thuemorse.03 & \num{8} & \num{8} & \num{226} & \num{87} \\
thuemorse.04 & \num{16} & \num{16} & \num{1922} & \num{449} \\
thuemorse.05 & \num{32} & \num{32} & \num{15874} & \num{2323} \\
thuemorse.06 & \num{64} & \num{64} & \num{129026} & \num{11535} \\
thuemorse.07 & \num{128} & \num{128} & \num{1040386} & \num{55553} \\
thuemorse.08 & \num{256} & \num{256} & no time & \num{260215} \\
thuemorse.09 & \num{512} & \num{512} & no mem & no mem \\
thuemorse.10 & \num{1024} & \num{1024} & no mem & no mem \\
thuemorse.11 & \num{2048} & \num{2048} & no mem & no mem \\
thuemorse.12 & \num{4096} & \num{4096} & no mem & no mem \\
thuemorse.13 & \num{8192} & \num{8192} & no mem & no mem \\
thuemorse.14 & \num{16384} & \num{16384} & no mem & no mem \\
thuemorse.15 & \num{32768} & \num{32768} & no mem & no mem \\
thuemorse.16 & \num{65536} & \num{65536} & no mem & no mem \\
thuemorse.17 & \num{131072} & \num{131072} & no mem & no time \\
thuemorse.18 & \num{262144} & \num{262144} & no mem & no time \\
thuemorse.19 & \num{524288} & \num{524288} & no mem & no time \\
thuemorse.20 & \num{1048576} & \num{1048576} & no mem & no time \\
xargs.1 & \num{4227} & \num{4227} & no mem & no mem \\
 \end{longtable}

\begin{longtable}{l*{4}{r}{c}}
	\caption{Number of Hard Clauses}
	\label{tab:nhard}
	\\
	file & $n$ & attractor~$\gamma$ & BMS~$b$ & SLCP~$g$ \\
alice29.txt & \num{152089} & \num{86964} & no mem & no mem \\
asyoulik.txt & \num{125179} & \num{78822} & no mem & no mem \\
bib & \num{111261} & \num{46197} & no mem & no mem \\
cp.html & \num{24603} & \num{10855} & no mem & no mem \\
fibonacci.00 & \num{1} & \num{1} & \num{6} & unknown \\
fibonacci.01 & \num{2} & \num{2} & \num{11} & \num{12} \\
fibonacci.02 & \num{3} & \num{2} & \num{37} & \num{27} \\
fibonacci.03 & \num{5} & \num{4} & \num{161} & \num{98} \\
fibonacci.04 & \num{8} & \num{4} & \num{913} & \num{358} \\
fibonacci.05 & \num{13} & \num{7} & \num{4444} & \num{1397} \\
fibonacci.06 & \num{21} & \num{7} & \num{20467} & \num{5759} \\
fibonacci.07 & \num{34} & \num{10} & \num{90582} & \num{24989} \\
fibonacci.08 & \num{55} & \num{10} & \num{394644} & \num{117964} \\
fibonacci.09 & \num{89} & \num{13} & \num{1699154} & \num{614887} \\
fibonacci.10 & \num{144} & \num{13} & \num{7271531} & \num{3525773} \\
fibonacci.11 & \num{233} & \num{16} & \num{30992634} & \num{21741434} \\
fibonacci.12 & \num{377} & \num{16} & no mem & no mem \\
fibonacci.13 & \num{610} & \num{19} & no mem & no mem \\
fibonacci.14 & \num{987} & \num{19} & no mem & no mem \\
fibonacci.15 & \num{1597} & \num{22} & no mem & no mem \\
fibonacci.16 & \num{2584} & \num{22} & no mem & no mem \\
fibonacci.17 & \num{4181} & \num{25} & no mem & no mem \\
fibonacci.18 & \num{6765} & \num{25} & no mem & no mem \\
fibonacci.19 & \num{10946} & \num{28} & no mem & no mem \\
fibonacci.20 & \num{17711} & \num{28} & no mem & no mem \\
fields.c & \num{11150} & \num{4206} & no mem & no mem \\
geo & \num{102400} & \num{68169} & no mem & no mem \\
grammar.lsp & \num{3721} & \num{1669} & no mem & no mem \\
obj1 & \num{21504} & \num{11469} & no mem & no mem \\
paper1 & \num{53161} & \num{27795} & no mem & no mem \\
paper2 & \num{82199} & \num{46594} & no mem & no mem \\
paper3 & \num{46526} & \num{27979} & no mem & no mem \\
paper4 & \num{13286} & \num{8355} & no mem & no mem \\
paper5 & \num{11954} & \num{7080} & no mem & no mem \\
paper6 & \num{38105} & \num{19923} & no mem & no mem \\
paperfold.00 & \num{2} & \num{2} & \num{31} & \num{11} \\
paperfold.01 & \num{4} & \num{3} & \num{121} & \num{51} \\
paperfold.02 & \num{8} & \num{8} & \num{907} & \num{312} \\
paperfold.03 & \num{16} & \num{14} & \num{7489} & \num{1864} \\
paperfold.04 & \num{32} & \num{28} & \num{60657} & \num{11885} \\
paperfold.05 & \num{64} & \num{38} & \num{488401} & \num{82596} \\
paperfold.06 & \num{128} & \num{47} & no time & \num{673692} \\
paperfold.07 & \num{256} & \num{56} & no time & \num{6643404} \\
paperfold.08 & \num{512} & \num{65} & no mem & \num{78308716} \\
paperfold.09 & \num{1024} & \num{74} & no mem & no mem \\
paperfold.10 & \num{2048} & \num{83} & no mem & no mem \\
paperfold.11 & \num{4096} & \num{92} & no mem & no mem \\
paperfold.12 & \num{8192} & \num{101} & no mem & no mem \\
paperfold.13 & \num{16384} & \num{110} & no mem & no mem \\
paperfold.14 & \num{32768} & \num{119} & no mem & no mem \\
perioddoubling.00 & \num{1} & \num{1} & \num{6} & unknown \\
perioddoubling.01 & \num{2} & \num{2} & \num{11} & \num{12} \\
perioddoubling.02 & \num{4} & \num{3} & \num{121} & \num{51} \\
perioddoubling.03 & \num{8} & \num{8} & \num{910} & \num{317} \\
perioddoubling.04 & \num{16} & \num{11} & \num{10583} & \num{2209} \\
perioddoubling.05 & \num{32} & \num{18} & \num{83457} & \num{16336} \\
perioddoubling.06 & \num{64} & \num{21} & \num{722991} & \num{138636} \\
perioddoubling.07 & \num{128} & \num{28} & no time & \num{1451220} \\
perioddoubling.08 & \num{256} & \num{31} & no time & \num{18467052} \\
perioddoubling.09 & \num{512} & \num{38} & no mem & no mem \\
perioddoubling.10 & \num{1024} & \num{41} & no mem & no mem \\
perioddoubling.11 & \num{2048} & \num{48} & no mem & no mem \\
perioddoubling.12 & \num{4096} & \num{51} & no mem & no mem \\
perioddoubling.13 & \num{8192} & \num{58} & no mem & no mem \\
perioddoubling.14 & \num{16384} & \num{61} & no mem & no mem \\
perioddoubling.15 & \num{32768} & \num{68} & no mem & no mem \\
perioddoubling.16 & \num{65536} & \num{71} & no mem & no mem \\
perioddoubling.17 & \num{131072} & \num{78} & no mem & no time \\
perioddoubling.18 & \num{262144} & \num{81} & no mem & no time \\
perioddoubling.19 & \num{524288} & \num{88} & no mem & no time \\
perioddoubling.20 & \num{1048576} & \num{91} & no mem & no time \\
progc & \num{39611} & \num{19008} & no mem & no mem \\
random.txt & \num{100000} & \num{97166} & no mem & no mem \\
sum & \num{38240} & \num{15328} & no mem & no mem \\
thuemorse.00 & \num{1} & \num{1} & \num{6} & unknown \\
thuemorse.01 & \num{2} & \num{2} & \num{11} & \num{12} \\
thuemorse.02 & \num{4} & \num{5} & \num{61} & \num{51} \\
thuemorse.03 & \num{8} & \num{8} & \num{666} & \num{304} \\
thuemorse.04 & \num{16} & \num{16} & \num{6476} & \num{1916} \\
thuemorse.05 & \num{32} & \num{24} & \num{56612} & \num{12681} \\
thuemorse.06 & \num{64} & \num{32} & \num{472180} & \num{90587} \\
thuemorse.07 & \num{128} & \num{40} & \num{3854804} & \num{747417} \\
thuemorse.08 & \num{256} & \num{48} & no time & \num{7427323} \\
thuemorse.09 & \num{512} & \num{56} & no mem & no mem \\
thuemorse.10 & \num{1024} & \num{64} & no mem & no mem \\
thuemorse.11 & \num{2048} & \num{72} & no mem & no mem \\
thuemorse.12 & \num{4096} & \num{80} & no mem & no mem \\
thuemorse.13 & \num{8192} & \num{88} & no mem & no mem \\
thuemorse.14 & \num{16384} & \num{96} & no mem & no mem \\
thuemorse.15 & \num{32768} & \num{104} & no mem & no mem \\
thuemorse.16 & \num{65536} & \num{112} & no mem & no mem \\
thuemorse.17 & \num{131072} & \num{120} & no mem & no time \\
thuemorse.18 & \num{262144} & \num{128} & no mem & no time \\
thuemorse.19 & \num{524288} & \num{136} & no mem & no time \\
thuemorse.20 & \num{1048576} & \num{144} & no mem & no time \\
xargs.1 & \num{4227} & \num{2366} & no mem & no mem \\
 \end{longtable}

\begin{longtable}{l*{4}{r}{c}}
	\caption{Size of Largest Clause (\# Literals)}
	\label{tab:sol_nmaxclause}
	\\
	file & $n$ & attractor~$\gamma$ & BMS~$b$ & SLCP~$g$ \\
alice29.txt & \num{152089} & \num{28900} & no mem & no mem \\
asyoulik.txt & \num{125179} & \num{19359} & no mem & no mem \\
bib & \num{111261} & \num{13739} & no mem & no mem \\
cp.html & \num{24603} & \num{1824} & no mem & no mem \\
fibonacci.00 & \num{1} & \num{1} & \num{3} & unknown \\
fibonacci.01 & \num{2} & \num{1} & \num{5} & \num{4} \\
fibonacci.02 & \num{3} & \num{2} & \num{14} & \num{5} \\
fibonacci.03 & \num{5} & \num{3} & \num{45} & \num{7} \\
fibonacci.04 & \num{8} & \num{5} & \num{179} & \num{10} \\
fibonacci.05 & \num{13} & \num{8} & \num{726} & \num{15} \\
fibonacci.06 & \num{21} & \num{13} & \num{3051} & \num{23} \\
fibonacci.07 & \num{34} & \num{26} & \num{12845} & \num{36} \\
fibonacci.08 & \num{55} & \num{35} & \num{54354} & \num{58} \\
fibonacci.09 & \num{89} & \num{73} & \num{230045} & \num{97} \\
fibonacci.10 & \num{144} & \num{100} & \num{974339} & \num{160} \\
fibonacci.11 & \num{233} & \num{196} & \num{4126842} & \num{262} \\
fibonacci.12 & \num{377} & \num{270} & no mem & no mem \\
fibonacci.13 & \num{610} & \num{518} & no mem & no mem \\
fibonacci.14 & \num{987} & \num{715} & no mem & no mem \\
fibonacci.15 & \num{1597} & \num{1361} & no mem & no mem \\
fibonacci.16 & \num{2584} & \num{1880} & no mem & no mem \\
fibonacci.17 & \num{4181} & \num{3568} & no mem & no mem \\
fibonacci.18 & \num{6765} & \num{4930} & no mem & no mem \\
fibonacci.19 & \num{10946} & \num{9346} & no mem & no mem \\
fibonacci.20 & \num{17711} & \num{12915} & no mem & no mem \\
fields.c & \num{11150} & \num{2213} & no mem & no mem \\
geo & \num{102400} & \num{28626} & no mem & no mem \\
grammar.lsp & \num{3721} & \num{802} & no mem & no mem \\
obj1 & \num{21504} & \num{5552} & no mem & no mem \\
paper1 & \num{53161} & \num{7301} & no mem & no mem \\
paper2 & \num{82199} & \num{12112} & no mem & no mem \\
paper3 & \num{46526} & \num{6154} & no mem & no mem \\
paper4 & \num{13286} & \num{1958} & no mem & no mem \\
paper5 & \num{11954} & \num{1869} & no mem & no mem \\
paper6 & \num{38105} & \num{5721} & no mem & no mem \\
paperfold.00 & \num{2} & \num{2} & \num{11} & \num{4} \\
paperfold.01 & \num{4} & \num{3} & \num{35} & \num{6} \\
paperfold.02 & \num{8} & \num{5} & \num{179} & \num{10} \\
paperfold.03 & \num{16} & \num{9} & \num{1187} & \num{18} \\
paperfold.04 & \num{32} & \num{18} & \num{8771} & \num{34} \\
paperfold.05 & \num{64} & \num{36} & \num{67715} & \num{66} \\
paperfold.06 & \num{128} & \num{72} & no time & \num{130} \\
paperfold.07 & \num{256} & \num{144} & no time & \num{258} \\
paperfold.08 & \num{512} & \num{288} & no mem & \num{514} \\
paperfold.09 & \num{1024} & \num{576} & no mem & no mem \\
paperfold.10 & \num{2048} & \num{1152} & no mem & no mem \\
paperfold.11 & \num{4096} & \num{2304} & no mem & no mem \\
paperfold.12 & \num{8192} & \num{4608} & no mem & no mem \\
paperfold.13 & \num{16384} & \num{9216} & no mem & no mem \\
paperfold.14 & \num{32768} & \num{18432} & no mem & no mem \\
perioddoubling.00 & \num{1} & \num{1} & \num{3} & unknown \\
perioddoubling.01 & \num{2} & \num{1} & \num{5} & \num{4} \\
perioddoubling.02 & \num{4} & \num{3} & \num{35} & \num{6} \\
perioddoubling.03 & \num{8} & \num{5} & \num{179} & \num{10} \\
perioddoubling.04 & \num{16} & \num{11} & \num{1623} & \num{18} \\
perioddoubling.05 & \num{32} & \num{21} & \num{11835} & \num{34} \\
perioddoubling.06 & \num{64} & \num{43} & \num{98535} & \num{66} \\
perioddoubling.07 & \num{128} & \num{85} & no time & \num{130} \\
perioddoubling.08 & \num{256} & \num{171} & no time & \num{258} \\
perioddoubling.09 & \num{512} & \num{341} & no mem & no mem \\
perioddoubling.10 & \num{1024} & \num{683} & no mem & no mem \\
perioddoubling.11 & \num{2048} & \num{1365} & no mem & no mem \\
perioddoubling.12 & \num{4096} & \num{2731} & no mem & no mem \\
perioddoubling.13 & \num{8192} & \num{5461} & no mem & no mem \\
perioddoubling.14 & \num{16384} & \num{10923} & no mem & no mem \\
perioddoubling.15 & \num{32768} & \num{21845} & no mem & no mem \\
perioddoubling.16 & \num{65536} & \num{43691} & no mem & no mem \\
perioddoubling.17 & \num{131072} & \num{87381} & no mem & no time \\
perioddoubling.18 & \num{262144} & \num{174763} & no mem & no time \\
perioddoubling.19 & \num{524288} & \num{349525} & no mem & no time \\
perioddoubling.20 & \num{1048576} & \num{699051} & no mem & no time \\
progc & \num{39611} & \num{6925} & no mem & no mem \\
random.txt & \num{100000} & \num{1668} & no mem & no mem \\
sum & \num{38240} & \num{12258} & no mem & no mem \\
thuemorse.00 & \num{1} & \num{1} & \num{3} & unknown \\
thuemorse.01 & \num{2} & \num{1} & \num{5} & \num{4} \\
thuemorse.02 & \num{4} & \num{2} & \num{21} & \num{6} \\
thuemorse.03 & \num{8} & \num{6} & \num{137} & \num{10} \\
thuemorse.04 & \num{16} & \num{10} & \num{1041} & \num{18} \\
thuemorse.05 & \num{32} & \num{22} & \num{8225} & \num{34} \\
thuemorse.06 & \num{64} & \num{42} & \num{65601} & \num{66} \\
thuemorse.07 & \num{128} & \num{86} & \num{524417} & \num{130} \\
thuemorse.08 & \num{256} & \num{170} & no time & \num{258} \\
thuemorse.09 & \num{512} & \num{342} & no mem & no mem \\
thuemorse.10 & \num{1024} & \num{682} & no mem & no mem \\
thuemorse.11 & \num{2048} & \num{1366} & no mem & no mem \\
thuemorse.12 & \num{4096} & \num{2730} & no mem & no mem \\
thuemorse.13 & \num{8192} & \num{5462} & no mem & no mem \\
thuemorse.14 & \num{16384} & \num{10922} & no mem & no mem \\
thuemorse.15 & \num{32768} & \num{21846} & no mem & no mem \\
thuemorse.16 & \num{65536} & \num{43690} & no mem & no mem \\
thuemorse.17 & \num{131072} & \num{87382} & no mem & no time \\
thuemorse.18 & \num{262144} & \num{174762} & no mem & no time \\
thuemorse.19 & \num{524288} & \num{349526} & no mem & no time \\
thuemorse.20 & \num{1048576} & \num{699050} & no mem & no time \\
xargs.1 & \num{4227} & \num{550} & no mem & no mem \\
 \end{longtable}

\begin{longtable}{l*{4}{r}{c}}
	\caption{Sum of all literal occurrences}
	\label{tab:sol_ntotalvars}
	\\
	file & $n$ & attractor~$\gamma$ & BMS~$b$ & SLCP~$g$ \\
alice29.txt & \num{152089} & \num{3022673} & no mem & no mem \\
asyoulik.txt & \num{125179} & \num{2111965} & no mem & no mem \\
bib & \num{111261} & \num{1573621} & no mem & no mem \\
cp.html & \num{24603} & \num{309987} & no mem & no mem \\
fibonacci.00 & \num{1} & \num{1} & \num{9} & unknown \\
fibonacci.01 & \num{2} & \num{2} & \num{18} & \num{20} \\
fibonacci.02 & \num{3} & \num{3} & \num{81} & \num{53} \\
fibonacci.03 & \num{5} & \num{9} & \num{383} & \num{212} \\
fibonacci.04 & \num{8} & \num{13} & \num{2196} & \num{813} \\
fibonacci.05 & \num{13} & \num{43} & \num{10655} & \num{3277} \\
fibonacci.06 & \num{21} & \num{61} & \num{49061} & \num{13792} \\
fibonacci.07 & \num{34} & \num{162} & \num{217049} & \num{59834} \\
fibonacci.08 & \num{55} & \num{236} & \num{945535} & \num{276845} \\
fibonacci.09 & \num{89} & \num{551} & \num{4070671} & \num{1394954} \\
fibonacci.10 & \num{144} & \num{817} & \num{17419865} & \num{7711627} \\
fibonacci.11 & \num{233} & \num{1766} & \num{74244963} & \num{46119239} \\
fibonacci.12 & \num{377} & \num{2654} & no mem & no mem \\
fibonacci.13 & \num{610} & \num{5456} & no mem & no mem \\
fibonacci.14 & \num{987} & \num{8285} & no mem & no mem \\
fibonacci.15 & \num{1597} & \num{16444} & no mem & no mem \\
fibonacci.16 & \num{2584} & \num{25173} & no mem & no mem \\
fibonacci.17 & \num{4181} & \num{48681} & no mem & no mem \\
fibonacci.18 & \num{6765} & \num{74999} & no mem & no mem \\
fibonacci.19 & \num{10946} & \num{142158} & no mem & no mem \\
fibonacci.20 & \num{17711} & \num{220134} & no mem & no mem \\
fields.c & \num{11150} & \num{101934} & no mem & no mem \\
geo & \num{102400} & \num{689428} & no mem & no mem \\
grammar.lsp & \num{3721} & \num{29809} & no mem & no mem \\
obj1 & \num{21504} & \num{2254811} & no mem & no mem \\
paper1 & \num{53161} & \num{727940} & no mem & no mem \\
paper2 & \num{82199} & \num{1345820} & no mem & no mem \\
paper3 & \num{46526} & \num{645173} & no mem & no mem \\
paper4 & \num{13286} & \num{139390} & no mem & no mem \\
paper5 & \num{11954} & \num{113544} & no mem & no mem \\
paper6 & \num{38105} & \num{488634} & no mem & no mem \\
paperfold.00 & \num{2} & \num{4} & \num{68} & \num{19} \\
paperfold.01 & \num{4} & \num{6} & \num{286} & \num{107} \\
paperfold.02 & \num{8} & \num{28} & \num{2172} & \num{707} \\
paperfold.03 & \num{16} & \num{85} & \num{17937} & \num{4499} \\
paperfold.04 & \num{32} & \num{317} & \num{145409} & \num{30206} \\
paperfold.05 & \num{64} & \num{764} & \num{1171425} & \num{215216} \\
paperfold.06 & \num{128} & \num{1699} & no time & \num{1732448} \\
paperfold.07 & \num{256} & \num{3714} & no time & \num{16269232} \\
paperfold.08 & \num{512} & \num{8033} & no mem & \num{179917264} \\
paperfold.09 & \num{1024} & \num{17248} & no mem & no mem \\
paperfold.10 & \num{2048} & \num{36831} & no mem & no mem \\
paperfold.11 & \num{4096} & \num{78302} & no mem & no mem \\
paperfold.12 & \num{8192} & \num{165853} & no mem & no mem \\
paperfold.13 & \num{16384} & \num{350172} & no mem & no mem \\
paperfold.14 & \num{32768} & \num{737243} & no mem & no mem \\
perioddoubling.00 & \num{1} & \num{1} & \num{9} & unknown \\
perioddoubling.01 & \num{2} & \num{2} & \num{18} & \num{20} \\
perioddoubling.02 & \num{4} & \num{6} & \num{286} & \num{107} \\
perioddoubling.03 & \num{8} & \num{27} & \num{2187} & \num{721} \\
perioddoubling.04 & \num{16} & \num{70} & \num{25341} & \num{5290} \\
perioddoubling.05 & \num{32} & \num{229} & \num{199881} & \num{40074} \\
perioddoubling.06 & \num{64} & \num{500} & \num{1731325} & \num{334754} \\
perioddoubling.07 & \num{128} & \num{1371} & no time & \num{3334218} \\
perioddoubling.08 & \num{256} & \num{2874} & no time & \num{40174754} \\
perioddoubling.09 & \num{512} & \num{7233} & no mem & no mem \\
perioddoubling.10 & \num{1024} & \num{14944} & no mem & no mem \\
perioddoubling.11 & \num{2048} & \num{35815} & no mem & no mem \\
perioddoubling.12 & \num{4096} & \num{73478} & no mem & no mem \\
perioddoubling.13 & \num{8192} & \num{170637} & no mem & no mem \\
perioddoubling.14 & \num{16384} & \num{348588} & no mem & no mem \\
perioddoubling.15 & \num{32768} & \num{791859} & no mem & no mem \\
perioddoubling.16 & \num{65536} & \num{1612882} & no mem & no mem \\
perioddoubling.17 & \num{131072} & \num{3604441} & no mem & no time \\
perioddoubling.18 & \num{262144} & \num{7325432} & no mem & no time \\
perioddoubling.19 & \num{524288} & \num{16165503} & no mem & no time \\
perioddoubling.20 & \num{1048576} & \num{32797086} & no mem & no time \\
progc & \num{39611} & \num{470793} & no mem & no mem \\
random.txt & \num{100000} & \num{640087} & no mem & no mem \\
sum & \num{38240} & \num{394529} & no mem & no mem \\
thuemorse.00 & \num{1} & \num{1} & \num{9} & unknown \\
thuemorse.01 & \num{2} & \num{2} & \num{18} & \num{20} \\
thuemorse.02 & \num{4} & \num{10} & \num{136} & \num{107} \\
thuemorse.03 & \num{8} & \num{28} & \num{1595} & \num{694} \\
thuemorse.04 & \num{16} & \num{90} & \num{15537} & \num{4625} \\
thuemorse.05 & \num{32} & \num{248} & \num{135841} & \num{31981} \\
thuemorse.06 & \num{64} & \num{604} & \num{1133057} & \num{232815} \\
thuemorse.07 & \num{128} & \num{1414} & \num{9250753} & \num{1890675} \\
thuemorse.08 & \num{256} & \num{3202} & no time & \num{17895815} \\
thuemorse.09 & \num{512} & \num{7132} & no mem & no mem \\
thuemorse.10 & \num{1024} & \num{15672} & no mem & no mem \\
thuemorse.11 & \num{2048} & \num{34130} & no mem & no mem \\
thuemorse.12 & \num{4096} & \num{73774} & no mem & no mem \\
thuemorse.13 & \num{8192} & \num{158536} & no mem & no mem \\
thuemorse.14 & \num{16384} & \num{338980} & no mem & no mem \\
thuemorse.15 & \num{32768} & \num{721726} & no mem & no mem \\
thuemorse.16 & \num{65536} & \num{1530906} & no mem & no mem \\
thuemorse.17 & \num{131072} & \num{3236660} & no mem & no time \\
thuemorse.18 & \num{262144} & \num{6822928} & no mem & no time \\
thuemorse.19 & \num{524288} & \num{14345002} & no mem & no time \\
thuemorse.20 & \num{1048576} & \num{30088198} & no mem & no time \\
xargs.1 & \num{4227} & \num{30749} & no mem & no mem \\
 \end{longtable}

\begin{longtable}{l*{4}{r}{c}}
	\caption{Average number of literals per clause}
	\label{tab:sol_navgclause}
	\\
	file & $n$ & attractor~$\gamma$ & BMS~$b$ & SLCP~$g$ \\
alice29.txt & \num{152089} & \num{34.76} & no mem & no mem \\
asyoulik.txt & \num{125179} & \num{26.79} & no mem & no mem \\
bib & \num{111261} & \num{34.06} & no mem & no mem \\
cp.html & \num{24603} & \num{28.56} & no mem & no mem \\
fibonacci.00 & \num{1} & \num{1.00} & \num{1.50} & unknown \\
fibonacci.01 & \num{2} & \num{1.00} & \num{1.64} & \num{1.67} \\
fibonacci.02 & \num{3} & \num{1.50} & \num{2.19} & \num{1.96} \\
fibonacci.03 & \num{5} & \num{2.25} & \num{2.38} & \num{2.16} \\
fibonacci.04 & \num{8} & \num{3.25} & \num{2.41} & \num{2.27} \\
fibonacci.05 & \num{13} & \num{6.14} & \num{2.40} & \num{2.35} \\
fibonacci.06 & \num{21} & \num{8.71} & \num{2.40} & \num{2.39} \\
fibonacci.07 & \num{34} & \num{16.20} & \num{2.40} & \num{2.39} \\
fibonacci.08 & \num{55} & \num{23.60} & \num{2.40} & \num{2.35} \\
fibonacci.09 & \num{89} & \num{42.38} & \num{2.40} & \num{2.27} \\
fibonacci.10 & \num{144} & \num{62.85} & \num{2.40} & \num{2.19} \\
fibonacci.11 & \num{233} & \num{110.38} & \num{2.40} & \num{2.12} \\
fibonacci.12 & \num{377} & \num{165.88} & no mem & no mem \\
fibonacci.13 & \num{610} & \num{287.16} & no mem & no mem \\
fibonacci.14 & \num{987} & \num{436.05} & no mem & no mem \\
fibonacci.15 & \num{1597} & \num{747.45} & no mem & no mem \\
fibonacci.16 & \num{2584} & \num{1144.23} & no mem & no mem \\
fibonacci.17 & \num{4181} & \num{1947.24} & no mem & no mem \\
fibonacci.18 & \num{6765} & \num{2999.96} & no mem & no mem \\
fibonacci.19 & \num{10946} & \num{5077.07} & no mem & no mem \\
fibonacci.20 & \num{17711} & \num{7861.93} & no mem & no mem \\
fields.c & \num{11150} & \num{24.24} & no mem & no mem \\
geo & \num{102400} & \num{10.11} & no mem & no mem \\
grammar.lsp & \num{3721} & \num{17.86} & no mem & no mem \\
obj1 & \num{21504} & \num{196.60} & no mem & no mem \\
paper1 & \num{53161} & \num{26.19} & no mem & no mem \\
paper2 & \num{82199} & \num{28.88} & no mem & no mem \\
paper3 & \num{46526} & \num{23.06} & no mem & no mem \\
paper4 & \num{13286} & \num{16.68} & no mem & no mem \\
paper5 & \num{11954} & \num{16.04} & no mem & no mem \\
paper6 & \num{38105} & \num{24.53} & no mem & no mem \\
paperfold.00 & \num{2} & \num{2.00} & \num{2.19} & \num{1.73} \\
paperfold.01 & \num{4} & \num{2.00} & \num{2.36} & \num{2.10} \\
paperfold.02 & \num{8} & \num{3.50} & \num{2.39} & \num{2.27} \\
paperfold.03 & \num{16} & \num{6.07} & \num{2.40} & \num{2.41} \\
paperfold.04 & \num{32} & \num{11.32} & \num{2.40} & \num{2.54} \\
paperfold.05 & \num{64} & \num{20.11} & \num{2.40} & \num{2.61} \\
paperfold.06 & \num{128} & \num{36.15} & no time & \num{2.57} \\
paperfold.07 & \num{256} & \num{66.32} & no time & \num{2.45} \\
paperfold.08 & \num{512} & \num{123.58} & no mem & \num{2.30} \\
paperfold.09 & \num{1024} & \num{233.08} & no mem & no mem \\
paperfold.10 & \num{2048} & \num{443.75} & no mem & no mem \\
paperfold.11 & \num{4096} & \num{851.11} & no mem & no mem \\
paperfold.12 & \num{8192} & \num{1642.11} & no mem & no mem \\
paperfold.13 & \num{16384} & \num{3183.38} & no mem & no mem \\
paperfold.14 & \num{32768} & \num{6195.32} & no mem & no mem \\
perioddoubling.00 & \num{1} & \num{1.00} & \num{1.50} & unknown \\
perioddoubling.01 & \num{2} & \num{1.00} & \num{1.64} & \num{1.67} \\
perioddoubling.02 & \num{4} & \num{2.00} & \num{2.36} & \num{2.10} \\
perioddoubling.03 & \num{8} & \num{3.38} & \num{2.40} & \num{2.27} \\
perioddoubling.04 & \num{16} & \num{6.36} & \num{2.39} & \num{2.39} \\
perioddoubling.05 & \num{32} & \num{12.72} & \num{2.40} & \num{2.45} \\
perioddoubling.06 & \num{64} & \num{23.81} & \num{2.39} & \num{2.41} \\
perioddoubling.07 & \num{128} & \num{48.96} & no time & \num{2.30} \\
perioddoubling.08 & \num{256} & \num{92.71} & no time & \num{2.18} \\
perioddoubling.09 & \num{512} & \num{190.34} & no mem & no mem \\
perioddoubling.10 & \num{1024} & \num{364.49} & no mem & no mem \\
perioddoubling.11 & \num{2048} & \num{746.15} & no mem & no mem \\
perioddoubling.12 & \num{4096} & \num{1440.75} & no mem & no mem \\
perioddoubling.13 & \num{8192} & \num{2942.02} & no mem & no mem \\
perioddoubling.14 & \num{16384} & \num{5714.56} & no mem & no mem \\
perioddoubling.15 & \num{32768} & \num{11644.99} & no mem & no mem \\
perioddoubling.16 & \num{65536} & \num{22716.65} & no mem & no mem \\
perioddoubling.17 & \num{131072} & \num{46210.78} & no mem & no time \\
perioddoubling.18 & \num{262144} & \num{90437.43} & no mem & no time \\
perioddoubling.19 & \num{524288} & \num{183698.90} & no mem & no time \\
perioddoubling.20 & \num{1048576} & \num{360407.54} & no mem & no time \\
progc & \num{39611} & \num{24.77} & no mem & no mem \\
random.txt & \num{100000} & \num{6.59} & no mem & no mem \\
sum & \num{38240} & \num{25.74} & no mem & no mem \\
thuemorse.00 & \num{1} & \num{1.00} & \num{1.50} & unknown \\
thuemorse.01 & \num{2} & \num{1.00} & \num{1.64} & \num{1.67} \\
thuemorse.02 & \num{4} & \num{2.00} & \num{2.23} & \num{2.10} \\
thuemorse.03 & \num{8} & \num{3.50} & \num{2.39} & \num{2.28} \\
thuemorse.04 & \num{16} & \num{5.63} & \num{2.40} & \num{2.41} \\
thuemorse.05 & \num{32} & \num{10.33} & \num{2.40} & \num{2.52} \\
thuemorse.06 & \num{64} & \num{18.88} & \num{2.40} & \num{2.57} \\
thuemorse.07 & \num{128} & \num{35.35} & \num{2.40} & \num{2.53} \\
thuemorse.08 & \num{256} & \num{66.71} & no time & \num{2.41} \\
thuemorse.09 & \num{512} & \num{127.36} & no mem & no mem \\
thuemorse.10 & \num{1024} & \num{244.88} & no mem & no mem \\
thuemorse.11 & \num{2048} & \num{474.03} & no mem & no mem \\
thuemorse.12 & \num{4096} & \num{922.18} & no mem & no mem \\
thuemorse.13 & \num{8192} & \num{1801.55} & no mem & no mem \\
thuemorse.14 & \num{16384} & \num{3531.04} & no mem & no mem \\
thuemorse.15 & \num{32768} & \num{6939.67} & no mem & no mem \\
thuemorse.16 & \num{65536} & \num{13668.80} & no mem & no mem \\
thuemorse.17 & \num{131072} & \num{26972.17} & no mem & no time \\
thuemorse.18 & \num{262144} & \num{53304.13} & no mem & no time \\
thuemorse.19 & \num{524288} & \num{105477.96} & no mem & no time \\
thuemorse.20 & \num{1048576} & \num{208945.82} & no mem & no time \\
xargs.1 & \num{4227} & \num{13.00} & no mem & no mem \\
 \end{longtable}

\begin{longtable}{l*{4}{r}{c}}
	\caption{Output Size}
	\label{tab:factor_size}
	\\
	file & $n$ & attractor~$\gamma$ & BMS~$b$ & SLCP~$g$ \\
alice29.txt & \num{152089} & \num{17138} & no mem & no mem \\
asyoulik.txt & \num{125179} & \num{15938} & no mem & no mem \\
bib & \num{111261} & \num{10371} & no mem & no mem \\
cp.html & \num{24603} & \num{2813} & no mem & no mem \\
fibonacci.00 & \num{1} & \num{1} & \num{1} & unknown \\
fibonacci.01 & \num{2} & \num{2} & \num{2} & \num{3} \\
fibonacci.02 & \num{3} & \num{2} & \num{3} & \num{4} \\
fibonacci.03 & \num{5} & \num{2} & \num{4} & \num{5} \\
fibonacci.04 & \num{8} & \num{2} & \num{4} & \num{6} \\
fibonacci.05 & \num{13} & \num{2} & \num{4} & \num{7} \\
fibonacci.06 & \num{21} & \num{2} & \num{4} & \num{8} \\
fibonacci.07 & \num{34} & \num{2} & \num{4} & \num{9} \\
fibonacci.08 & \num{55} & \num{2} & \num{4} & \num{10} \\
fibonacci.09 & \num{89} & \num{2} & \num{4} & \num{11} \\
fibonacci.10 & \num{144} & \num{2} & \num{4} & \num{12} \\
fibonacci.11 & \num{233} & \num{2} & \num{4} & \num{13} \\
fibonacci.12 & \num{377} & \num{2} & no mem & no mem \\
fibonacci.13 & \num{610} & \num{2} & no mem & no mem \\
fibonacci.14 & \num{987} & \num{2} & no mem & no mem \\
fibonacci.15 & \num{1597} & \num{2} & no mem & no mem \\
fibonacci.16 & \num{2584} & \num{2} & no mem & no mem \\
fibonacci.17 & \num{4181} & \num{2} & no mem & no mem \\
fibonacci.18 & \num{6765} & \num{2} & no mem & no mem \\
fibonacci.19 & \num{10946} & \num{2} & no mem & no mem \\
fibonacci.20 & \num{17711} & \num{2} & no mem & no mem \\
fields.c & \num{11150} & \num{1141} & no mem & no mem \\
geo & \num{102400} & \num{21590} & no mem & no mem \\
grammar.lsp & \num{3721} & \num{497} & no mem & no mem \\
obj1 & \num{21504} & \num{3866} & no mem & no mem \\
paper1 & \num{53161} & \num{6355} & no mem & no mem \\
paper2 & \num{82199} & \num{9884} & no mem & no mem \\
paper3 & \num{46526} & \num{6295} & no mem & no mem \\
paper4 & \num{13286} & \num{2055} & no mem & no mem \\
paper5 & \num{11954} & \num{1879} & no mem & no mem \\
paper6 & \num{38105} & \num{4668} & no mem & no mem \\
paperfold.00 & \num{2} & \num{1} & \num{2} & \num{2} \\
paperfold.01 & \num{4} & \num{2} & \num{4} & \num{5} \\
paperfold.02 & \num{8} & \num{2} & \num{5} & \num{7} \\
paperfold.03 & \num{16} & \num{3} & \num{6} & \num{10} \\
paperfold.04 & \num{32} & \num{4} & \num{8} & \num{14} \\
paperfold.05 & \num{64} & \num{5} & \num{9} & \num{18} \\
paperfold.06 & \num{128} & \num{5} & no time & \num{22} \\
paperfold.07 & \num{256} & \num{6} & no time & \num{26} \\
paperfold.08 & \num{512} & \num{6} & no mem & \num{30} \\
paperfold.09 & \num{1024} & \num{7} & no mem & no mem \\
paperfold.10 & \num{2048} & \num{7} & no mem & no mem \\
paperfold.11 & \num{4096} & \num{7} & no mem & no mem \\
paperfold.12 & \num{8192} & \num{7} & no mem & no mem \\
paperfold.13 & \num{16384} & \num{7} & no mem & no mem \\
paperfold.14 & \num{32768} & \num{7} & no mem & no mem \\
perioddoubling.00 & \num{1} & \num{1} & \num{1} & unknown \\
perioddoubling.01 & \num{2} & \num{2} & \num{2} & \num{3} \\
perioddoubling.02 & \num{4} & \num{2} & \num{4} & \num{5} \\
perioddoubling.03 & \num{8} & \num{2} & \num{5} & \num{7} \\
perioddoubling.04 & \num{16} & \num{2} & \num{6} & \num{9} \\
perioddoubling.05 & \num{32} & \num{2} & \num{7} & \num{11} \\
perioddoubling.06 & \num{64} & \num{2} & \num{7} & \num{13} \\
perioddoubling.07 & \num{128} & \num{2} & no time & \num{15} \\
perioddoubling.08 & \num{256} & \num{2} & no time & \num{17} \\
perioddoubling.09 & \num{512} & \num{2} & no mem & no mem \\
perioddoubling.10 & \num{1024} & \num{2} & no mem & no mem \\
perioddoubling.11 & \num{2048} & \num{2} & no mem & no mem \\
perioddoubling.12 & \num{4096} & \num{2} & no mem & no mem \\
perioddoubling.13 & \num{8192} & \num{2} & no mem & no mem \\
perioddoubling.14 & \num{16384} & \num{2} & no mem & no mem \\
perioddoubling.15 & \num{32768} & \num{2} & no mem & no mem \\
perioddoubling.16 & \num{65536} & \num{2} & no mem & no mem \\
perioddoubling.17 & \num{131072} & \num{2} & no mem & no time \\
perioddoubling.18 & \num{262144} & \num{2} & no mem & no time \\
perioddoubling.19 & \num{524288} & \num{2} & no mem & no time \\
perioddoubling.20 & \num{1048576} & \num{2} & no mem & no time \\
progc & \num{39611} & \num{4714} & no mem & no mem \\
random.txt & \num{100000} & \num{30208} & no mem & no mem \\
sum & \num{38240} & \num{4431} & no mem & no mem \\
thuemorse.00 & \num{1} & \num{1} & \num{1} & unknown \\
thuemorse.01 & \num{2} & \num{2} & \num{2} & \num{3} \\
thuemorse.02 & \num{4} & \num{2} & \num{4} & \num{5} \\
thuemorse.03 & \num{8} & \num{3} & \num{5} & \num{7} \\
thuemorse.04 & \num{16} & \num{4} & \num{6} & \num{9} \\
thuemorse.05 & \num{32} & \num{4} & \num{7} & \num{11} \\
thuemorse.06 & \num{64} & \num{4} & \num{8} & \num{13} \\
thuemorse.07 & \num{128} & \num{4} & \num{9} & \num{15} \\
thuemorse.08 & \num{256} & \num{4} & no time & \num{17} \\
thuemorse.09 & \num{512} & \num{4} & no mem & no mem \\
thuemorse.10 & \num{1024} & \num{4} & no mem & no mem \\
thuemorse.11 & \num{2048} & \num{4} & no mem & no mem \\
thuemorse.12 & \num{4096} & \num{4} & no mem & no mem \\
thuemorse.13 & \num{8192} & \num{4} & no mem & no mem \\
thuemorse.14 & \num{16384} & \num{4} & no mem & no mem \\
thuemorse.15 & \num{32768} & \num{4} & no mem & no mem \\
thuemorse.16 & \num{65536} & \num{4} & no mem & no mem \\
thuemorse.17 & \num{131072} & \num{4} & no mem & no time \\
thuemorse.18 & \num{262144} & \num{4} & no mem & no time \\
thuemorse.19 & \num{524288} & \num{4} & no mem & no time \\
thuemorse.20 & \num{1048576} & \num{4} & no mem & no time \\
xargs.1 & \num{4227} & \num{696} & no mem & no mem \\
 \end{longtable}

\begin{figure}
	\centering{\includegraphics[width=0.4\linewidth,page=4]{plot/plot}
\includegraphics[width=0.4\linewidth,page=5]{plot/plot}
\includegraphics[width=0.4\linewidth,page=6]{plot/plot}
\includegraphics[width=0.4\linewidth,page=7]{plot/plot}
\includegraphics[width=0.4\linewidth,page=8]{plot/plot}
\includegraphics[width=0.4\linewidth,page=9]{plot/plot}
	}\begin{minipage}{0.65\linewidth}
	\caption{Plots for dataset \textsc{aaa.txt}.}
	\label{plot:aaa.txt}
    \end{minipage}
    \hfill
    \begin{minipage}{0.25\linewidth}
        \includegraphics[page=1]{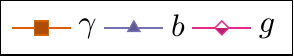}
    \end{minipage}
\end{figure}

\begin{figure}
	\centering{\includegraphics[width=0.4\linewidth,page=10]{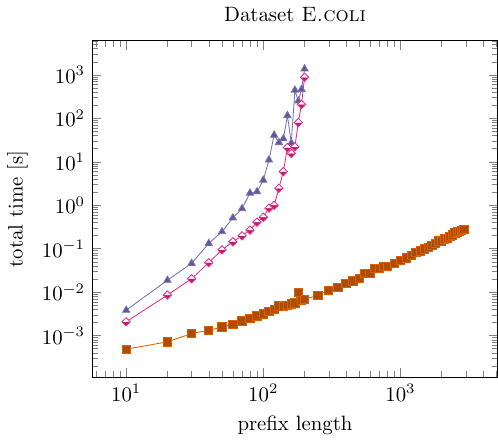}
\includegraphics[width=0.4\linewidth,page=11]{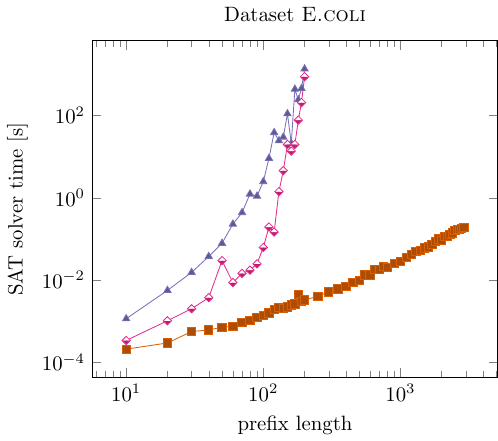}
\includegraphics[width=0.4\linewidth,page=12]{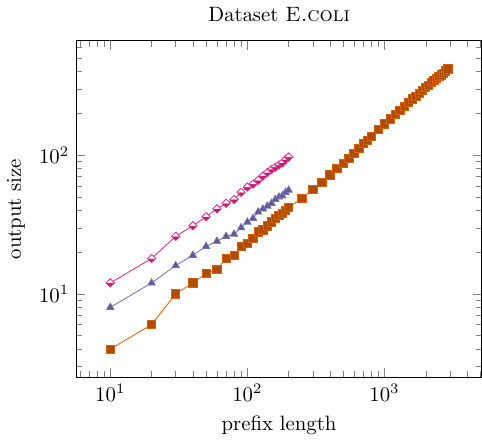}
\includegraphics[width=0.4\linewidth,page=13]{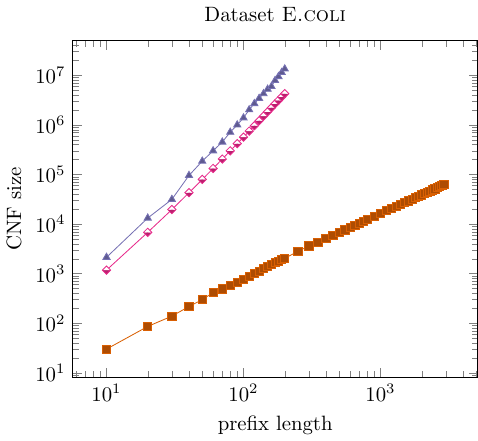}
\includegraphics[width=0.4\linewidth,page=14]{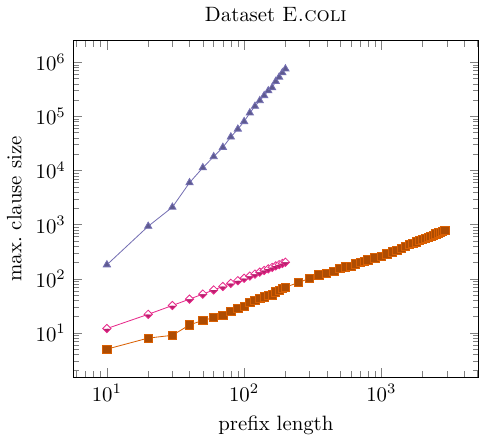}
\includegraphics[width=0.4\linewidth,page=15]{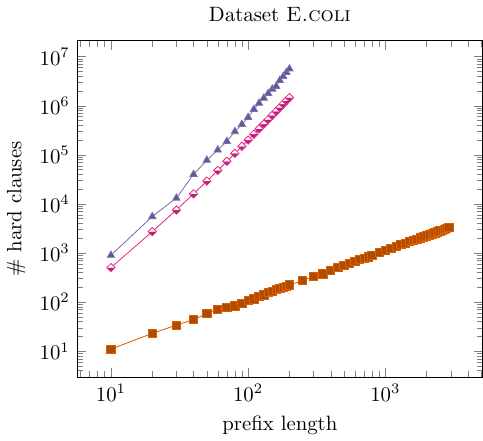}
	}\begin{minipage}{0.65\linewidth}
	\caption{Plots for dataset \textsc{E.coli}.}
	\label{plot:E.coli}
    \end{minipage}
    \hfill
    \begin{minipage}{0.25\linewidth}
        \includegraphics[page=1]{plot/plot}
    \end{minipage}
\end{figure}

\begin{figure}
	\centering{\includegraphics[width=0.4\linewidth,page=16]{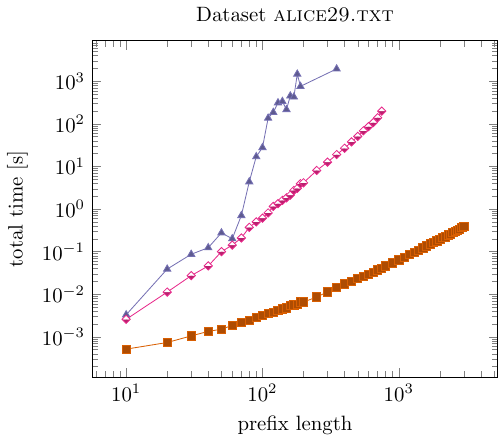}
\includegraphics[width=0.4\linewidth,page=17]{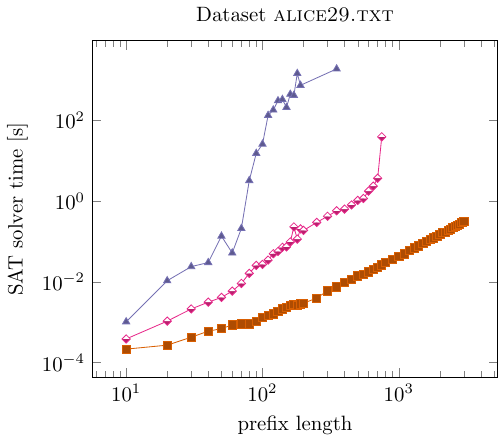}
\includegraphics[width=0.4\linewidth,page=18]{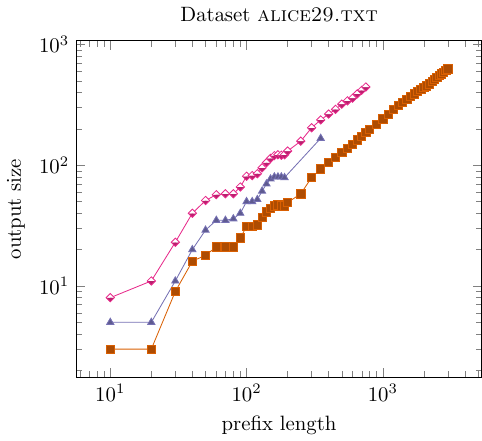}
\includegraphics[width=0.4\linewidth,page=19]{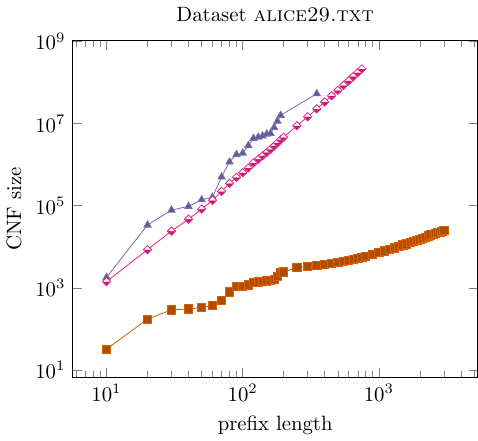}
\includegraphics[width=0.4\linewidth,page=20]{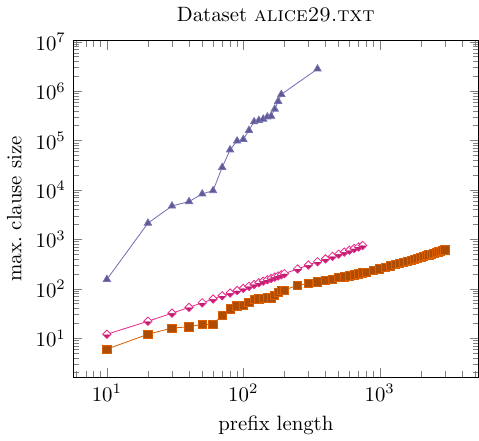}
\includegraphics[width=0.4\linewidth,page=21]{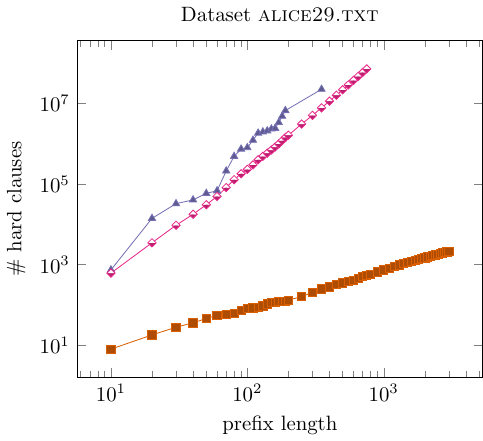}
	}\begin{minipage}{0.65\linewidth}
	\caption{Plots for dataset \textsc{alice29.txt}.}
	\label{plot:alice29.txt}
    \end{minipage}
    \hfill
    \begin{minipage}{0.25\linewidth}
        \includegraphics[page=1]{plot/plot}
    \end{minipage}
\end{figure}

\begin{figure}
	\centering{\includegraphics[width=0.4\linewidth,page=22]{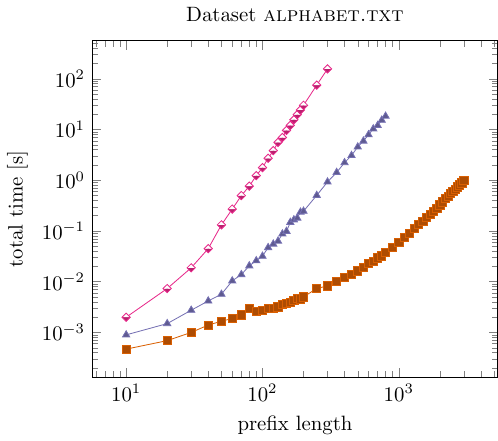}
\includegraphics[width=0.4\linewidth,page=23]{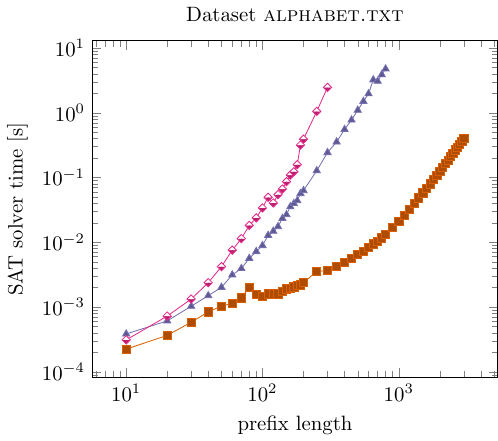}
\includegraphics[width=0.4\linewidth,page=24]{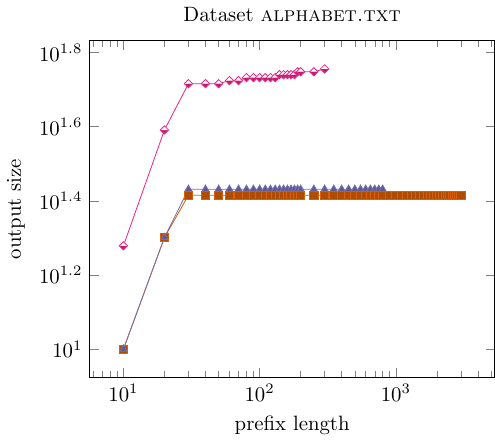}
\includegraphics[width=0.4\linewidth,page=25]{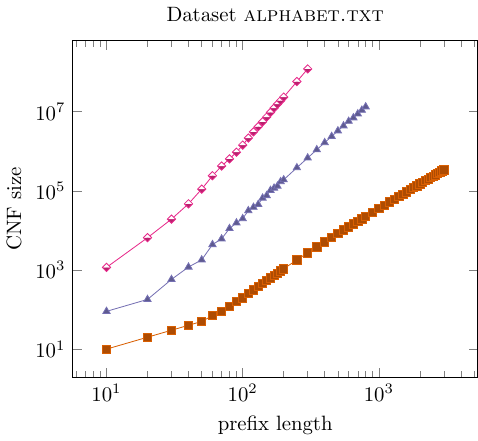}
\includegraphics[width=0.4\linewidth,page=26]{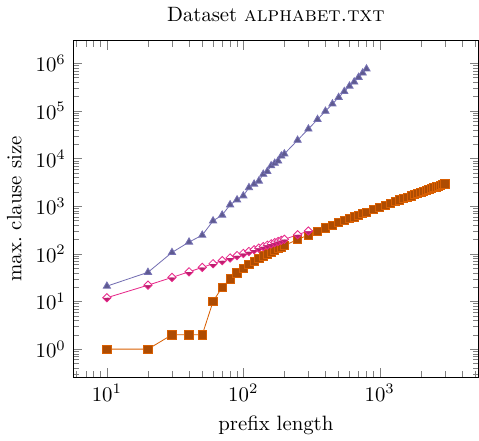}
\includegraphics[width=0.4\linewidth,page=27]{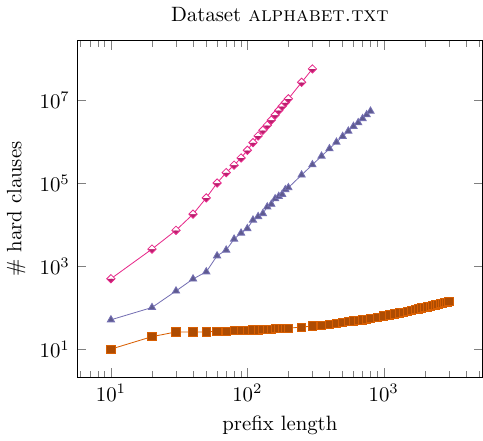}
	}\begin{minipage}{0.65\linewidth}
	\caption{Plots for dataset \textsc{alphabet.txt}.}
	\label{plot:alphabet.txt}
    \end{minipage}
    \hfill
    \begin{minipage}{0.25\linewidth}
        \includegraphics[page=1]{plot/plot}
    \end{minipage}
\end{figure}

\begin{figure}
	\centering{\includegraphics[width=0.4\linewidth,page=28]{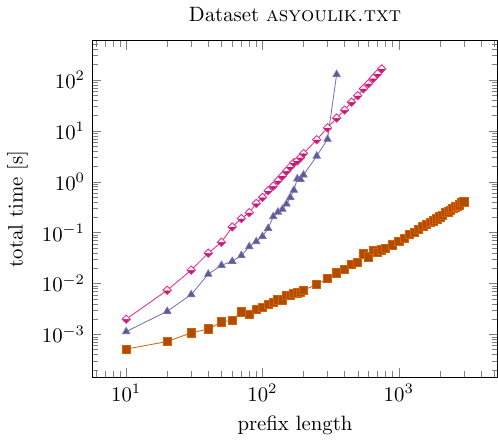}
\includegraphics[width=0.4\linewidth,page=29]{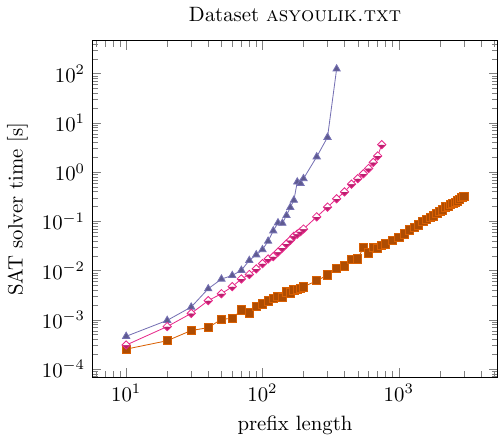}
\includegraphics[width=0.4\linewidth,page=30]{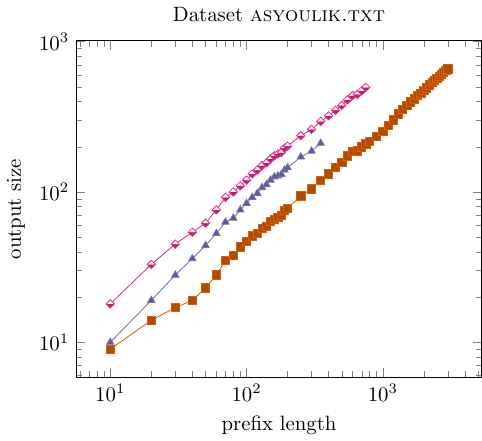}
\includegraphics[width=0.4\linewidth,page=31]{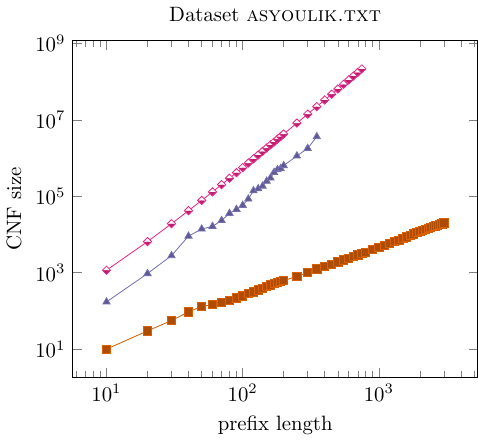}
\includegraphics[width=0.4\linewidth,page=32]{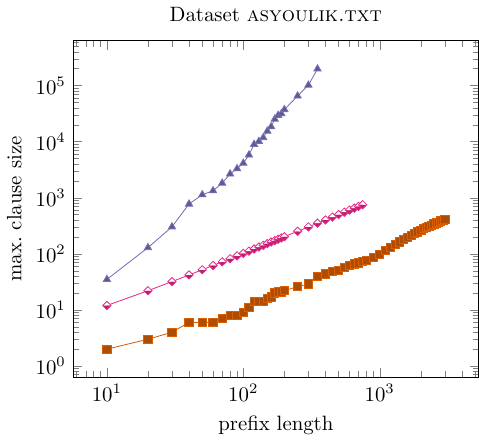}
\includegraphics[width=0.4\linewidth,page=33]{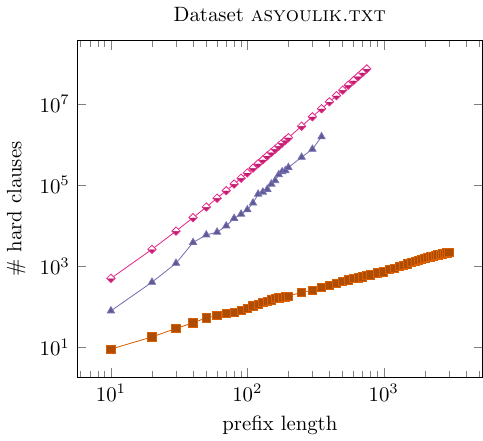}
	}\begin{minipage}{0.65\linewidth}
	\caption{Plots for dataset \textsc{asyoulik.txt}.}
	\label{plot:asyoulik.txt}
    \end{minipage}
    \hfill
    \begin{minipage}{0.25\linewidth}
        \includegraphics[page=1]{plot/plot}
    \end{minipage}
\end{figure}

\begin{figure}
	\centering{\includegraphics[width=0.4\linewidth,page=34]{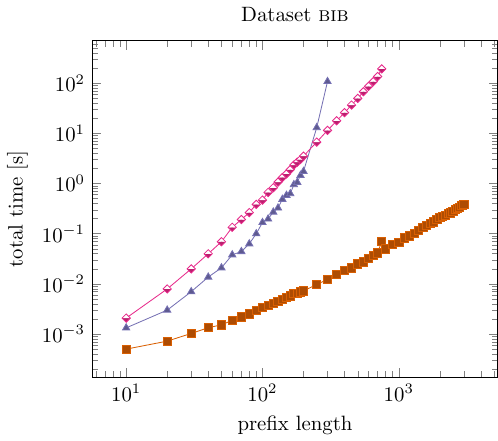}
\includegraphics[width=0.4\linewidth,page=35]{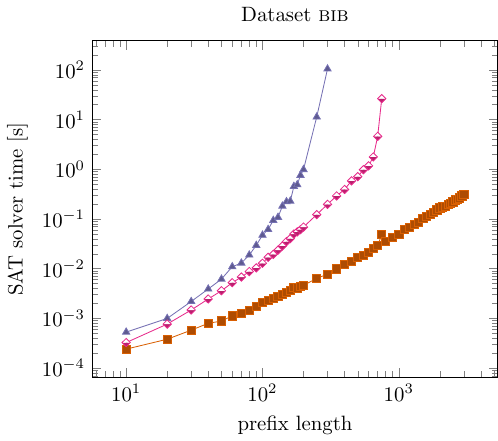}
\includegraphics[width=0.4\linewidth,page=36]{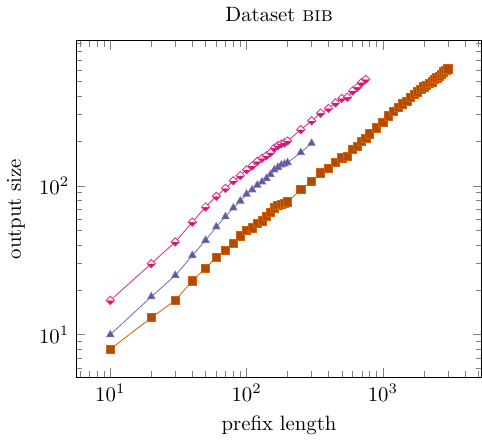}
\includegraphics[width=0.4\linewidth,page=37]{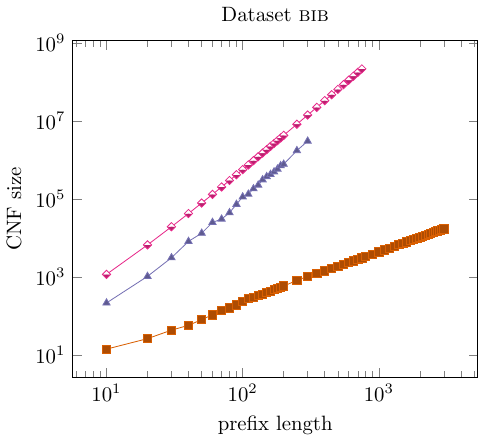}
\includegraphics[width=0.4\linewidth,page=38]{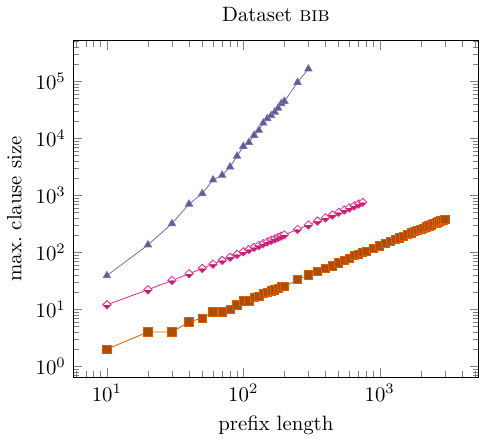}
\includegraphics[width=0.4\linewidth,page=39]{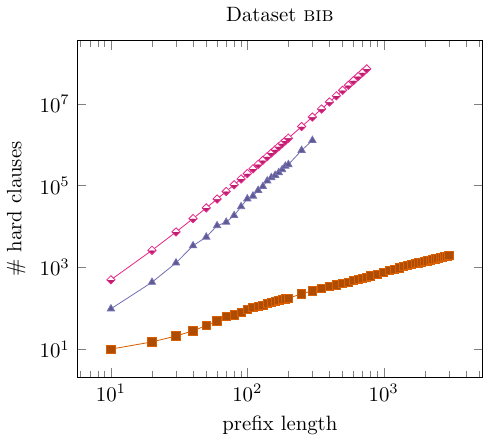}
	}\begin{minipage}{0.65\linewidth}
	\caption{Plots for dataset \textsc{bib}.}
	\label{plot:bib}
    \end{minipage}
    \hfill
    \begin{minipage}{0.25\linewidth}
        \includegraphics[page=1]{plot/plot}
    \end{minipage}
\end{figure}

\begin{figure}
	\centering{\includegraphics[width=0.4\linewidth,page=40]{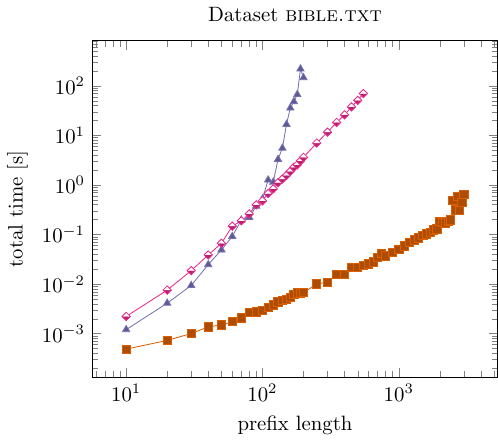}
\includegraphics[width=0.4\linewidth,page=41]{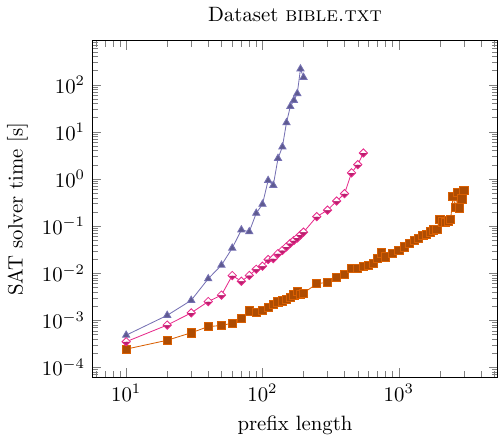}
\includegraphics[width=0.4\linewidth,page=42]{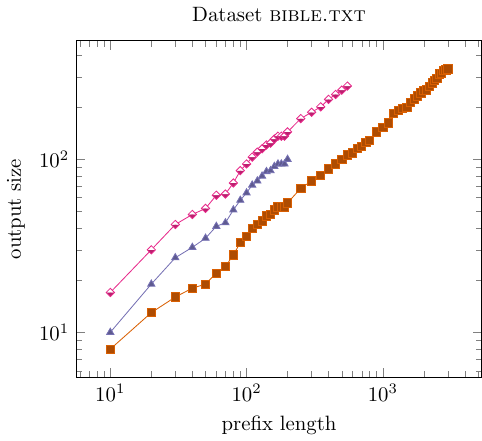}
\includegraphics[width=0.4\linewidth,page=43]{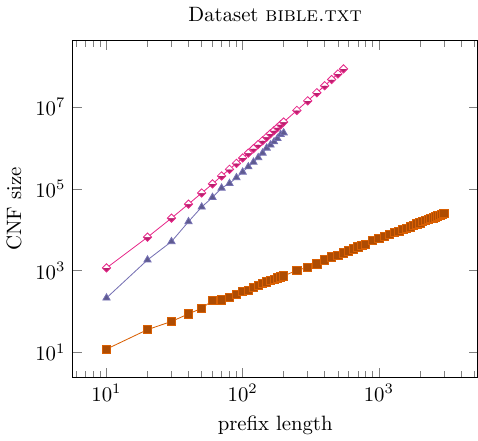}
\includegraphics[width=0.4\linewidth,page=44]{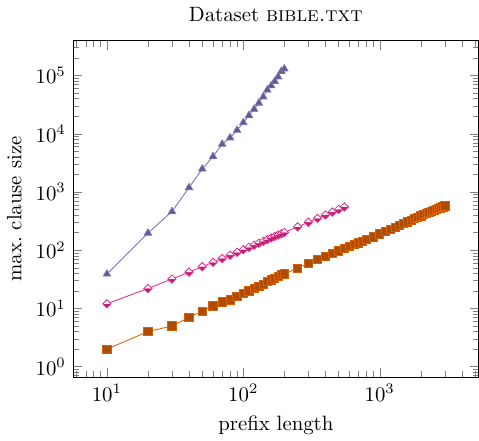}
\includegraphics[width=0.4\linewidth,page=45]{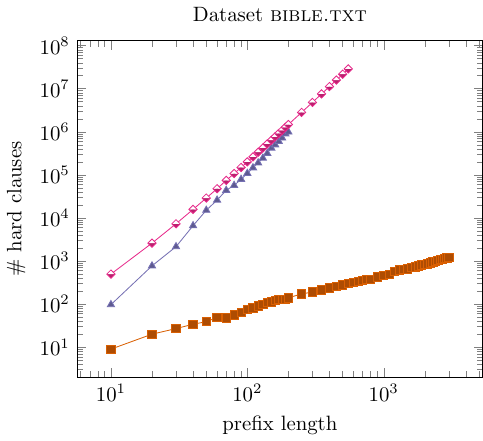}
	}\begin{minipage}{0.65\linewidth}
	\caption{Plots for dataset \textsc{bible.txt}.}
	\label{plot:bible.txt}
    \end{minipage}
    \hfill
    \begin{minipage}{0.25\linewidth}
        \includegraphics[page=1]{plot/plot}
    \end{minipage}
\end{figure}

\begin{figure}
	\centering{\includegraphics[width=0.4\linewidth,page=46]{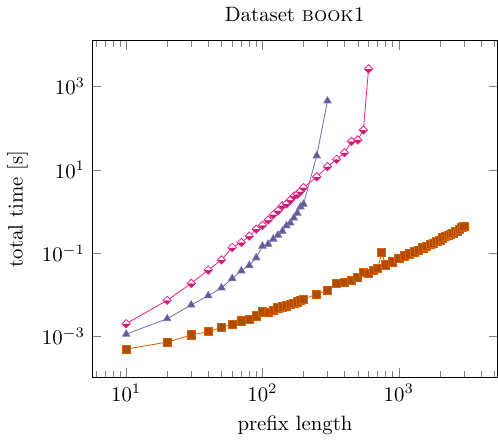}
\includegraphics[width=0.4\linewidth,page=47]{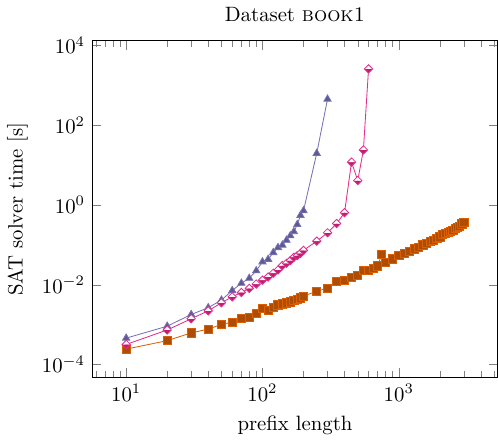}
\includegraphics[width=0.4\linewidth,page=48]{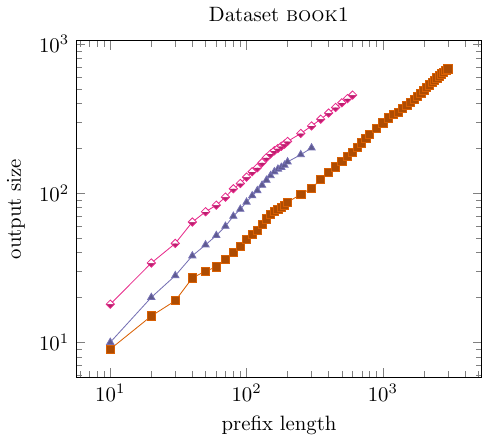}
\includegraphics[width=0.4\linewidth,page=49]{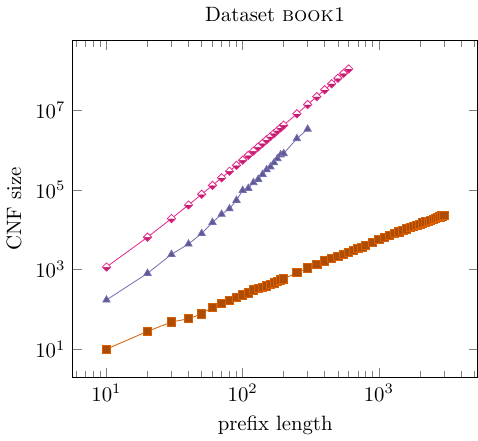}
\includegraphics[width=0.4\linewidth,page=50]{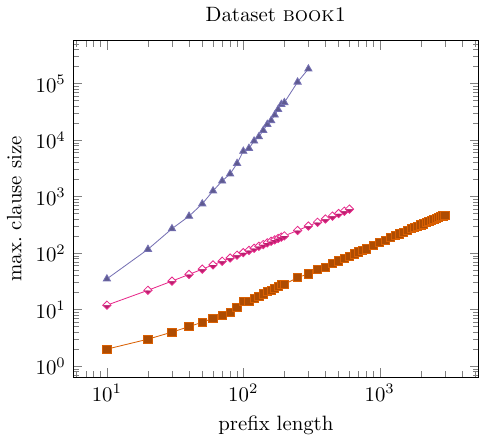}
\includegraphics[width=0.4\linewidth,page=51]{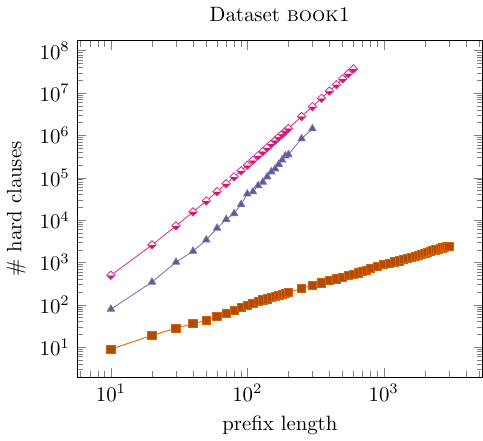}
	}\begin{minipage}{0.65\linewidth}
	\caption{Plots for dataset \textsc{book1}.}
	\label{plot:book1}
    \end{minipage}
    \hfill
    \begin{minipage}{0.25\linewidth}
        \includegraphics[page=1]{plot/plot}
    \end{minipage}
\end{figure}

\begin{figure}
	\centering{\includegraphics[width=0.4\linewidth,page=52]{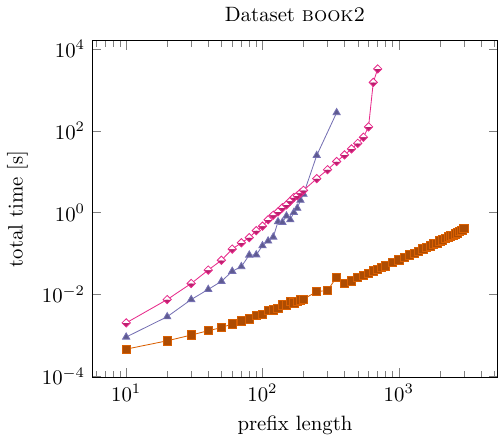}
\includegraphics[width=0.4\linewidth,page=53]{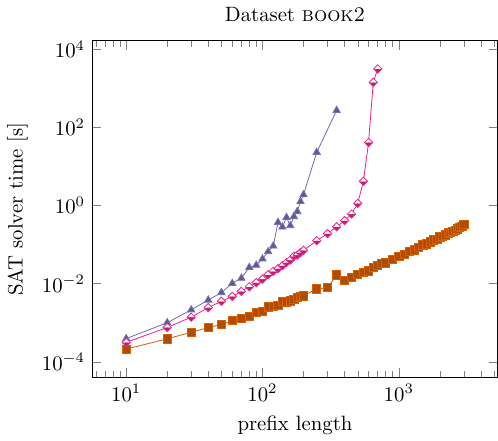}
\includegraphics[width=0.4\linewidth,page=54]{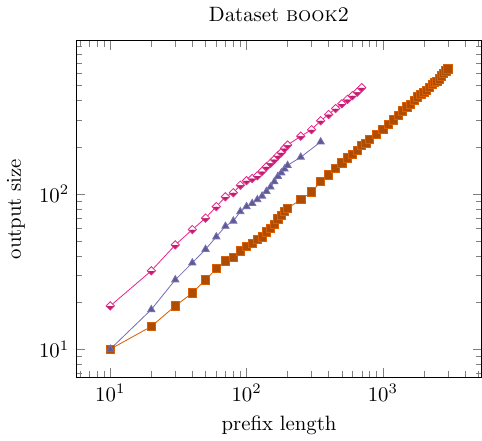}
\includegraphics[width=0.4\linewidth,page=55]{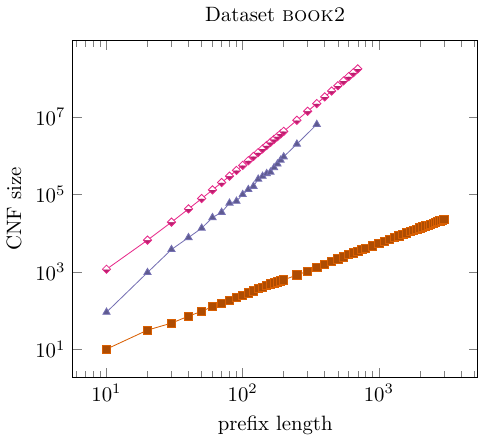}
\includegraphics[width=0.4\linewidth,page=56]{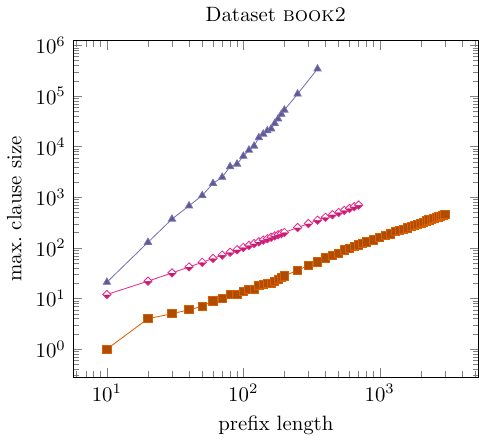}
\includegraphics[width=0.4\linewidth,page=57]{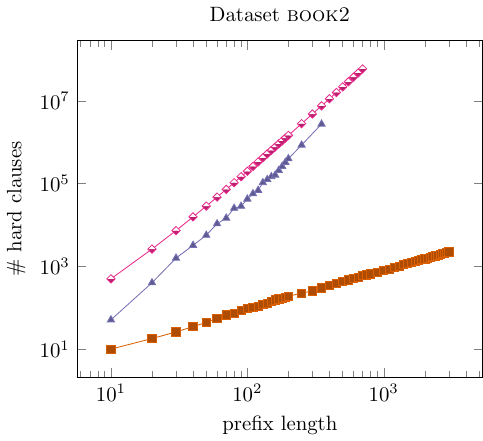}
	}\begin{minipage}{0.65\linewidth}
	\caption{Plots for dataset \textsc{book2}.}
	\label{plot:book2}
    \end{minipage}
    \hfill
    \begin{minipage}{0.25\linewidth}
        \includegraphics[page=1]{plot/plot}
    \end{minipage}
\end{figure}

\begin{figure}
	\centering{\includegraphics[width=0.4\linewidth,page=58]{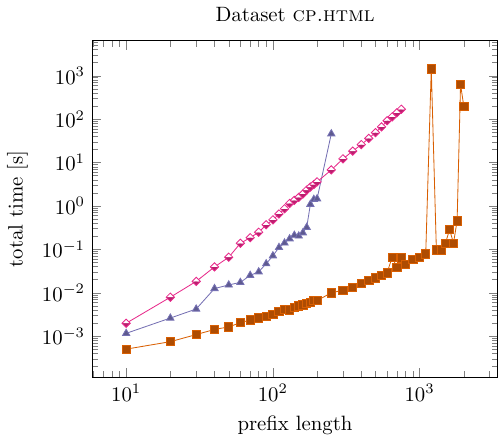}
\includegraphics[width=0.4\linewidth,page=59]{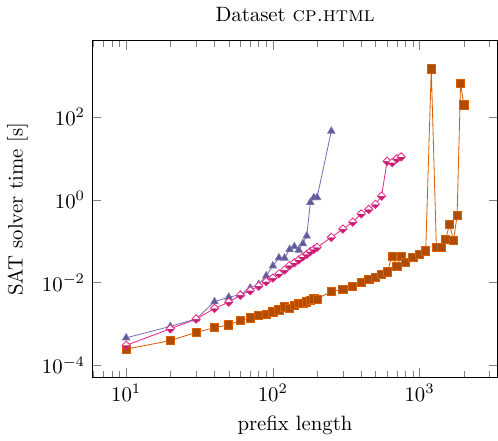}
\includegraphics[width=0.4\linewidth,page=60]{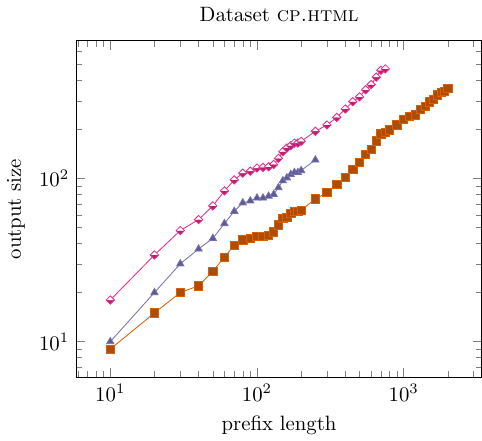}
\includegraphics[width=0.4\linewidth,page=61]{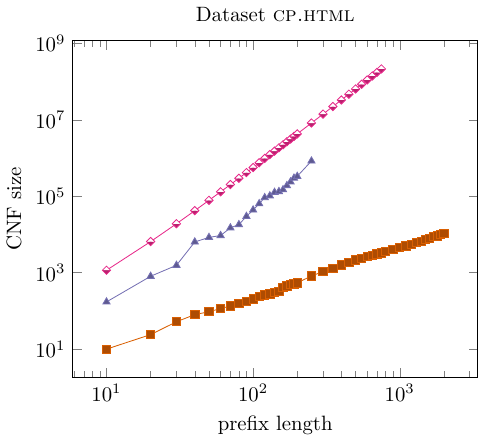}
\includegraphics[width=0.4\linewidth,page=62]{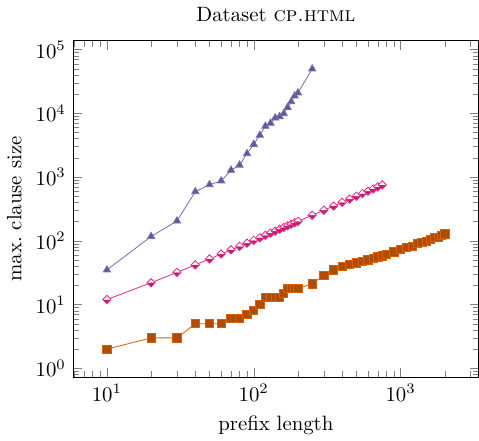}
\includegraphics[width=0.4\linewidth,page=63]{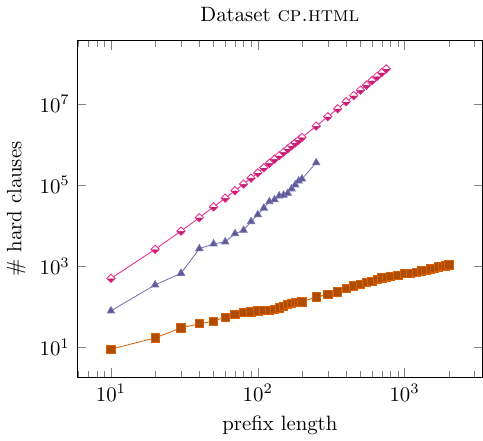}
	}\begin{minipage}{0.65\linewidth}
	\caption{Plots for dataset \textsc{cp.html}.}
	\label{plot:cp.html}
    \end{minipage}
    \hfill
    \begin{minipage}{0.25\linewidth}
        \includegraphics[page=1]{plot/plot}
    \end{minipage}
\end{figure}

\begin{figure}
	\centering{\includegraphics[width=0.4\linewidth,page=64]{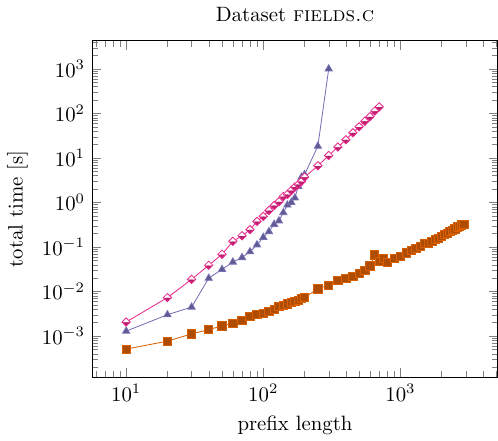}
\includegraphics[width=0.4\linewidth,page=65]{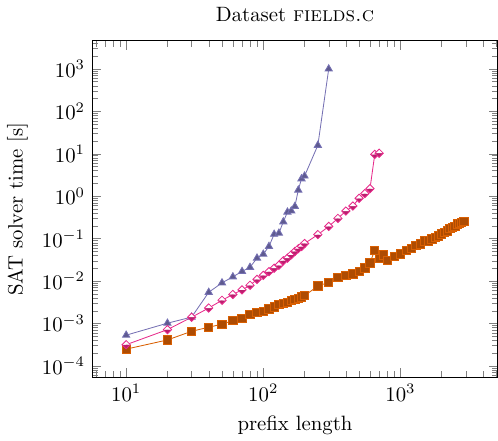}
\includegraphics[width=0.4\linewidth,page=66]{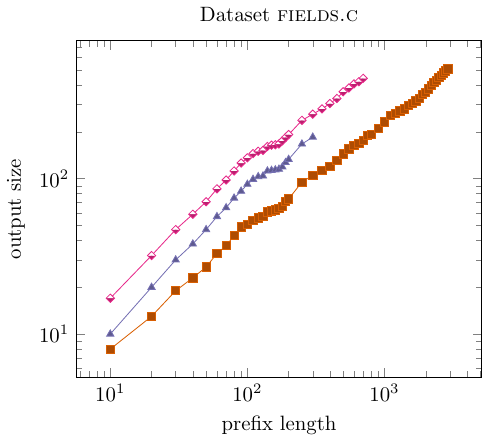}
\includegraphics[width=0.4\linewidth,page=67]{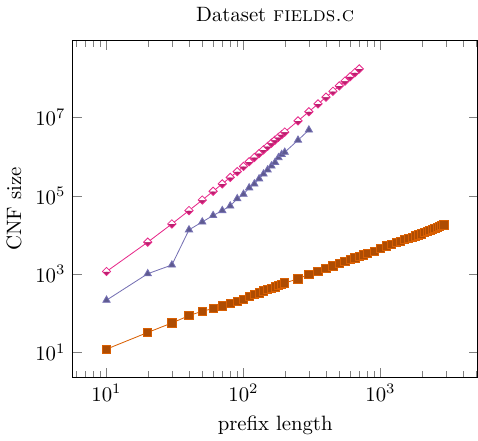}
\includegraphics[width=0.4\linewidth,page=68]{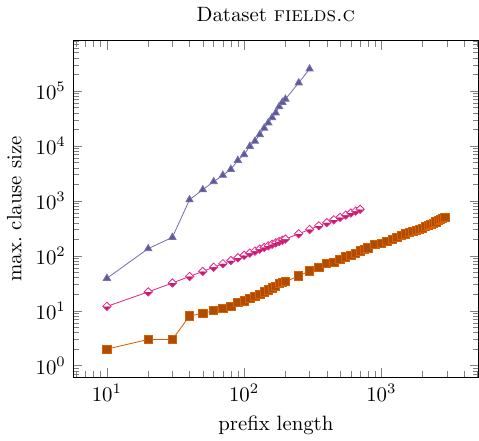}
\includegraphics[width=0.4\linewidth,page=69]{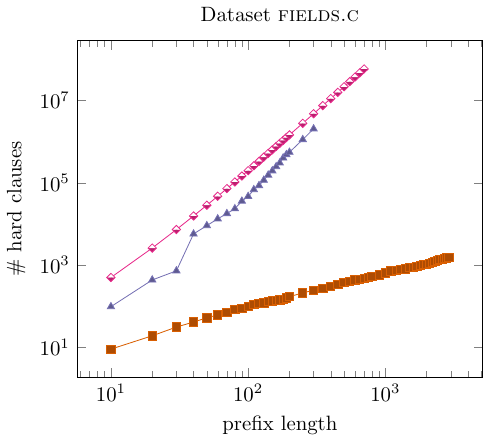}
	}\begin{minipage}{0.65\linewidth}
	\caption{Plots for dataset \textsc{fields.c}.}
	\label{plot:fields.c}
    \end{minipage}
    \hfill
    \begin{minipage}{0.25\linewidth}
        \includegraphics[page=1]{plot/plot}
    \end{minipage}
\end{figure}

\begin{figure}
	\centering{\includegraphics[width=0.4\linewidth,page=70]{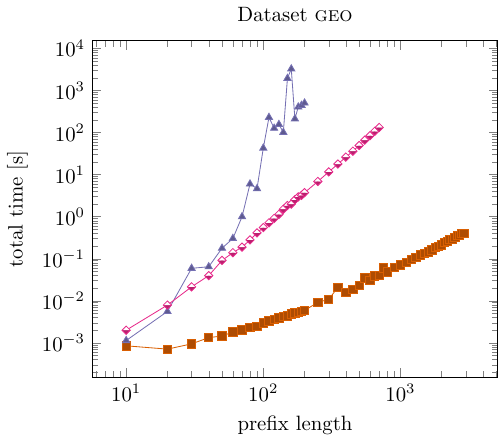}
\includegraphics[width=0.4\linewidth,page=71]{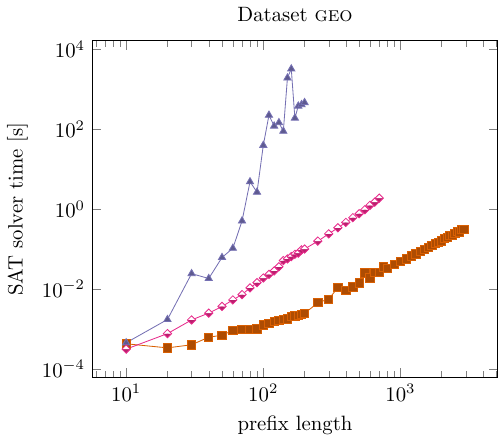}
\includegraphics[width=0.4\linewidth,page=72]{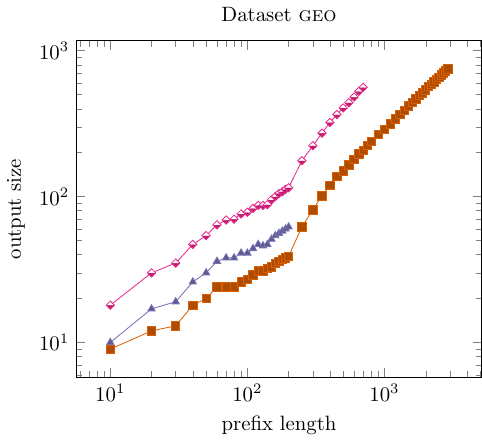}
\includegraphics[width=0.4\linewidth,page=73]{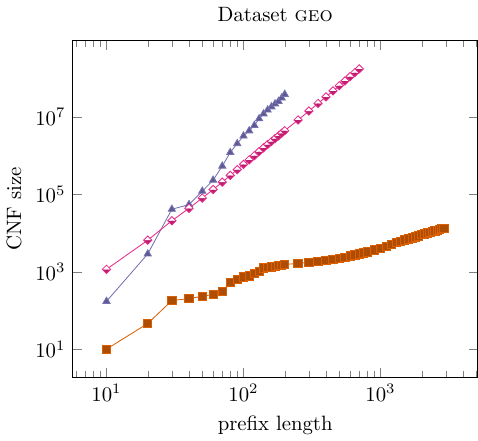}
\includegraphics[width=0.4\linewidth,page=74]{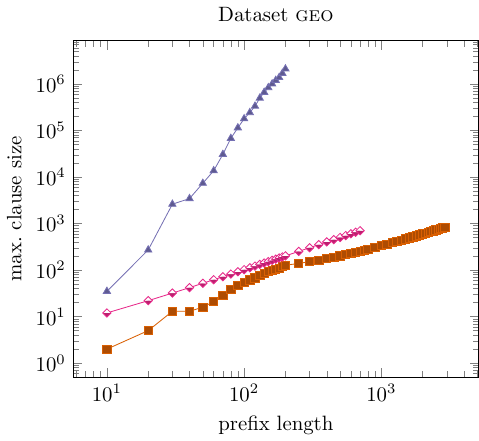}
\includegraphics[width=0.4\linewidth,page=75]{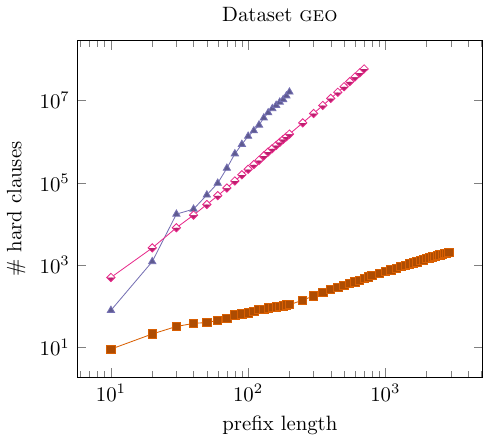}
	}\begin{minipage}{0.65\linewidth}
	\caption{Plots for dataset \textsc{geo}.}
	\label{plot:geo}
    \end{minipage}
    \hfill
    \begin{minipage}{0.25\linewidth}
        \includegraphics[page=1]{plot/plot}
    \end{minipage}
\end{figure}

\begin{figure}
	\centering{\includegraphics[width=0.4\linewidth,page=76]{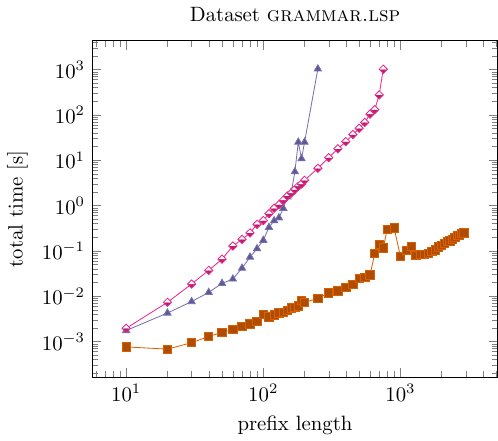}
\includegraphics[width=0.4\linewidth,page=77]{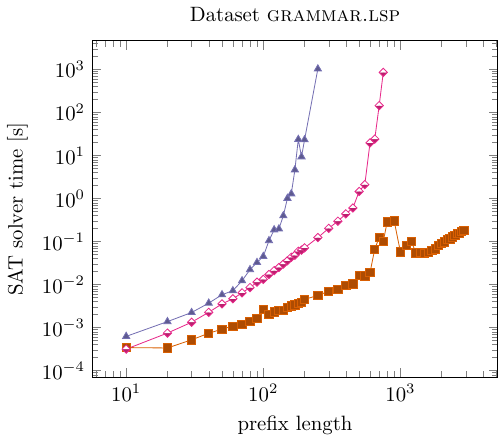}
\includegraphics[width=0.4\linewidth,page=78]{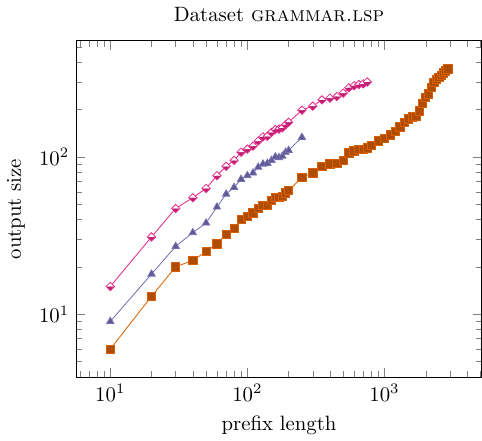}
\includegraphics[width=0.4\linewidth,page=79]{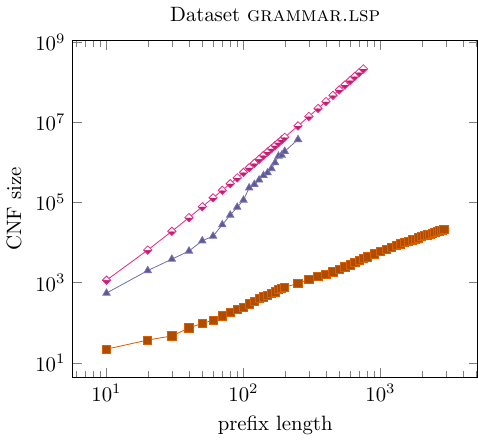}
\includegraphics[width=0.4\linewidth,page=80]{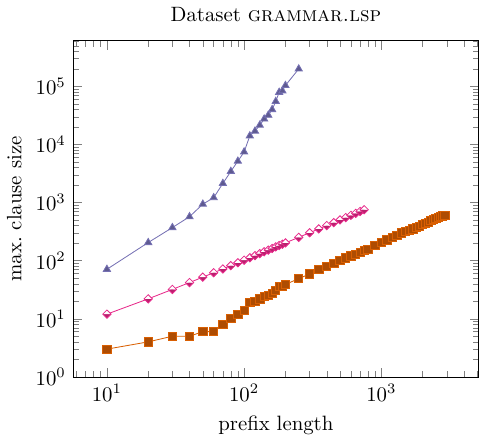}
\includegraphics[width=0.4\linewidth,page=81]{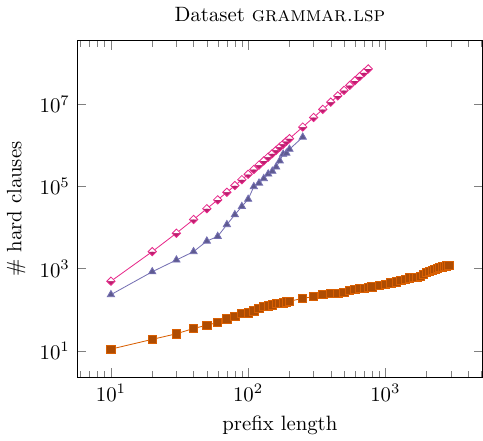}
	}\begin{minipage}{0.65\linewidth}
	\caption{Plots for dataset \textsc{grammar.lsp}.}
	\label{plot:grammar.lsp}
    \end{minipage}
    \hfill
    \begin{minipage}{0.25\linewidth}
        \includegraphics[page=1]{plot/plot}
    \end{minipage}
\end{figure}

\begin{figure}
	\centering{\includegraphics[width=0.4\linewidth,page=82]{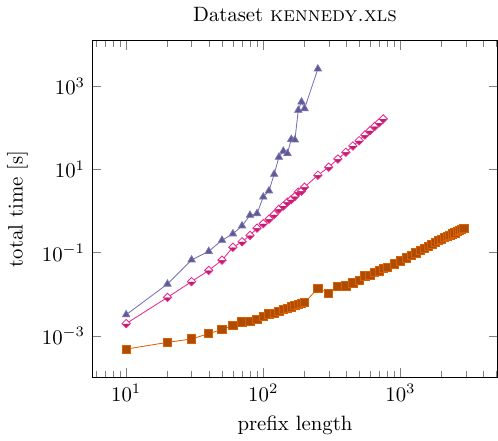}
\includegraphics[width=0.4\linewidth,page=83]{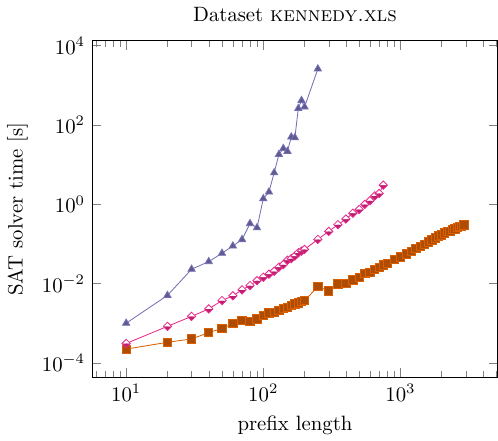}
\includegraphics[width=0.4\linewidth,page=84]{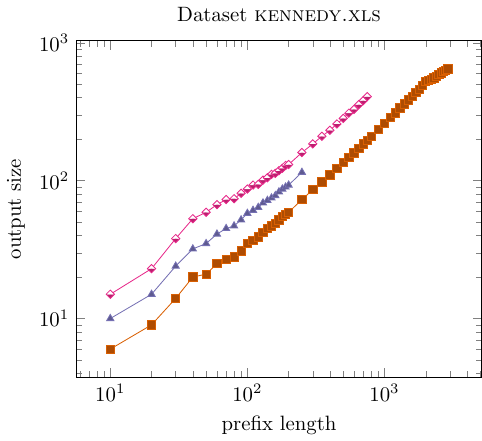}
\includegraphics[width=0.4\linewidth,page=85]{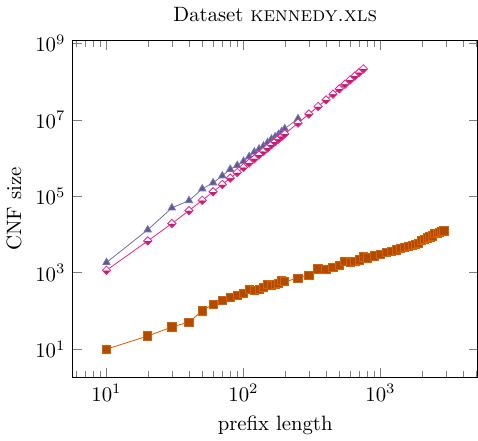}
\includegraphics[width=0.4\linewidth,page=86]{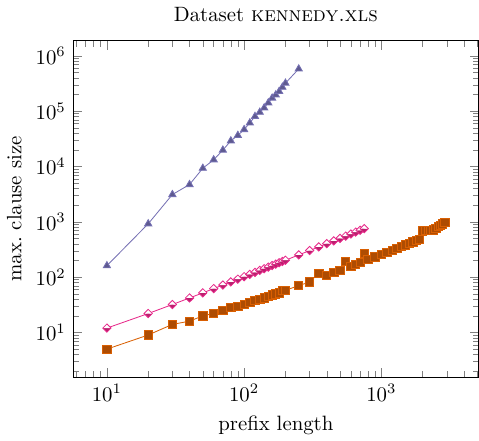}
\includegraphics[width=0.4\linewidth,page=87]{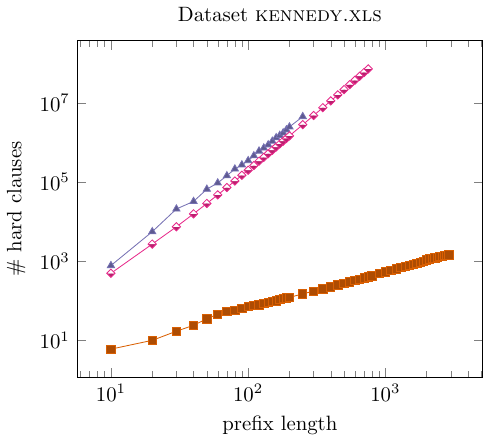}
	}\begin{minipage}{0.65\linewidth}
	\caption{Plots for dataset \textsc{kennedy.xls}.}
	\label{plot:kennedy.xls}
    \end{minipage}
    \hfill
    \begin{minipage}{0.25\linewidth}
        \includegraphics[page=1]{plot/plot}
    \end{minipage}
\end{figure}

\begin{figure}
	\centering{\includegraphics[width=0.4\linewidth,page=88]{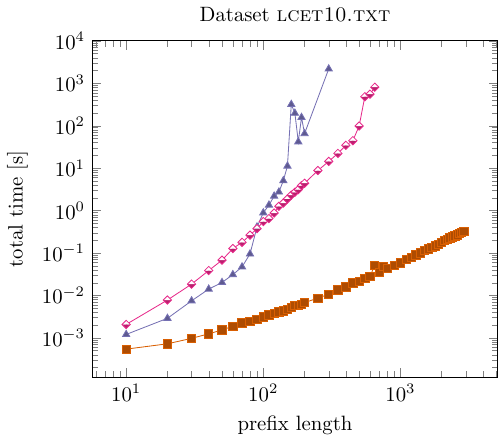}
\includegraphics[width=0.4\linewidth,page=89]{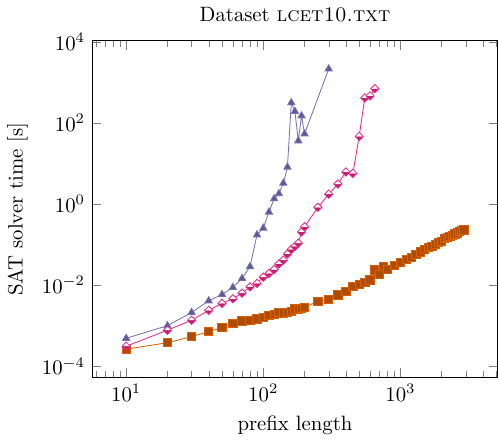}
\includegraphics[width=0.4\linewidth,page=90]{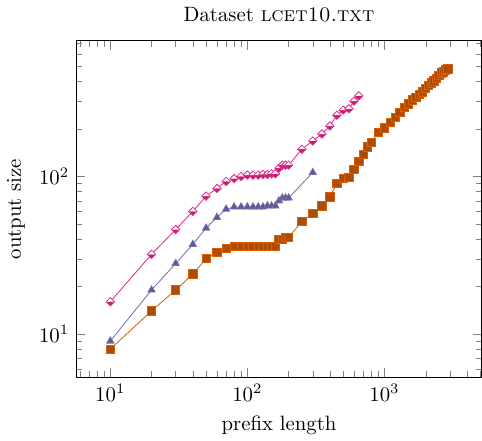}
\includegraphics[width=0.4\linewidth,page=91]{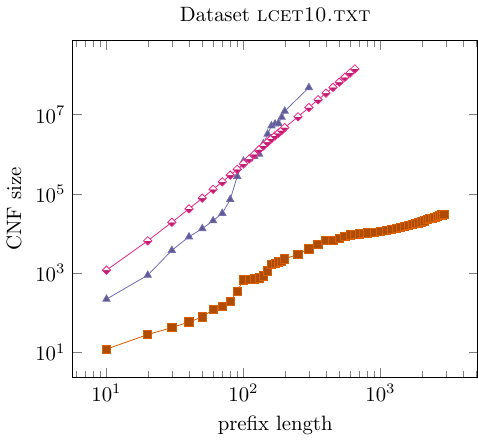}
\includegraphics[width=0.4\linewidth,page=92]{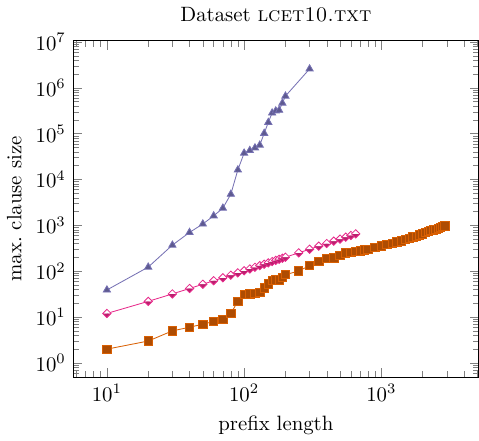}
\includegraphics[width=0.4\linewidth,page=93]{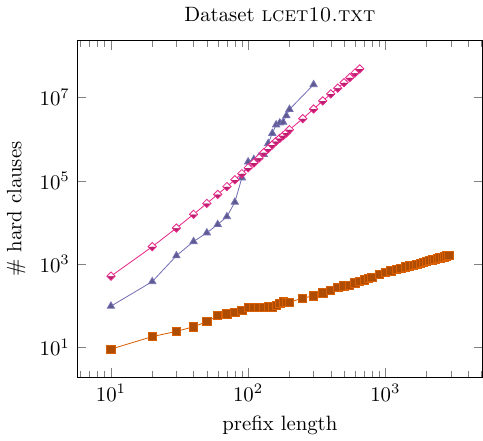}
	}\begin{minipage}{0.65\linewidth}
	\caption{Plots for dataset \textsc{lcet10.txt}.}
	\label{plot:lcet10.txt}
    \end{minipage}
    \hfill
    \begin{minipage}{0.25\linewidth}
        \includegraphics[page=1]{plot/plot}
    \end{minipage}
\end{figure}

\begin{figure}
	\centering{\includegraphics[width=0.4\linewidth,page=94]{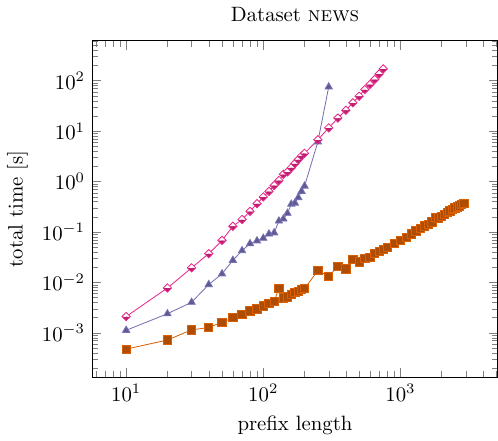}
\includegraphics[width=0.4\linewidth,page=95]{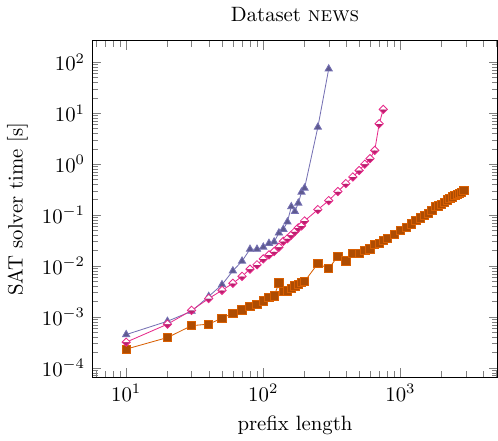}
\includegraphics[width=0.4\linewidth,page=96]{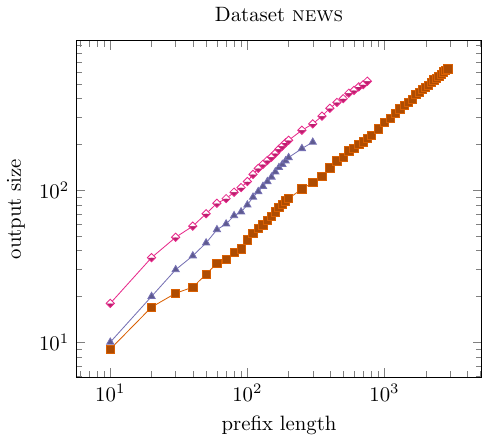}
\includegraphics[width=0.4\linewidth,page=97]{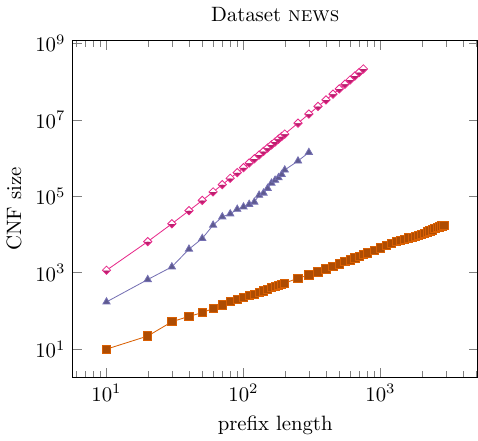}
\includegraphics[width=0.4\linewidth,page=98]{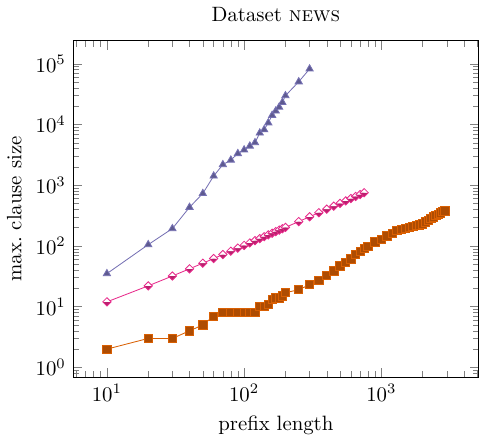}
\includegraphics[width=0.4\linewidth,page=99]{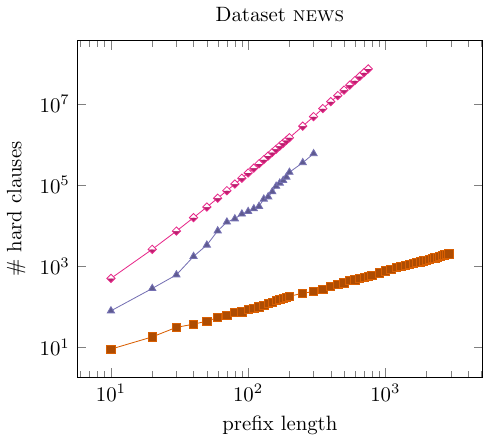}
	}\begin{minipage}{0.65\linewidth}
	\caption{Plots for dataset \textsc{news}.}
	\label{plot:news}
    \end{minipage}
    \hfill
    \begin{minipage}{0.25\linewidth}
        \includegraphics[page=1]{plot/plot}
    \end{minipage}
\end{figure}

\begin{figure}
	\centering{\includegraphics[width=0.4\linewidth,page=100]{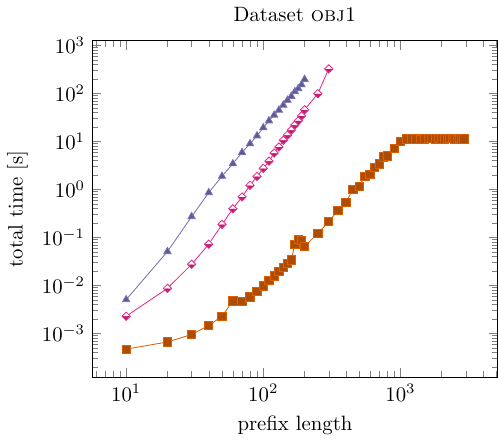}
\includegraphics[width=0.4\linewidth,page=101]{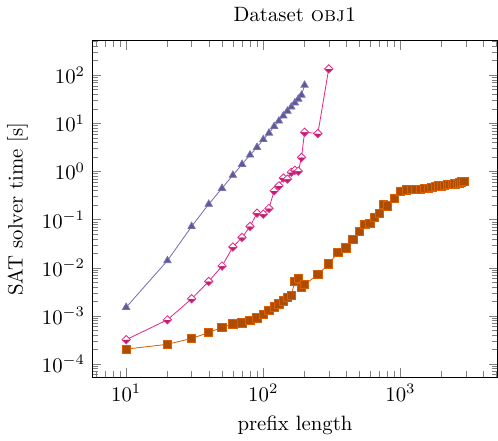}
\includegraphics[width=0.4\linewidth,page=102]{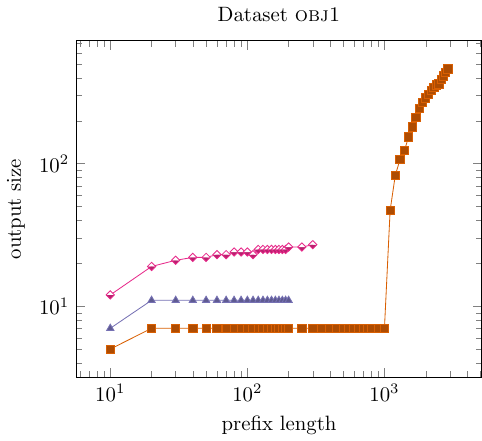}
\includegraphics[width=0.4\linewidth,page=103]{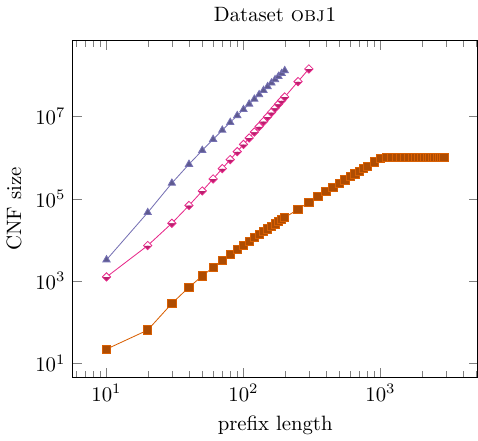}
\includegraphics[width=0.4\linewidth,page=104]{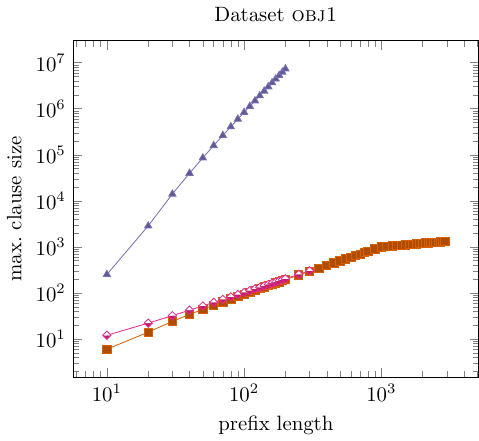}
\includegraphics[width=0.4\linewidth,page=105]{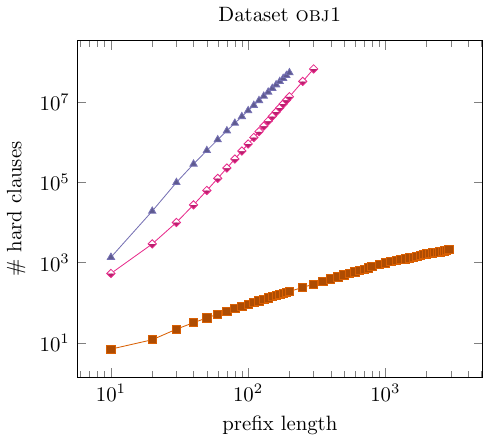}
	}\begin{minipage}{0.65\linewidth}
	\caption{Plots for dataset \textsc{obj1}.}
	\label{plot:obj1}
    \end{minipage}
    \hfill
    \begin{minipage}{0.25\linewidth}
        \includegraphics[page=1]{plot/plot}
    \end{minipage}
\end{figure}

\begin{figure}
	\centering{\includegraphics[width=0.4\linewidth,page=106]{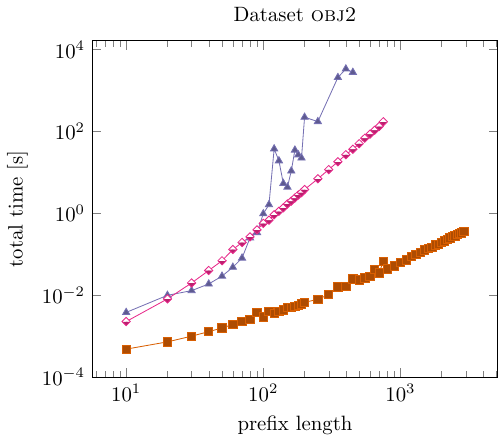}
\includegraphics[width=0.4\linewidth,page=107]{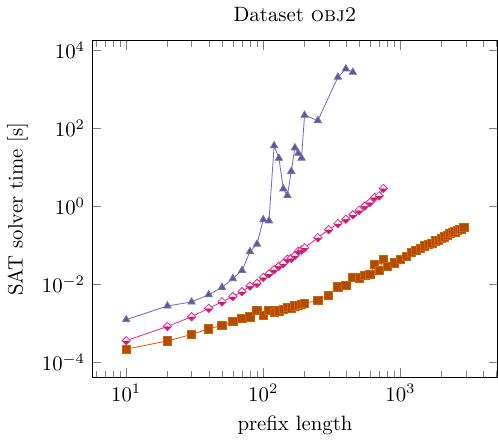}
\includegraphics[width=0.4\linewidth,page=108]{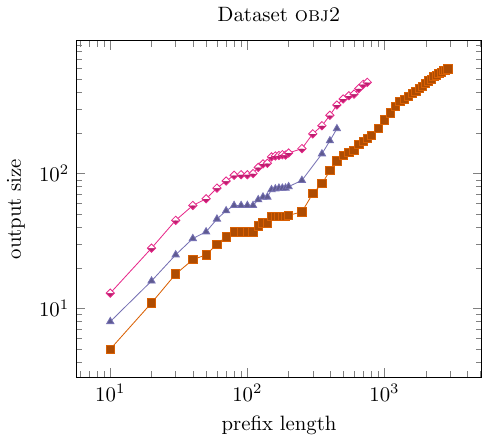}
\includegraphics[width=0.4\linewidth,page=109]{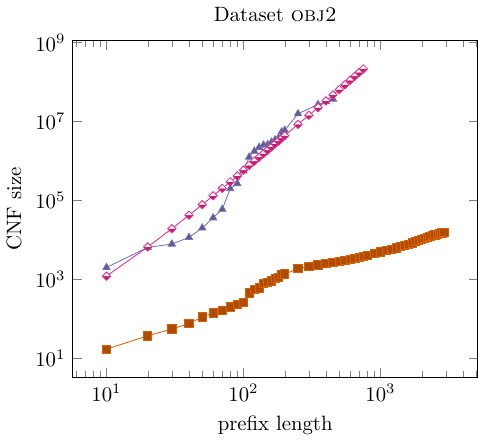}
\includegraphics[width=0.4\linewidth,page=110]{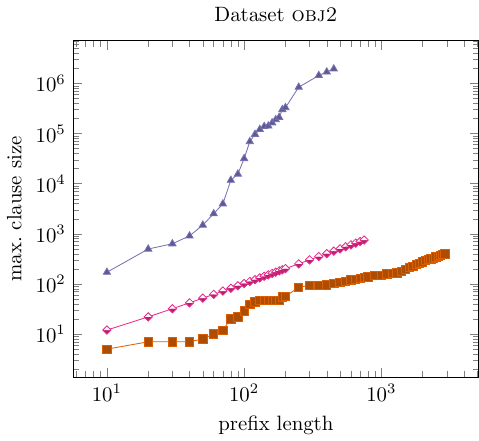}
\includegraphics[width=0.4\linewidth,page=111]{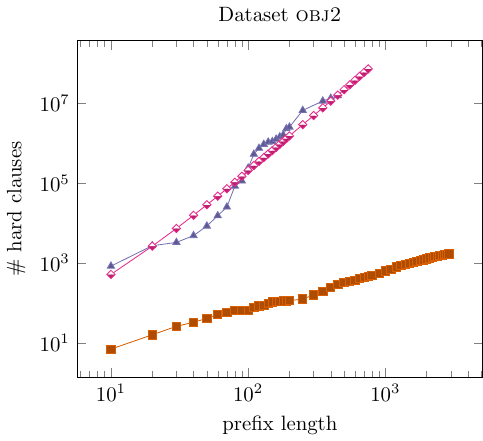}
	}\begin{minipage}{0.65\linewidth}
	\caption{Plots for dataset \textsc{obj2}.}
	\label{plot:obj2}
    \end{minipage}
    \hfill
    \begin{minipage}{0.25\linewidth}
        \includegraphics[page=1]{plot/plot}
    \end{minipage}
\end{figure}

\begin{figure}
	\centering{\includegraphics[width=0.4\linewidth,page=112]{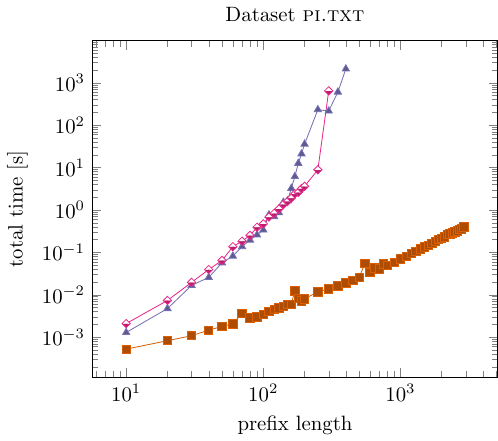}
\includegraphics[width=0.4\linewidth,page=113]{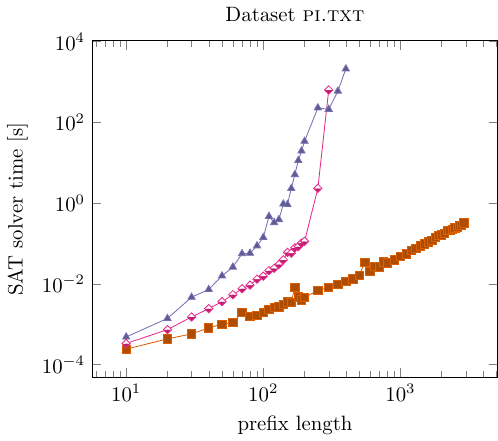}
\includegraphics[width=0.4\linewidth,page=114]{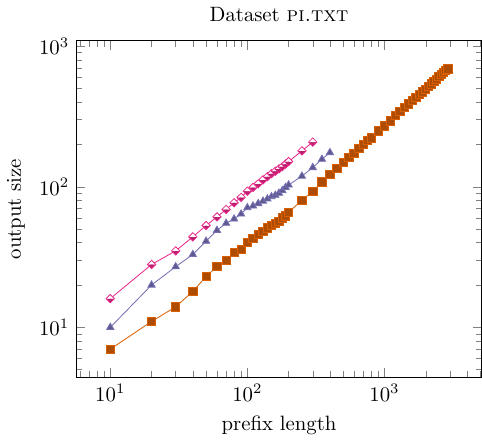}
\includegraphics[width=0.4\linewidth,page=115]{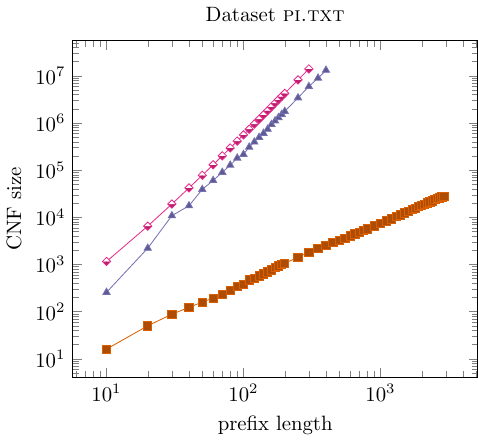}
\includegraphics[width=0.4\linewidth,page=116]{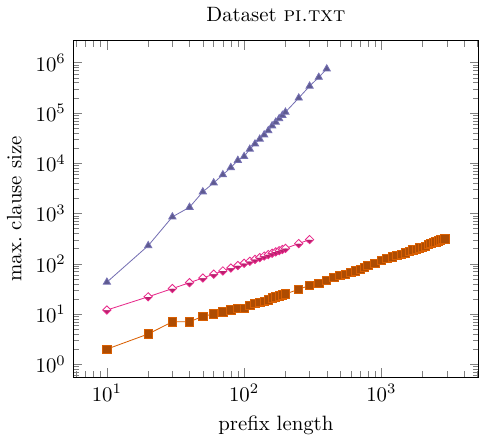}
\includegraphics[width=0.4\linewidth,page=117]{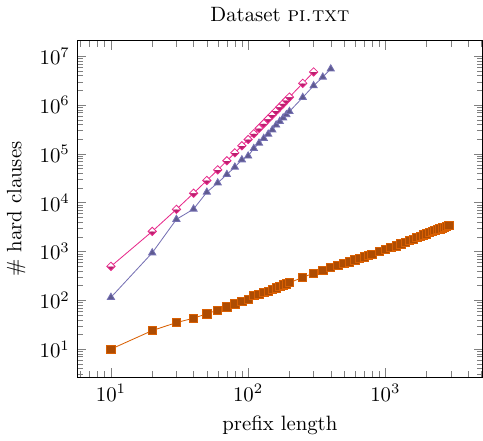}
	}\begin{minipage}{0.65\linewidth}
	\caption{Plots for dataset \textsc{pi.txt}.}
	\label{plot:pi.txt}
    \end{minipage}
    \hfill
    \begin{minipage}{0.25\linewidth}
        \includegraphics[page=1]{plot/plot}
    \end{minipage}
\end{figure}

\begin{figure}
	\centering{\includegraphics[width=0.4\linewidth,page=118]{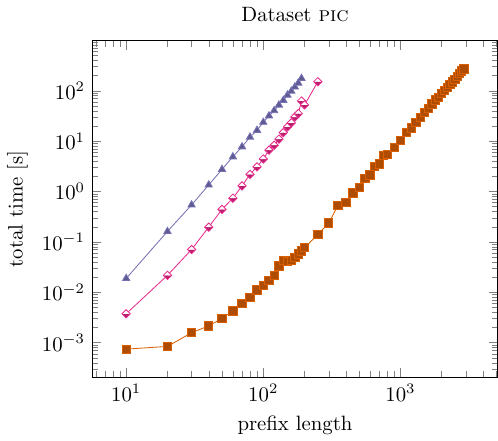}
\includegraphics[width=0.4\linewidth,page=119]{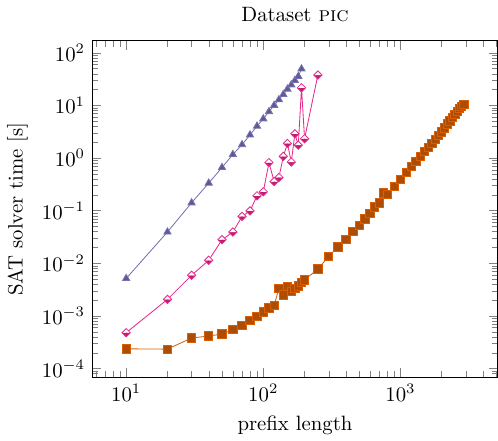}
\includegraphics[width=0.4\linewidth,page=120]{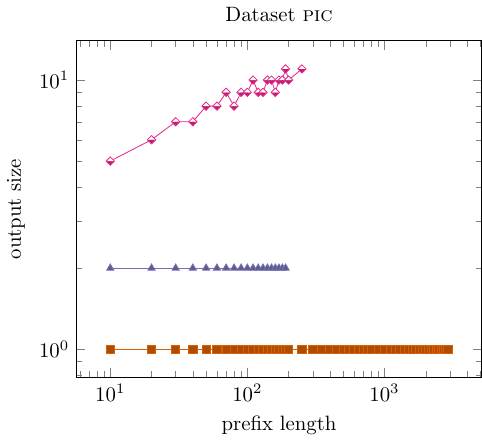}
\includegraphics[width=0.4\linewidth,page=121]{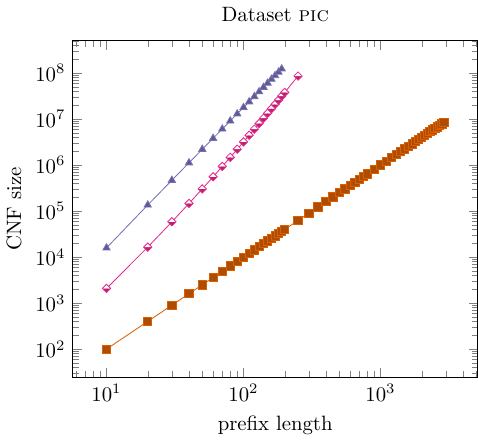}
\includegraphics[width=0.4\linewidth,page=122]{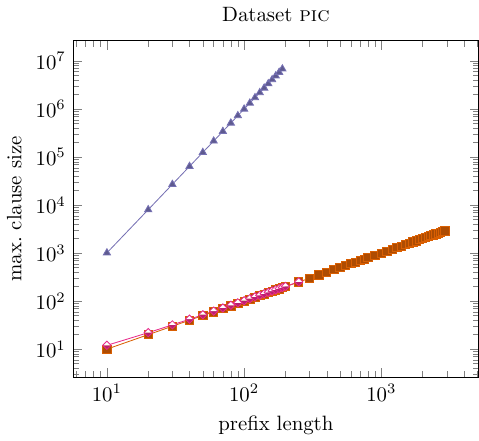}
\includegraphics[width=0.4\linewidth,page=123]{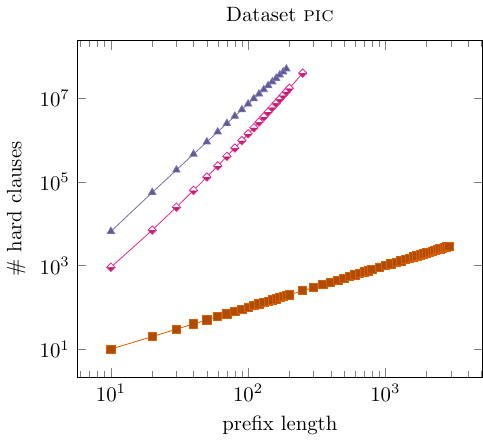}
	}\begin{minipage}{0.65\linewidth}
	\caption{Plots for dataset \textsc{pic}.}
	\label{plot:pic}
    \end{minipage}
    \hfill
    \begin{minipage}{0.25\linewidth}
        \includegraphics[page=1]{plot/plot}
    \end{minipage}
\end{figure}

\begin{figure}
	\centering{\includegraphics[width=0.4\linewidth,page=124]{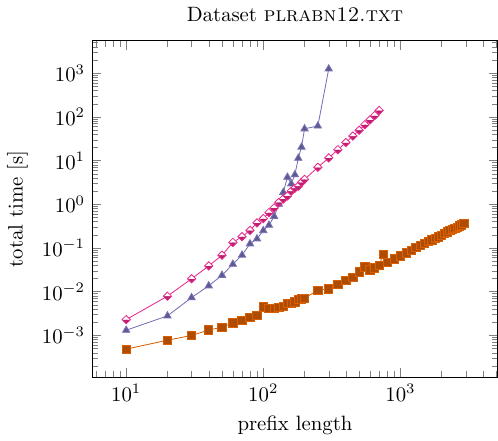}
\includegraphics[width=0.4\linewidth,page=125]{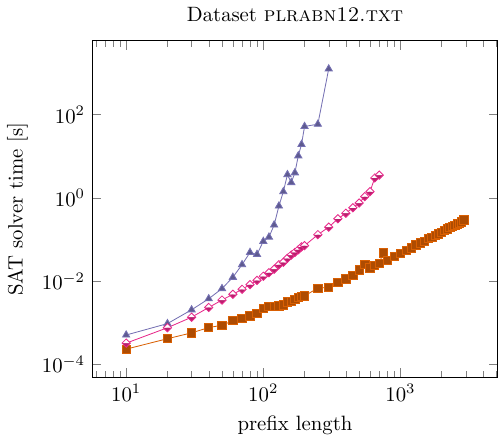}
\includegraphics[width=0.4\linewidth,page=126]{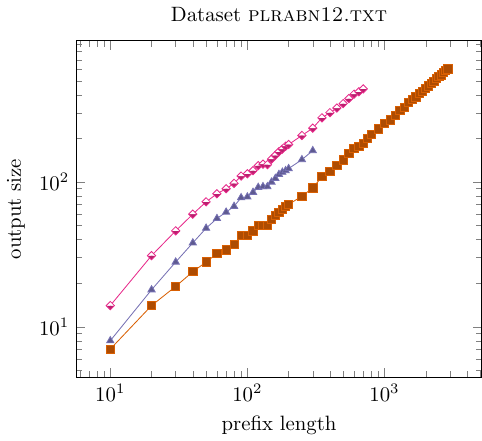}
\includegraphics[width=0.4\linewidth,page=127]{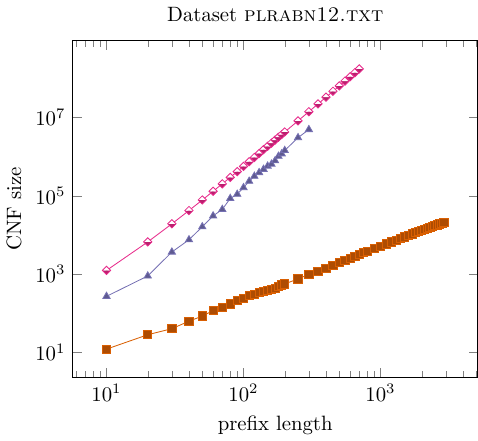}
\includegraphics[width=0.4\linewidth,page=128]{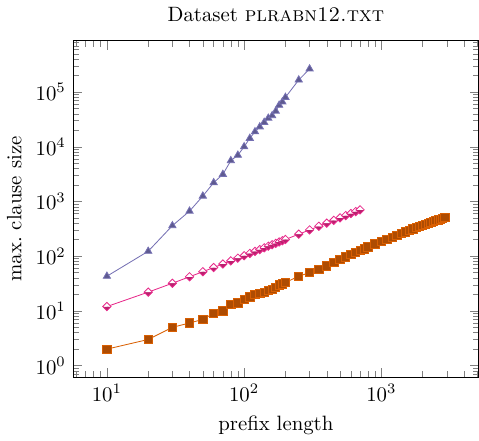}
\includegraphics[width=0.4\linewidth,page=129]{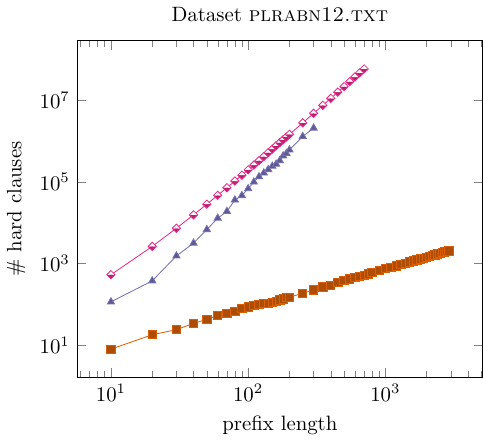}
	}\begin{minipage}{0.65\linewidth}
	\caption{Plots for dataset \textsc{plrabn12.txt}.}
	\label{plot:plrabn12.txt}
    \end{minipage}
    \hfill
    \begin{minipage}{0.25\linewidth}
        \includegraphics[page=1]{plot/plot}
    \end{minipage}
\end{figure}

\begin{figure}
	\centering{\includegraphics[width=0.4\linewidth,page=130]{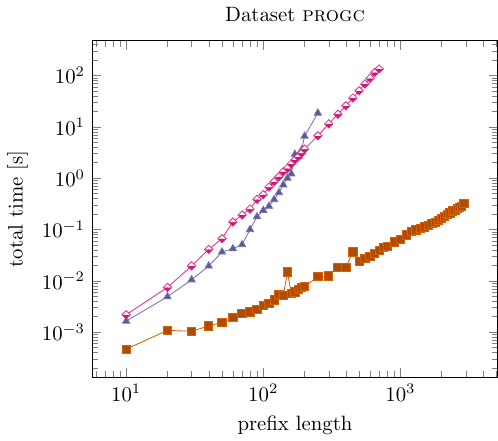}
\includegraphics[width=0.4\linewidth,page=131]{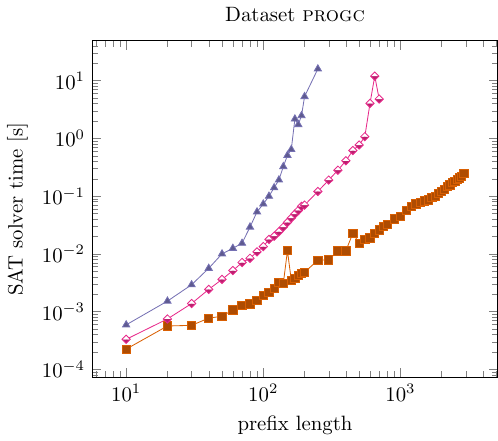}
\includegraphics[width=0.4\linewidth,page=132]{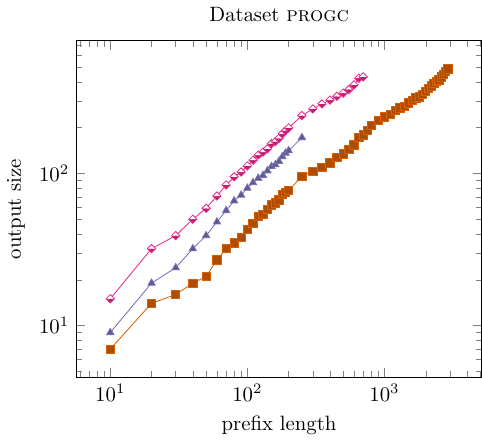}
\includegraphics[width=0.4\linewidth,page=133]{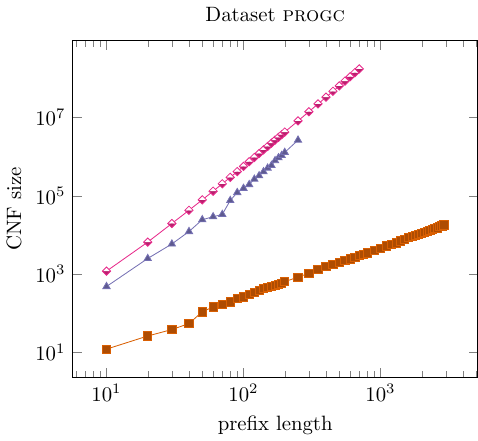}
\includegraphics[width=0.4\linewidth,page=134]{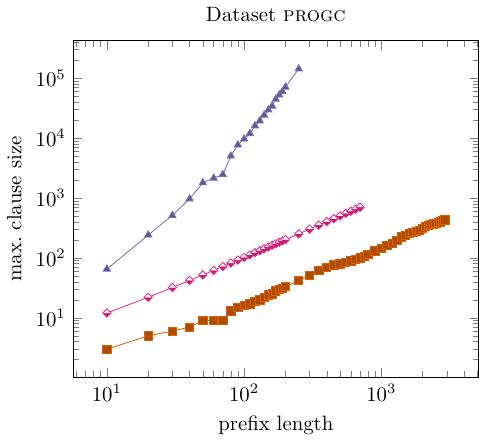}
\includegraphics[width=0.4\linewidth,page=135]{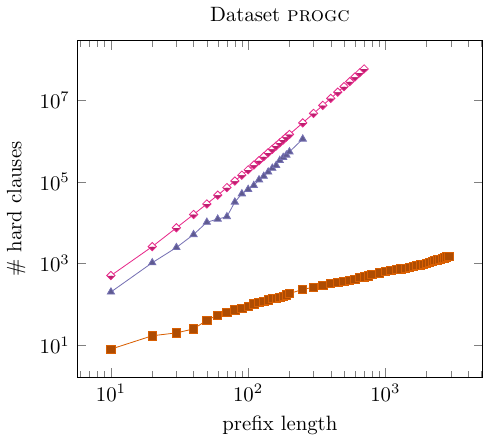}
	}\begin{minipage}{0.65\linewidth}
	\caption{Plots for dataset \textsc{progc}.}
	\label{plot:progc}
    \end{minipage}
    \hfill
    \begin{minipage}{0.25\linewidth}
        \includegraphics[page=1]{plot/plot}
    \end{minipage}
\end{figure}

\begin{figure}
	\centering{\includegraphics[width=0.4\linewidth,page=136]{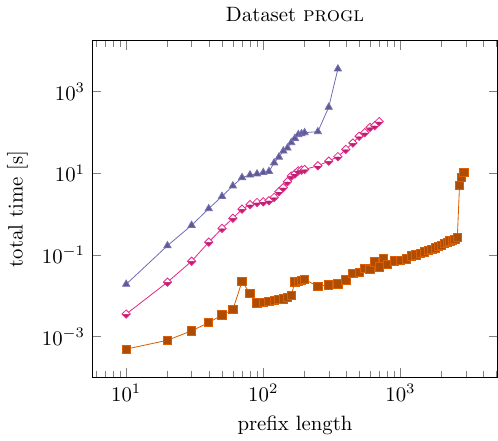}
\includegraphics[width=0.4\linewidth,page=137]{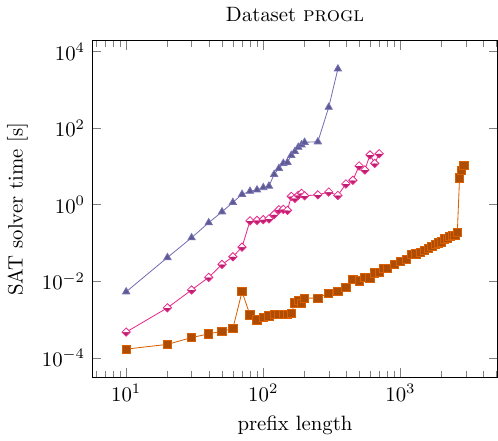}
\includegraphics[width=0.4\linewidth,page=138]{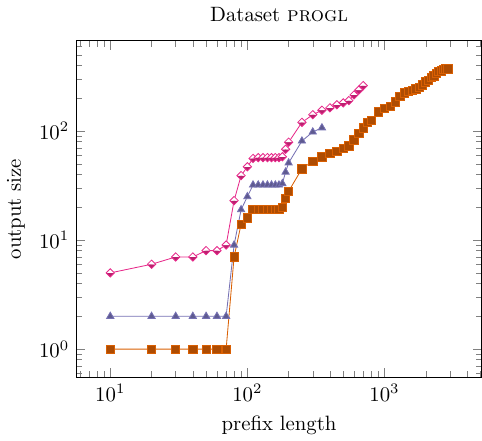}
\includegraphics[width=0.4\linewidth,page=139]{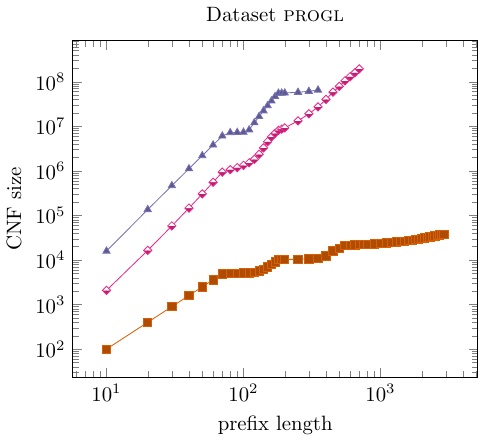}
\includegraphics[width=0.4\linewidth,page=140]{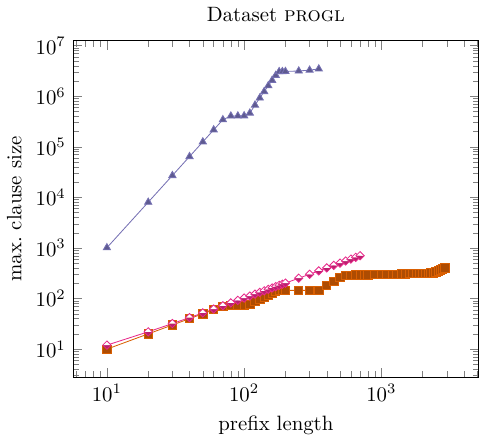}
\includegraphics[width=0.4\linewidth,page=141]{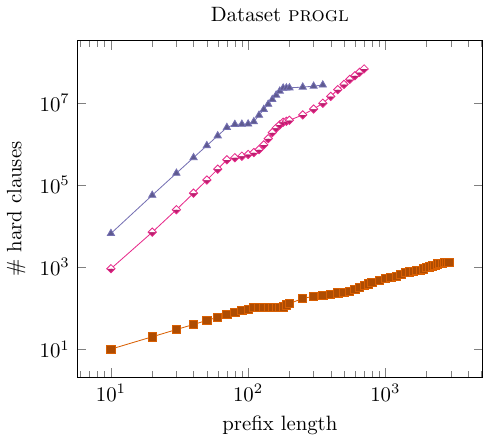}
	}\begin{minipage}{0.65\linewidth}
	\caption{Plots for dataset \textsc{progl}.}
	\label{plot:progl}
    \end{minipage}
    \hfill
    \begin{minipage}{0.25\linewidth}
        \includegraphics[page=1]{plot/plot}
    \end{minipage}
\end{figure}

\begin{figure}
	\centering{\includegraphics[width=0.4\linewidth,page=142]{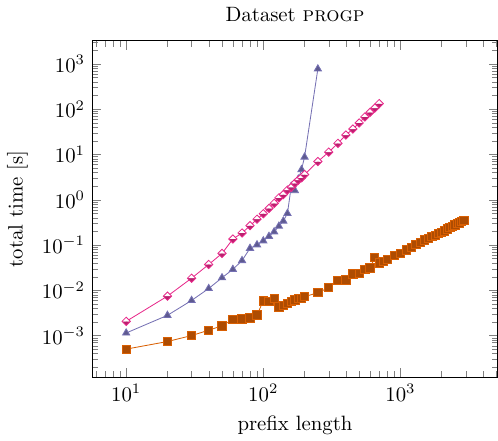}
\includegraphics[width=0.4\linewidth,page=143]{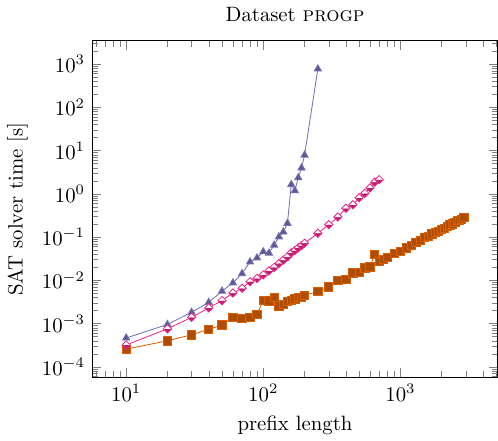}
\includegraphics[width=0.4\linewidth,page=144]{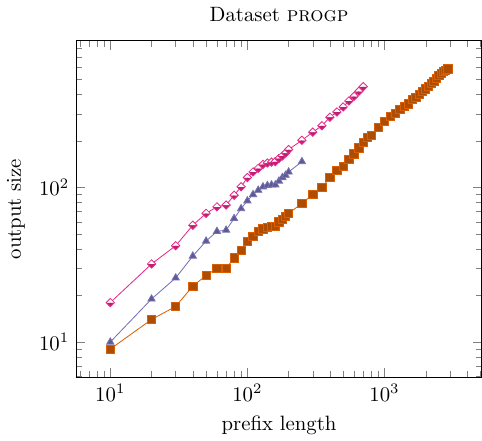}
\includegraphics[width=0.4\linewidth,page=145]{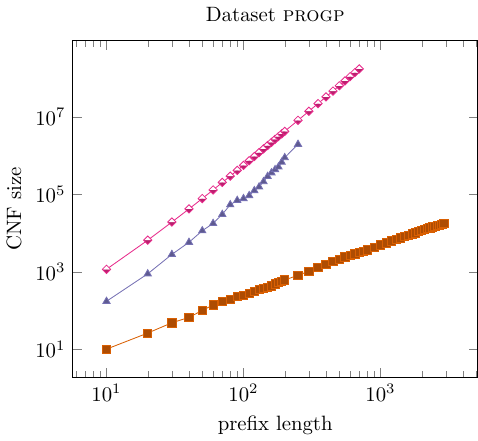}
\includegraphics[width=0.4\linewidth,page=146]{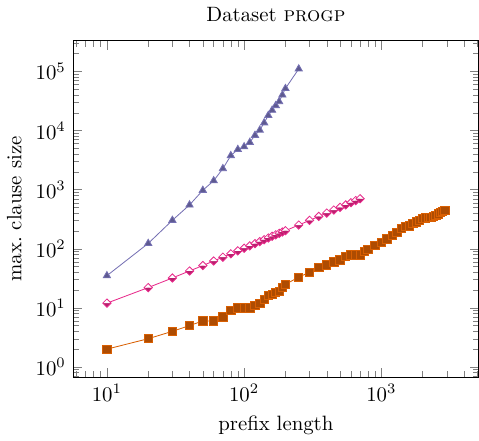}
\includegraphics[width=0.4\linewidth,page=147]{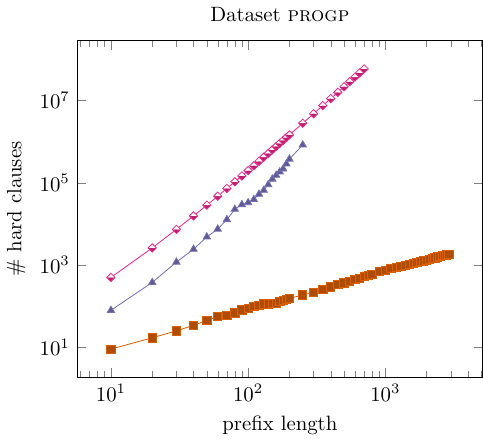}
	}\begin{minipage}{0.65\linewidth}
	\caption{Plots for dataset \textsc{progp}.}
	\label{plot:progp}
    \end{minipage}
    \hfill
    \begin{minipage}{0.25\linewidth}
        \includegraphics[page=1]{plot/plot}
    \end{minipage}
\end{figure}

\begin{figure}
	\centering{\includegraphics[width=0.4\linewidth,page=148]{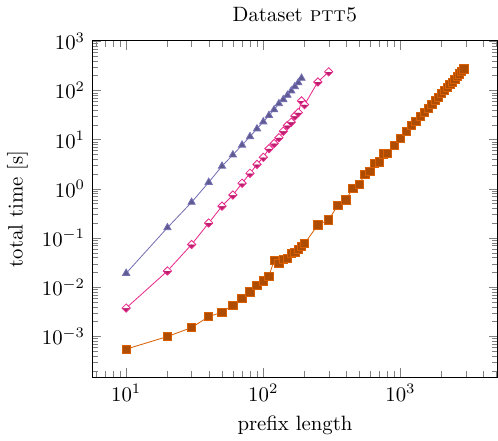}
\includegraphics[width=0.4\linewidth,page=149]{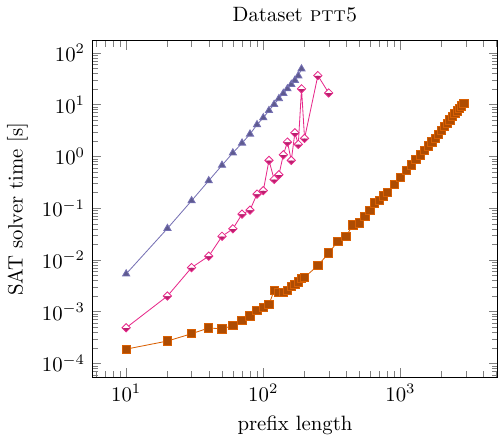}
\includegraphics[width=0.4\linewidth,page=150]{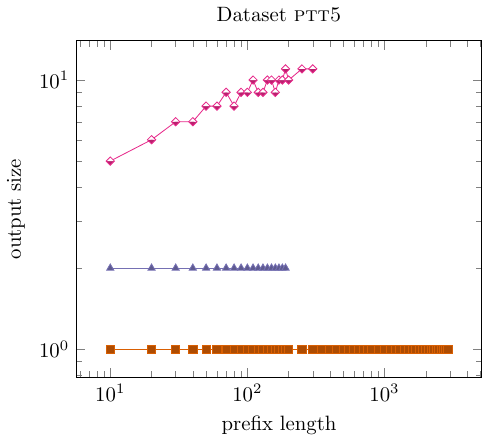}
\includegraphics[width=0.4\linewidth,page=151]{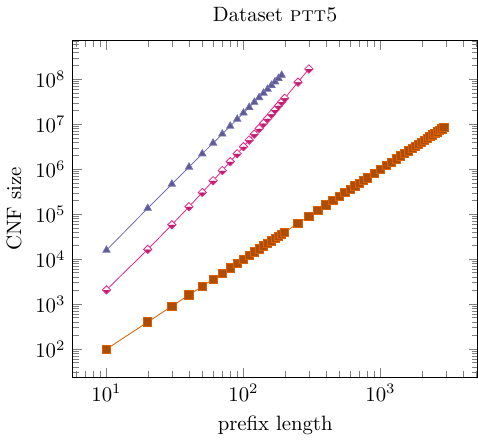}
\includegraphics[width=0.4\linewidth,page=152]{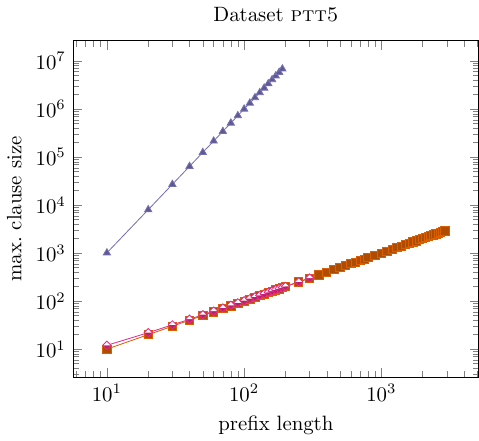}
\includegraphics[width=0.4\linewidth,page=153]{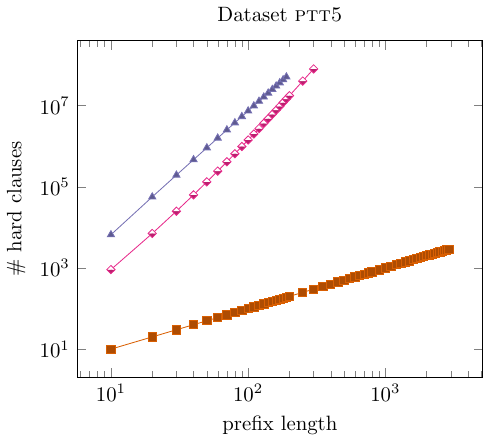}
	}\begin{minipage}{0.65\linewidth}
	\caption{Plots for dataset \textsc{ptt5}.}
	\label{plot:ptt5}
    \end{minipage}
    \hfill
    \begin{minipage}{0.25\linewidth}
        \includegraphics[page=1]{plot/plot}
    \end{minipage}
\end{figure}

\begin{figure}
	\centering{\includegraphics[width=0.4\linewidth,page=154]{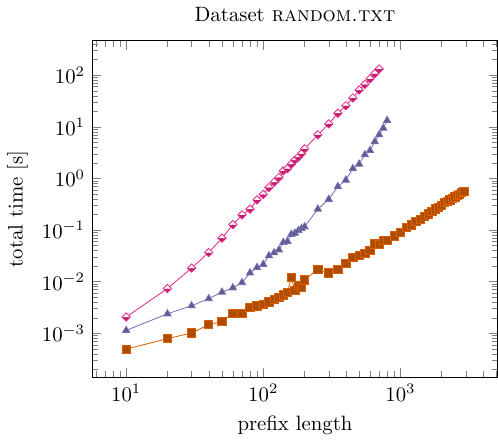}
\includegraphics[width=0.4\linewidth,page=155]{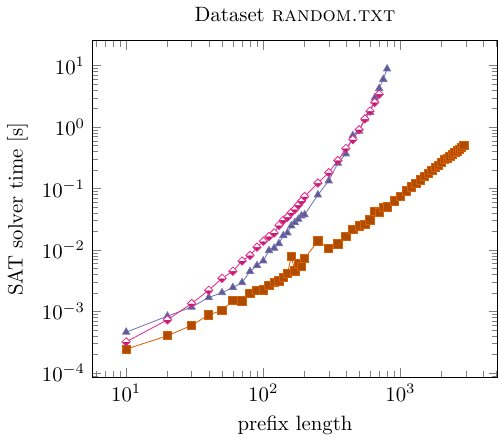}
\includegraphics[width=0.4\linewidth,page=156]{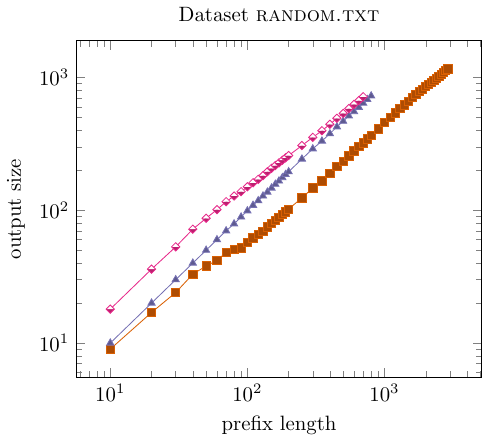}
\includegraphics[width=0.4\linewidth,page=157]{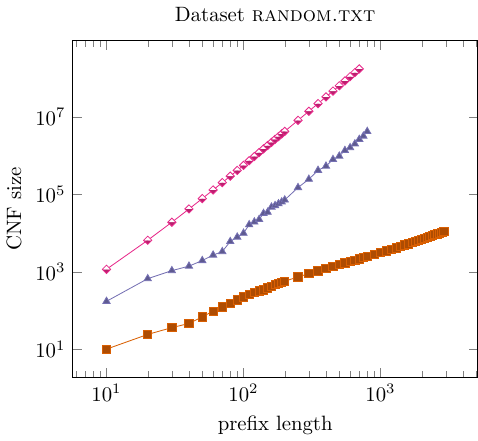}
\includegraphics[width=0.4\linewidth,page=158]{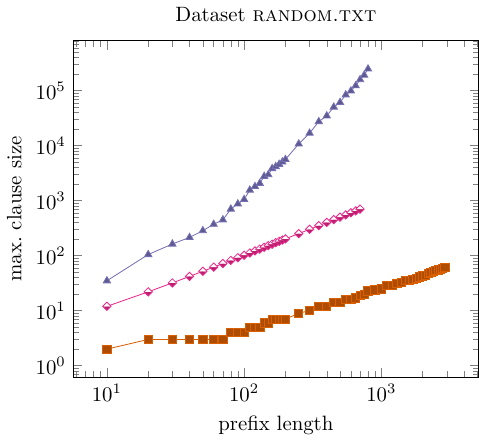}
\includegraphics[width=0.4\linewidth,page=159]{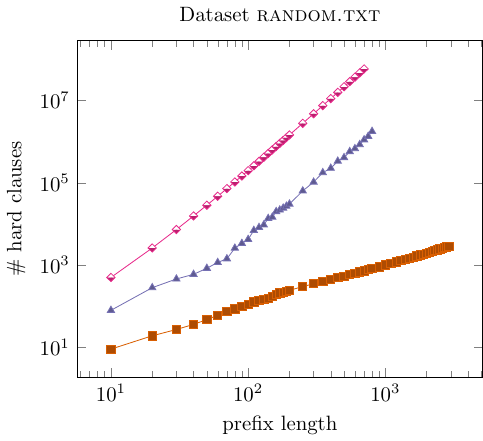}
	}\begin{minipage}{0.65\linewidth}
	\caption{Plots for dataset \textsc{random.txt}.}
	\label{plot:random.txt}
    \end{minipage}
    \hfill
    \begin{minipage}{0.25\linewidth}
        \includegraphics[page=1]{plot/plot}
    \end{minipage}
\end{figure}

\begin{figure}
	\centering{\includegraphics[width=0.4\linewidth,page=160]{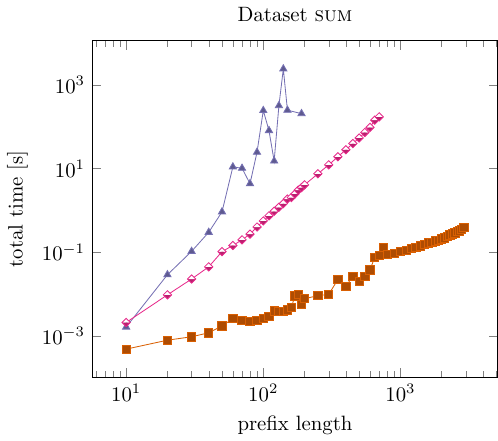}
\includegraphics[width=0.4\linewidth,page=161]{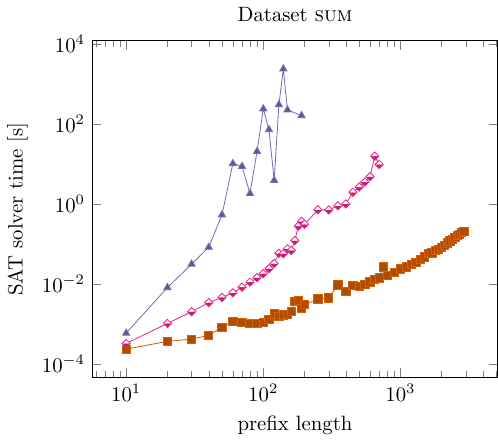}
\includegraphics[width=0.4\linewidth,page=162]{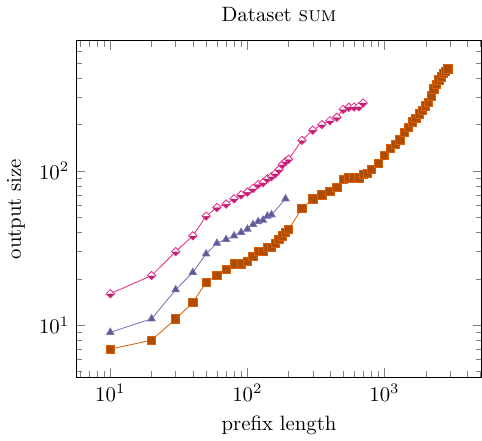}
\includegraphics[width=0.4\linewidth,page=163]{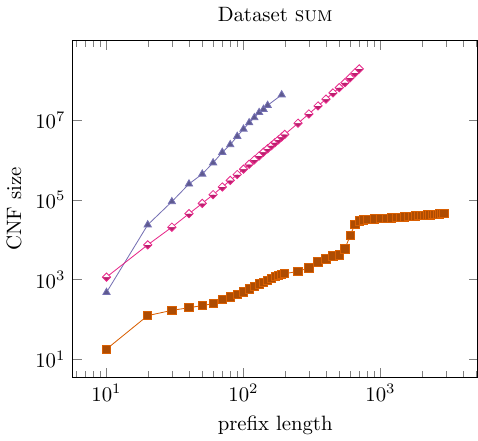}
\includegraphics[width=0.4\linewidth,page=164]{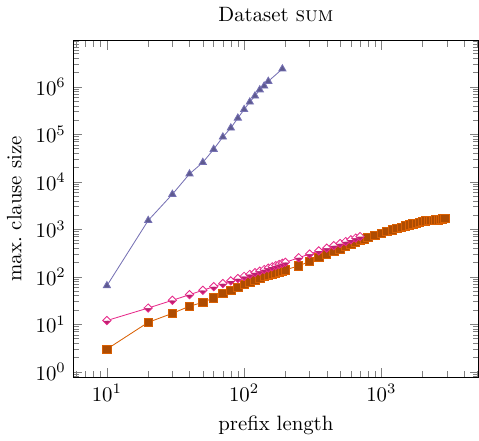}
\includegraphics[width=0.4\linewidth,page=165]{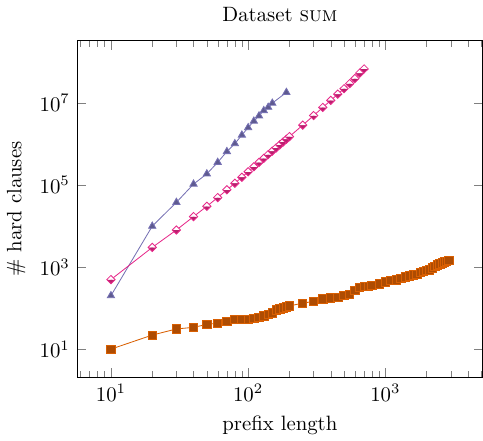}
	}\begin{minipage}{0.65\linewidth}
	\caption{Plots for dataset \textsc{sum}.}
	\label{plot:sum}
    \end{minipage}
    \hfill
    \begin{minipage}{0.25\linewidth}
        \includegraphics[page=1]{plot/plot}
    \end{minipage}
\end{figure}

\begin{figure}
	\centering{\includegraphics[width=0.4\linewidth,page=166]{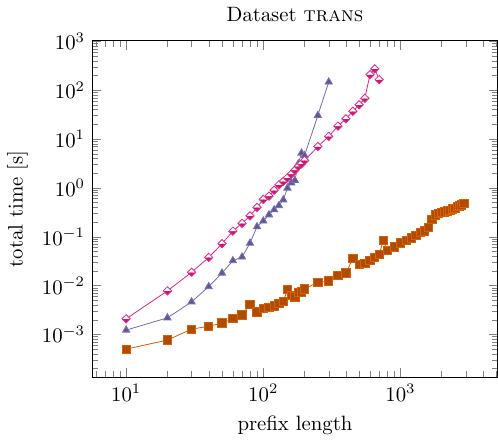}
\includegraphics[width=0.4\linewidth,page=167]{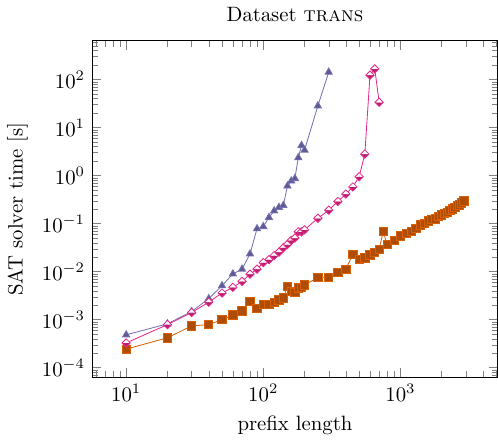}
\includegraphics[width=0.4\linewidth,page=168]{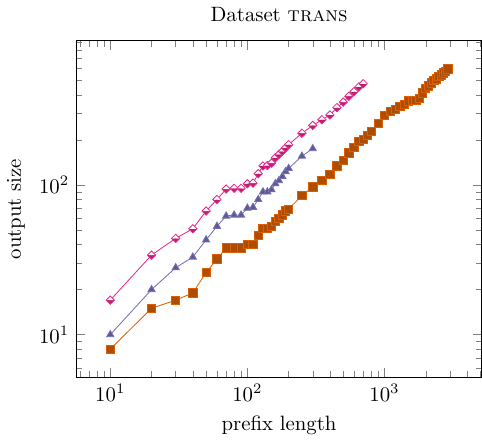}
\includegraphics[width=0.4\linewidth,page=169]{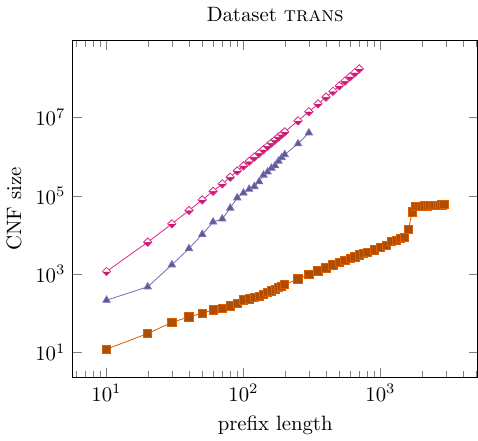}
\includegraphics[width=0.4\linewidth,page=170]{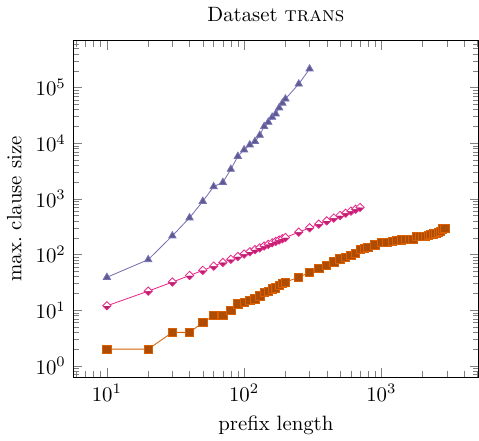}
\includegraphics[width=0.4\linewidth,page=171]{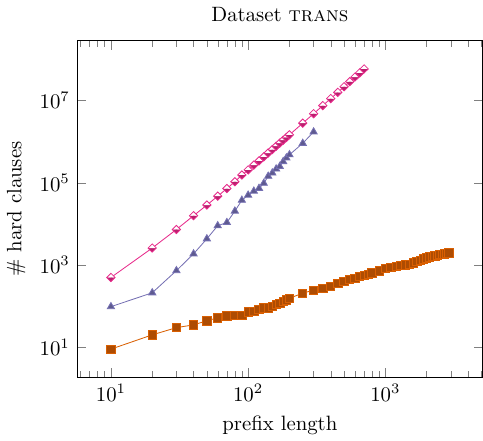}
	}\begin{minipage}{0.65\linewidth}
	\caption{Plots for dataset \textsc{trans}.}
	\label{plot:trans}
    \end{minipage}
    \hfill
    \begin{minipage}{0.25\linewidth}
        \includegraphics[page=1]{plot/plot}
    \end{minipage}
\end{figure}

\begin{figure}
	\centering{\includegraphics[width=0.4\linewidth,page=172]{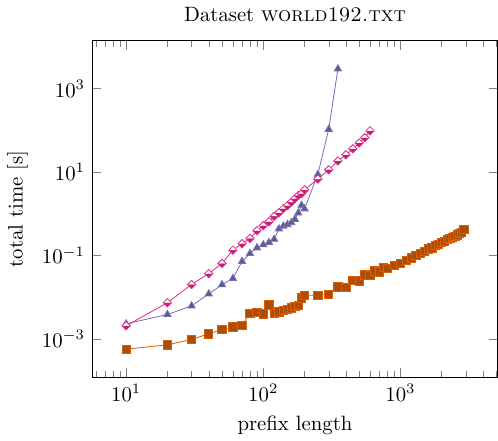}
\includegraphics[width=0.4\linewidth,page=173]{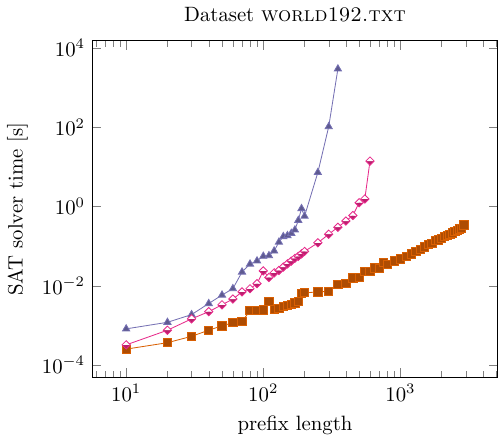}
\includegraphics[width=0.4\linewidth,page=174]{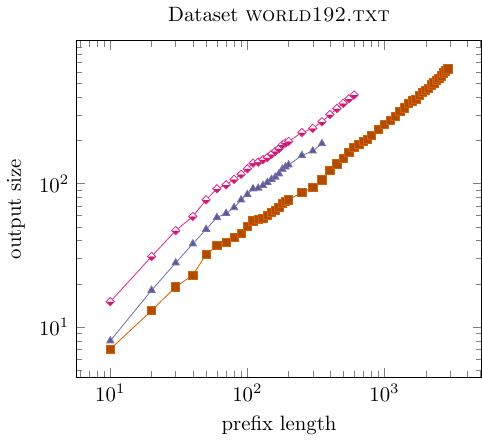}
\includegraphics[width=0.4\linewidth,page=175]{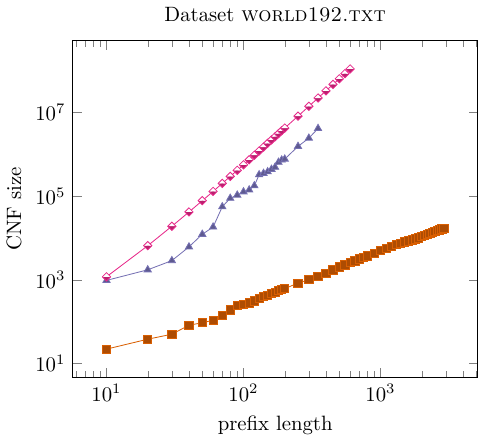}
\includegraphics[width=0.4\linewidth,page=176]{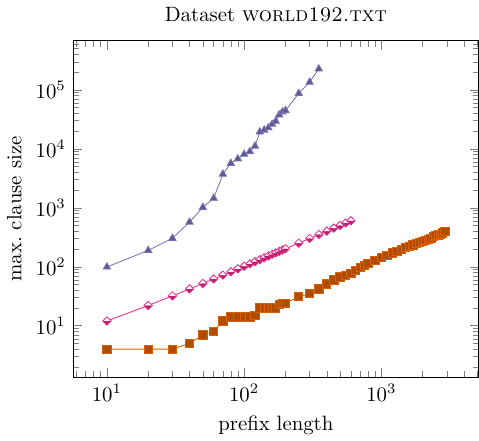}
\includegraphics[width=0.4\linewidth,page=177]{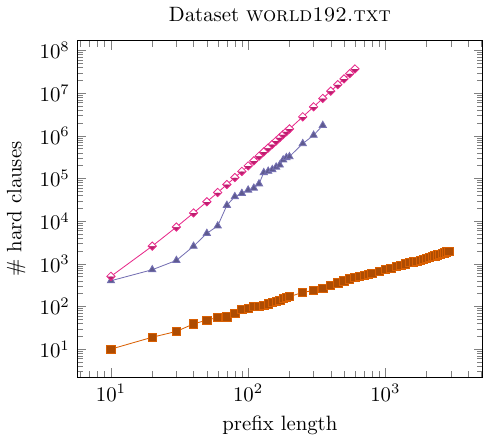}
	}\begin{minipage}{0.65\linewidth}
	\caption{Plots for dataset \textsc{world192.txt}.}
	\label{plot:world192.txt}
    \end{minipage}
    \hfill
    \begin{minipage}{0.25\linewidth}
        \includegraphics[page=1]{plot/plot}
    \end{minipage}
\end{figure}

\begin{figure}
	\centering{\includegraphics[width=0.4\linewidth,page=178]{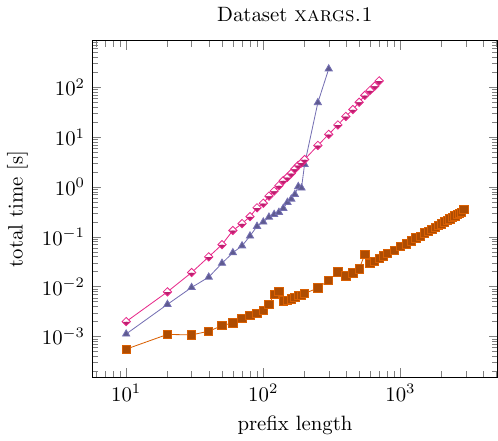}
\includegraphics[width=0.4\linewidth,page=179]{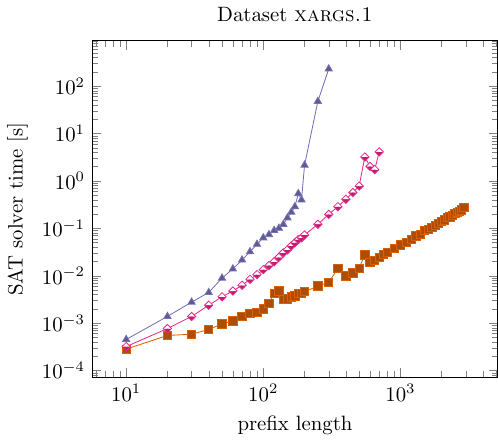}
\includegraphics[width=0.4\linewidth,page=180]{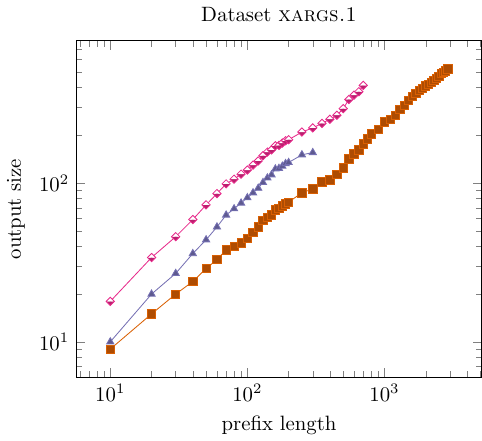}
\includegraphics[width=0.4\linewidth,page=181]{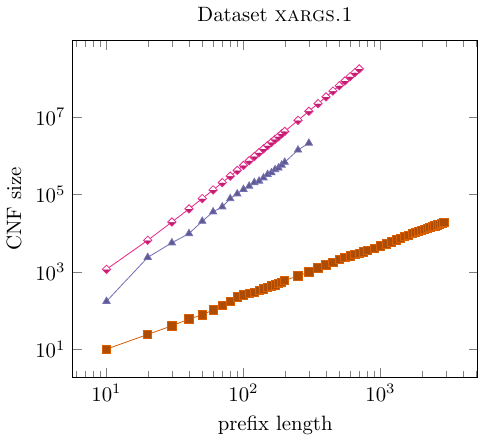}
\includegraphics[width=0.4\linewidth,page=182]{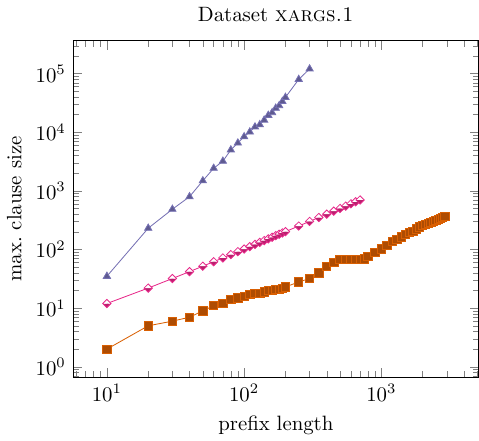}
\includegraphics[width=0.4\linewidth,page=183]{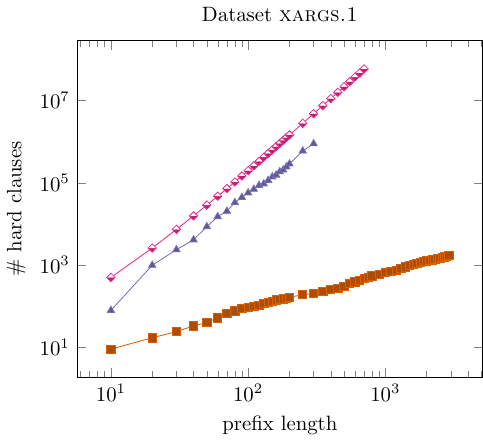}
	}\begin{minipage}{0.65\linewidth}
	\caption{Plots for dataset \textsc{xargs.1}.}
	\label{plot:xargs.1}
    \end{minipage}
    \hfill
    \begin{minipage}{0.25\linewidth}
        \includegraphics[page=1]{plot/plot}
    \end{minipage}
\end{figure}

\end{document}